\documentclass{article}
\usepackage{authblk}
\usepackage[margin=1in]{geometry}
\usepackage[onehalfspacing]{setspace} %set line space
\usepackage{amsmath,amssymb,bm,dsfont,mathtools,amsthm}
\usepackage{graphicx,graphics}
\usepackage{fancyhdr,indentfirst,epsf,multicol,listings}
\usepackage[shortlabels, inline]{enumitem}
\usepackage[cal=boondox]{mathalfa} 
\usepackage{algorithm,algorithmic}
\usepackage{natbib}
\usepackage[hidelinks]{hyperref} 
\usepackage{array,threeparttable,blindtext} % table note
\usepackage{longtable,caption} % long table

\newtheorem{theorem}{Theorem}
\newtheorem{corollary}{Corollary}
\newtheorem{proposition}{Proposition}
\newtheorem{lemma}{Lemma}
\newtheorem{definition}{Definition}

\newcommand{\diag}[1]{\mathop{\mathrm{diag}}\left(#1\right)}

\newcommand{\1}{\mathds{1}}
\newcommand{\E}{\mathbb{E}}
\newcommand{\R}{\mathbb{R}}
\newcommand{\bP}{\mathbb{P}}
\newcommand{\T}{^\top}

\newcommand{\p}{\mathcal{p}}
\newcommand{\cd}{\mathcal{d}}
\newcommand{\cl}{\mathcal{l}}
\newcommand{\cP}{\mathcal{P}}
\newcommand{\Q}{\mathcal{Q}}
\newcommand{\cO}{\mathcal{O}}
\newcommand{\cR}{\mathcal{R}}
\newcommand{\D}{\mathcal{D}}
\newcommand{\A}{\mathcal{A}}
\newcommand{\cE}{\mathcal{E}}

\newcommand{\B}{\normalfont\textbf{B}}
\newcommand{\W}{\normalfont\textbf{W}}
\newcommand{\tp}{\normalfont\textbf{p}}
\newcommand{\tP}{\normalfont\textbf{P}}
\newcommand{\td}{\normalfont\textbf{d}}
\newcommand{\tD}{\normalfont\textbf{D}}

\newcommand{\Din}{D^\text{\normalfont in}}
\newcommand{\Dout}{D^\text{\normalfont out}}
\newcommand{\din}{d^\text{\normalfont in}}
\newcommand{\dout}{d^\text{\normalfont out}}
\newcommand{\thi}{\theta^\text{\normalfont in}}
\newcommand{\tho}{\theta^\text{\normalfont out}}
\newcommand{\Ti}{\Theta^\text{\normalfont in}}
\newcommand{\To}{\Theta^\text{\normalfont out}}
\newcommand{\Di}{\mathcal{D}^\text{\normalfont in}}
\newcommand{\Do}{\mathcal{D}^\text{\normalfont out}}
\newcommand{\di}{\mathcal{d}^\text{\normalfont in}}
\newcommand{\dou}{\mathcal{d}^\text{\normalfont out}}

\newcommand{\tDi}{\normalfont\textbf{D}^\text{\normalfont in}}
\newcommand{\tDo}{\normalfont\textbf{D}^\text{\normalfont out}}
\newcommand{\tdi}{\normalfont\textbf{d}^\text{\normalfont in}}
\newcommand{\tdo}{\normalfont\textbf{d}^\text{\normalfont out}}

\newcommand{\Qw}{Q^{(w)}}
\newcommand{\Pw}{P^{(w)}}
\newcommand{\pw}{p^{(w)}}
\newcommand{\QwG}{Q^{(w,G)}}
\newcommand{\PwG}{P^{(w,G)}}

\newcommand{\logN}{\log{N}}

\newcolumntype{H}{>{\setbox0=\hbox\bgroup}c<{\egroup}@{}}

\newcommand{\underdcsbm}{Under the population DC-SBM with $K$ blocks and parameters $\left\{\B, Z, \Theta\right\}$, }
\newcommand{\underddcsbm}{Under the population directed DC-SBM with $K$ blocks and parameters $\left\{\B, Z, \Ti, \To\right\}$, }
\newcommand{\dcsbmset}{Let $G=(V,E)$ be a graph of $N$ vertices generated from the DC-SBM with $K$ blocks and parameters $\left\{\B, Z, \Theta\right\}$. }

\newcommand{\beginsupplement}{%
	\setcounter{section}{0}
	\renewcommand{\thesection}{S\arabic{section}}
	\setcounter{table}{0}
	\renewcommand{\thetable}{S\arabic{table}}
	\setcounter{figure}{0}
	\renewcommand{\thefigure}{S\arabic{figure}}
	\setcounter{lemma}{0}
	\renewcommand{\thelemma}{S\arabic{lemma}}
	\setcounter{theorem}{0}
	\renewcommand{\thetheorem}{S\arabic{theorem}}
	\setcounter{proposition}{0}
	\renewcommand{\theproposition}{S\arabic{proposition}}
	
	\renewcommand{\figurename}{Supplementary Figure}
	\renewcommand{\tablename}{Supplementary Table}
}

\title{Targeted sampling from massive block model graphs with personalized PageRank\thanks{This research is supported in part by NSF Grants DMS-1612456 and DMS-1916378 and ARO Grant W911NF-15-1-0423.}}

\author[1]{Fan Chen}
\author[2]{Yini Zhang}
\author[1]{Karl Rohe}
\affil[1]{Department of Statistics}
\affil[2]{School of Journalism and Mass Communication}
\affil[ ]{University of Wisconsin, Madison, WI 53706, USA}

\date{}

\begin{document}

\maketitle

\begin{abstract}
		
This paper provides statistical theory and intuition for Personalized PageRank (PPR), a popular technique that samples a small community from a massive network. 
We study a setting where the entire network is expensive to thoroughly obtain or maintain, but we can start from a seed node of interest and ``crawl'' the network to find other nodes through their connections.
By crawling the graph in a designed way, the PPR vector can be approximated without querying the entire massive graph, making it an alternative to snowball sampling.
Using the degree-corrected stochastic block model, we study whether the PPR vector can select nodes that belong to the same block as the seed node.
We provide a simple and interpretable form for the PPR vector, highlighting its biases towards high degree nodes outside of the target block. We examine a simple adjustment based on node degrees and establish consistency results for PPR clustering that allows for directed graphs. 
These results are enabled by recent technical advances showing the element-wise convergence of eigenvectors. 
We illustrate the method with the massive Twitter friendship graph, which we crawl using the Twitter API. We find that (i) the adjusted and unadjusted PPR techniques are complementary approaches, where the adjustment makes the results particularly localized around the seed node and (ii) the bias adjustment greatly benefits from degree regularization.
\end{abstract}

	
\section{Introduction}

Much of the literature on graph sampling has treated the entire graph, or all of the people in it, as the target population. However, in many settings, the target population is a small community in the massive graph. 
For example, a key difficulty in studying social media is to gather data that is sufficiently relevant for the scientific objective.
A motivating example for this paper is to sample the Twitter friendship graph for accounts that report and discuss current political events.\footnote{See our website \url{http://murmuration.wisc.edu} which does this.} This corresponds to sampling and identifying multiple different communities, each a potentially small part of the massive network.
In such an application, the graph is useful for two primary reasons. First, via link tracing, we can find potential members of the target population. Second, the graph connections are informative for identifying community membership. Throughout, we presume that the sampling is initiated around a ``seed node'' that belongs to the target community of interest.  
%Randomized link-tracing from the seed node includes highly connected nodes outside the target community before moderately connected nodes within the target community. 

Personalized PageRank (PPR) can be thought of as an alternative to snowball sampling, a popular technique for gathering individuals close to the seed node.  For some $d\ge 0$, snowball sampling gathers all individuals who are $d$ friends away from the seed.  This process has two competing flaws for our application which are addressed by PPR.  First, snowball sampling fails to account for the density of common friendships. For example, perhaps $i$ and $j$ are both one friend removed from the seed, but $i$ has 10 friends in common with the seed, while $j$ only has 1 friend in common.  It seems natural to suppose that $i$ is closer than $j$ to the seed.  Hence, the metric for snowball sampling can be misleading. Second, the snowball sample size grows very quickly with $d$. For example, under the ``six degrees of separation'' phenomenon \citep{watts1998collective, mark2006structure}, snowballing gathers the entire graph if $d\ge 6$.  

PPR gives a sample that is more localized around the seed node. The PPR vector is defined as the stationary distribution of what we call a \textit{personalized random walk} \citep{page1999pagerank}. 
%	PPR gives the inclusion probabilities for the following sampling procedure, what we call a \textit{personalized random walk} \citep{page1999pagerank}. 
At each step of the personalized random walk, the random walker returns to the seed node with probability $\alpha$, called the teleportation constant, and with probability $1-\alpha$, the random walker goes to  an adjacent node that is chosen uniformly at random. 
Consider the stationary distribution of this process as giving the inclusion probability for a sample of size 1. This is the PPR vector. 
PPR naturally leads to a clustering algorithm, where the cluster is made up of the nodes with a large inclusion probability. 
To quickly approximate the PPR vector, \citet{berkhin2006bookmark} proposed an algorithm that only examines nodes with large inclusion probabilities (i.e. nodes near the seed).
As such, PPR is particularly useful for its computational efficiency -- the running time and the amount of data it requires is nearly linear in the size of the output cluster, which is typically much smaller than that of the entire graph. 
Due to the local nature of the algorithm, it can be used to study large graphs such as Twitter where the entire graph is not available, but where one can query to find the connections to any small set of nodes.

One way to conduct local clustering is by exploring and ranking the nearby nodes of a seed node. \citep{andersen2006communities, andersen2009finding, alamgir2010multi, gharan2012approximating}. 
\citet{spielman2004nearly} pioneered local clustering  by defining nearness as the landing probability of a random walk starting from the seed node.
Their algorithm's guarantee was improved in follow-up work by \citet{andersen2006local} which proposed using an approximate PPR vector. 
Local algorithms can be applied recursively to solving more complicated problems such as graph partitions (k-way partitions) \citep{spielman1996spectral, karypis1998multilevelk}, and has many fruitful applications
%in social, information, and biochemical studies, among others 
\citep{jeh2003scaling, macropol2009rrw, liao2009isorankn, gupta2013wtf, gleich2015pagerank}, particularly when it comes to sampling and studying massive graphs. 
%\cite{andersen2006local} also proposed a way to simultaneously approximate the PPR vector and sample the network. \citep{kloster2014heat}

Along with the widespread use of PPR, there has been recent work to study its statistical estimation properties under a statistical model with latent community structure.
%To fully understand the behavior of PPR methods and to identify the conditions under which they reveal the latent community structure, this paper studies the statistical estimation properties of PPR under a statistical model with latent community structure. 
Beyond the scope of local clustering, \citet{kloumann2017block} showed that the PPR vector is
%	 itself was 
asymptotically equivalent to optimal linear discriminant analysis under the stochastic block model (SBM) \citep{holland1983stochastic}, assuming a symmetry condition on the block structure. 
We add to this statistical understanding of PPR by providing a simple and more general representation for PPR vectors that allows for different block sizes, more than two blocks, degree heterogeneity, and directed edges.  
In order to understand the effects of heterogeneous node degrees, this paper uses the degree-corrected stochastic block model (DC-SBM) \citep{karrer2011stochastic} and examines when the PPR clustering recovers nodes within the same block as the seed node (local cluster).
Breaking the symmetry that is imposed by  \citet{kloumann2017block} reveals additional insight. 
In particular, given a seed node in the first block, we show that PPR is likely to contain high degree nodes outside of that block.
We study an adjustment that was previously proposed in \cite{andersen2006local}.  We show how this adjustment can correct for the bias.  We illustrate these ideas with examples from the Twitter friendship graph.  
%	These results provide an alternative understanding for PPR

%Meanwhile, an active line of network sampling research dates back to \citet{leskovec2006sampling}, a thorough study on a number of sampling methods in empirical networks. 
%They found that the PPR based sampling (called Random Walk sampling in the paper) performed best overall for the Scale-down sampling goal, which is to find a sample graph that is most similar to the original network. 
%PPR based sampling also applies to streaming or initially unknown graphs, with applications in numerous areas \citep{sarma2011estimating, salehi2012sampling, zhang2016approximate, lofgren2016personalized, leskovec2016snap}.

\subsection{An illustrative example in social media}

Local clustering using PPR is particularly well suited to studying current political events on Twitter because (i) the accounts that discuss politics or current events are a small part of the entire Twitter graph, (ii) it is reasonable to believe that the accounts in our target population are well connected to one another in the Twitter friendship graph, and (iii) while the entire Twitter graph is not publicly available, the way that PPR (Algorithm \ref{alg:ppr} and \ref{alg:ppr_dir}) queries the graph matches the Twitter API protocol which is the primary mode of access for researchers.

While we do not suppose that the Twitter friendship graph is sampled from a DC-SBM, Twitter does have all of the heterogeneities that our results identify as important. 
The Twitter friendship graph is composed of users who can freely follow others but will not necessarily be followed back, or friended. Such asymmetry between following and friending forms a directed graph where follower count indicates status -- some popular/high-status nodes command millions of followers while the majority of nodes are followed by far fewer. 

\begin{table} 
	\caption{\label{tab:egppr} Top 15 handles by PPR clustering. Column names represent seed nodes, and the sampled nodes are ranked by PPR values, with teleportation constant $\alpha=0.15$ uniformly.} 
	\centering
	\footnotesize
	\fbox{%
		\begin{threeparttable}
			\begin{tabular}{llll}
				%				\hline
				\noalign{\vskip .3mm}  
				& \textbf{@CNN} & \textbf{@BreitbartNews} & \textbf{@dailykos} \\ 
				\hline
				\noalign{\vskip .7mm}  
				1	&	CNN Breaking News	&	Alex Marlow	&	Hillary Clinton	\\
				2	&	CNN International	&	AndrewBreitbart	&	Stephen Colbert	\\
				3	&	Wolf Blitzer	&	Big Hollywood	&	Rachel Maddow MSNBC	\\
				4	&	Anderson Cooper	&	Big Government	&	Jake Tapper	\\
				5	&	Christiane Amanpour	&	James O'Keefe	&	Joy Reid	\\
				6	&	Pope Francis	&	Sean Hannity	&	Chris Hayes	\\
				7	&	Dr. Sanjay Gupta   &	Raheem	&	Emma Gonzlez	\\
				8	&	CNNMoney	&	Joel B. Pollak	&	Markos Moulitsas	\\
				9	&	Jake Tapper	&	Ann Coulter	&	Maggie Haberman	\\
				10	&	Brian Stelter	&	Allum Bokhari	&	Sarah Silverman	\\
				11	&	CNN Newsroom	&	Ben Kew	&	Lin-Manuel Miranda	\\
				12	&	Dana Bash	&	Brandon Darby	&	Elizabeth Warren	\\
				13	&	CNN Politics	&	Noah Dulis	&	Jon Favreau	\\
				14	&	BBC Breaking News	&	Michelle Malkin	&	Michelle Obama	\\
				15	&	Brooke Baldwin	&	Nate Church	&	Bill Clinton	\\
				\hline
			\end{tabular}
			\begin{tablenotes}[flushleft]
				\item 
				Through the PPR vector, the top 15 handles returned to each of the three seed nodes fit well with the characteristics of the seed nodes. They are popular/high-status handles either directly related to the seed nodes or align with their political leanings. This shows the effectiveness of clustering via the PPR vector. It also shows the PPR vector's preference for highly connected nodes.
				%				\vspace{1mm}
			\end{tablenotes}
		\end{threeparttable}
	}
\end{table}

\begin{table} 
	\caption{\label{tab:egrppr} Top 15 handles by adjusted PPR (with regularization) sampling. Column names represent seed nodes, and the sampled nodes are ranked by adjusted PPR values, with teleportation constant $\alpha=0.15$ uniformly.} 
	\centering
	\footnotesize
	\fbox{%
		\begin{threeparttable}
			\begin{tabular}{llll}
				%				\hline
				\noalign{\vskip .3mm}  
				& \textbf{@CNN} & \textbf{@BreitbartNews} & \textbf{@dailykos} \\ 
				\hline
				\noalign{\vskip .7mm}  
				1	&	PowerZ	&	Robert	&	Two Thanks	\\
				2	&	Elissa Weldon	&	Lee Peace	&	Catherine Daligga	\\
				3	&	Tess Eastment	&	Wynn Marlow	&	exmearden	\\
				4	&	Chris\_Dawson	&	Logan Churchwell	&	Faith Gardner	\\
				5	&	carol kinstle	&	Peter Schweizer	&	Andrew Thornton	\\
				6	&	erinmclaughlin	&	Breitbart Sports	&	UnreasonableFridays	\\
				7	&	Taylor Ward	&	Jon Fleischman	&	DKos Top Comments	\\
				8	&	Jennifer Z. Deaton	&	Nate Church	&	2016 relitigator	\\
				9	&	Pam Benson	&	Daniel Nussbaum	&	Daily Kos	\\
				10	&	amy entelis	&	Noah Dulis	&	Walter Einenkel	\\
				11	&	Grace Bohnhoff	&	Jon David Kahn	&	Candelaria Vargas	\\
				12	&	kate lazarus	&	Breitbart California	&	Mara Schechter	\\
				13	&	Newstron	&	Ken Klukowski	&	Emi Feldman	\\
				14	&	Becky Brittain	&	pam key	&	The Soulful Negress	\\
				15	&	CNN Ballot Bowl	&	Auntie Hollywood	&	Kim Soffen	\\
				\hline
			\end{tabular}
			\begin{tablenotes}[flushleft]
				\item 
				%Through the PPR vector with adjustment, the top 15 handles returned to each of the three seed nodes range from the outlet's journalists, editors as well as related writers/campaigners/activists. This shows the effectiveness of adjustment.
				After adjustment, PPR returns a more localized cluster. Instead of the highly visible public faces of the three seed organizations, the individuals in this table serve a central role to the internal organization (e.g. editors and writers).  Depending on the application, one might prefer the results in Table \ref{tab:egppr} or Table \ref{tab:egrppr}.  
			\end{tablenotes}
%			\vspace{1mm}
		\end{threeparttable}
	}
\end{table}

The theoretical results in this paper suggest that such degree heterogeneities will make the PPR vector biased for detecting block memberships (Theorem \ref{explicit2}).  We propose a way to adjust for this bias (Algorithm \ref{alg:lc}) and show that it is a consistent estimator (Corollary  \ref{cor:main}).  Not surprisingly, this section demonstrates that PPR with and without the bias adjustment give fundamentally different results on the Twitter graph.  However, depending on the application, the biases in the PPR vector might be advantageous.  In this way, PPR with and without the bias adjustment are complementary, not competing, approaches. 

%	However, as this section shows, are complementary approaches.  They both obtain reasonable results that are 
%	not competing approaches in which one is uniformly better.  	
%	suggests the uniform advantages of the PPR adjustment, the application to Twitter shows that PPR with versus without adjustment are not competing approaches, but rather complementary approaches. 
To illustrate, Table \ref{tab:egppr} displays the top 15 handles ranked by the PPR vector (without adjustment) for three different seed nodes: @CNN, @BreitbartNews, and @dailykos, the Twitter accounts of three different types of media outlets that exhibit distinct political leanings (legacy broadcast news, online right-wing and online left-wing). 
For @CNN, all top 15 handles ranked by the PPR vector are its subsidiary accounts 
%	(like CNN Breaking News and CNN International) 
and its celebrity reporters and anchors (like Wolf Blitzer and Anderson Cooper), except for one account, Pope Francis, who enjoys an extremely larger following. The top 15 handles for @BreitbartNews are a mixed bag of influential conservatives (like Sean Hannity and Ann Coulter) and Breitbart's editors/writers. However, the top 15 handles returned to @dailykos by the PPR vector are all famous liberal personalities not directly affiliated with Daily Kos, but one, its founder Markos Moulitsas. Those people range from democratic politicians to liberal media personalities and journalists, such as Hillary Clinton, Stephen Colbert, and Rachel Maddow. All the handles align with the characteristics of their respective media outlets, attesting to the clustering effectiveness. However, it is worth noting that the top handles ranked by the PPR vector tend to be popular handles with millions of followers. This shows that the PPR vector's preference for high in-degree nodes. 

In contrast, for each of the three seeds, adjusted PPR finds accounts that are more central to the internal functioning of these organizations. 
Table \ref{tab:egrppr} lists those accounts. 
The bias adjustment also greatly benefits from a degree regularization \citep{qin2013regularized}.
For @CNN, those handles include primarily its own staff/producers/journalists (like Elissa Weldon, Chris\_Dawson, and Grace Bohnhoff), a freelance journalist (Tess Eastment). The pattern is similar for @BreitbartNews and @dailykos, their top 15 handles including their own journalists, editors as well as related writers/campaigners/activists. The general pattern is that the adjustment returns editors, journalists and staff working within each media outlet. As such, the adjustment is useful for identifying a more localized cluster.

%	a more local and exclusively relevant cluster, i.e. 

\subsection{Main contributions}
The main contributions of the paper are (a) a simple and interpretable form for the PPR vector and (b) a statistical guarantee for clustering with the adjusted PPR vector.

\begin{enumerate}[(a)]
	\item This paper reveals a simple two-stage form of the PPR vector under the population (expectation) DC-SBM. 
	Consider the $v$-th element of the PPR vector as the probability of sampling node $v$ in a sample of size 1 from the stationary distribution of the personalized random walk. This inclusion probability is akin to stratified sampling:
	\begin{quote}
		\em The inclusion probability for node $v$ is the product of two separate probabilities. First, the probability that the personalized random walk samples any node in $v$'s block.  Second, the probability that the personalized random walk selects node $v$, conditional on sampling that block.
	\end{quote}
	Both of these probabilities have simple expressions. If there are $K$ blocks in the graph, then the block-wise probability comes from the PPR vector of a graph with $K$ vertices, with edge weights specified by the ``block connection matrix'' in the DC-SBM. The second probability is proportional to the degree of node $v$.
	In addition to the population results, Theorem \ref{thm:concPPR} demonstrates that when the graph is random, the PPR vector concentrates around its population (expectation) under certain conditions.
	%	Based on these observations, the paper also extends the previous results of asymptotic equivalence between PPR and linear discriminant analysis under some particular DC-SBM. 

	\item This paper identifies two sources of bias of using a PPR vector for local clustering under the DC-SBM -- the ancillary effects of heterogeneous node degrees and block degrees.  
	With this finding, the paper examines a simple bias adjustment that remedies the two biases simultaneously and suggests conditions when the adjusted PPR can be used to return the correct local cluster. In other words,
	\begin{quote}
		\em PPR clustering with the adjustment achieves the precise identification of the local cluster, provided the graph is sufficiently dense.
	\end{quote}
	These results establish statistical performance (consistency) of PPR clustering under the DC-SBM, in the sparse regime where the minimum expected degree grows logarithmically with the number of nodes in the network.
	Our results provide an element-wise perturbation bound for PPR vectors, that allows the number of clusters to grow with the size of graphs, and generalize to a directed graph setting as PageRank does.
\end{enumerate}

The rest of the paper proceeds as follows. Section \ref{pre} formally introduces the PPR method and some of the known results.  Section \ref{pre} also introduces the degree-corrected stochastic block model.
Section \ref{sxn:pop} gives a population analysis of the PPR clustering under directed block model graphs. 
Section \ref{sxn:spl} provides concentration results for the PPR vector when the graph is random and provides a statistical guarantee on the PPR local clustering method. 
Section \ref{sxn:simu} presents several numerical results showing the effectiveness of the PPR clustering.
Section \ref{sxn:data} illustrates the PPR clustering through the massive Twitter friendship graph and demonstrates the benefits of a smoothing step in the PPR adjustment.

%\textsc{Notations.} For any vector $x\in\R^n$, define the vector infinity norm as $\|x\|_\infty=\max_i|x_i|$.

%	\section{Some \LaTeX{} Examples}
%	\label{sec:examples}
%	
%	\subsection{How to Leave Comments}
%	
%	Comments can be added to the margins of the document using the \todo{Here's a comment in the margin!} todo command, as shown in the example on the right. You can also add inline comments:
%	
%	\todo[inline, color=green!40]{This is an inline comment.}

\section{Preliminaries}\label{pre}

Throughout this paper, let $G=(V,E)$ denotes an unweighted and connected graph, where $E$ is the edge set and $V$ is the set of vertices indexed by $1, ..., N$. 
When $G$ is an undirected and unweighted graph, encode $E$ into a binary \textit{adjacency} matrix $A\in\{0,1\}^{N\times N}$ with $A_{uv}=A_{vu}=1$ if and only if edge $(u,v)$ appears in $E$. 
Define a diagonal matrix $D=\diag{d_1,...,d_N}$ and the \textit{graph transition} matrix $P$ as follows:
\begin{eqnarray*}
	d_u=\sum_{v\in V} A_{uv} &\text{and}& P=D^{-1}A.
\end{eqnarray*}
When $G$ is a directed graph, the adjacency matrix $A\in\{0,1\}^{N\times N}$ accordingly becomes asymmetric with $A_{uv}=1$ if and only if edge $(u,v)\in E$, and the graph transition matrix is defined as 
$$P=[\Dout]^{-1}A,$$ 
where $\Dout = \diag{\dout_1, ..., \dout_N}$ and $\dout_u = \sum_{v\in V}A_{uv}$ is the number of edges leaving from node $u$.
In addition, define $\Din=\diag{\din_1, ..., \din_N}$ where $\din_v=\sum_{u\in V}A_{uv}$ is the number of edges pointing to node $v$.

\subsection{Personalized PageRank and the local clustering algorithm} \label{prePPR}
The personalized PageRank (PPR) is an extension of Google's PageRank \citep{Brin1998The,haveliwala2003topic}. 
To illustrate, consider a personalized random walk (or originally called ``surfing'') on the graph $G=(V,E)$ with a \textit{seed node} $v_0\in V$. 
At each step, the random walker either restarts from the seed node $v_0$ with probability $\alpha$ (called the \textit{teleportation constant}) or continues the random walk from the current node to a neighbor uniformly at random. 
The \textit{personalized PageRank vector} $p\in [0,1]^N$ is the stationary distribution of this process, thus the solution to the equation
\begin{equation}\label{eqn:ppr_def}
p\T=\alpha \pi\T + (1-\alpha) p\T P,
\end{equation}
where $P$ is the graph transition matrix, and $\pi$ is the elementary unit vector in the direction of seed node $v_0$. 
Here $p$ is a column vector normalized by a positive scalar such that its elements sum to 1, and without loss of generality, we set $v_0=1$ and thus $\pi=(1,0,...,0)\T$.

In general, the \textit{preference vector} $\pi$ does not have to be an elementary unit vector, but any probability distribution on $V$.
For example, when $\pi=(1/N,...,1/N)\T$, PPR is equivalent to ordinary PageRank. 
Moreover, the PPR vector is a linear function of the preference vector. That is, let $p(\pi_1)$ and $p(\pi_2)$ be two PPR vectors corresponding to two preference vectors $\pi_1$ and $\pi_2$ respectively. Then, for a new preference vector that is a convex combination of $\pi_i$, the resulting PPR vector is constructive of $p(\pi_i)$,
$$p(w_1\pi_1+w_2\pi_2)=w_1p(\pi_1)+w_2p(\pi_2),$$
where $w_i\ge0$ and $w_1+w_2=1$.
Define $\Pi$ to be an $N \times N$  matrix with repeating rows $\pi\T$, and let $Q=\alpha \Pi+(1-\alpha)P$, then $Q$ is the Markov transition matrix for the stochastic process and Equation (\ref{eqn:ppr_def}) becomes $p\T=p\T Q$. 
Below are some useful properties of the PageRank vector (also see \citet{haveliwala2003topic, jeh2003scaling} and Appendix \ref{apd:proof}).

\begin{proposition}\label{ppr_sol}
	For any fixed $\alpha\in(0,1]$, the PPR vector $p$ is
	\begin{enumerate}[\normalfont(a)]
		\item the left leading eigenvector of $Q$, associated with the simple eigenvalue 1; and
		\item the infinite sum of landing probability $\{ \left(P^s\right)\T \pi \}_{s=0}^\infty$ with weights $\phi=\{\alpha(1-\alpha)^s\}_{s=0}^\infty$,
		\begin{equation}\label{eqn:ppr_sum}
		p\T=\alpha\sum_{s=0}^{\infty} (1-\alpha)^s \pi\T P^s.
		\end{equation}
	\end{enumerate}
\end{proposition}

%	Proposition \ref{ppr_sol} provides two obvious ways to obtain a PPR vector: either solve the linear system in Equation (\ref{eqn:ppr_def}) or calculate the weighted sum of \textit{landing probability} $P^s\pi$ in Equation (\ref{eqn:ppr_sum}).
\citet{berkhin2006bookmark} gives an iterative algorithm based on Proposition \ref{ppr_sol} to approximate the PPR vector (that scales to large graphs); 
%	shows that updating on just one randomly chosen vertex at each step will give the same stationary distribution. 
each update requires only neighborhood information of one visited vertex.
A few lines of linear algebra show that the PPR vector is equivalent to the solution to the linear system
$$p\T=\alpha'\pi\T+(1-\alpha')p\T W,$$ 
where $W=(I+P)/2$ is the lazy graph transition matrix and $\alpha'= \alpha/(2-\alpha)$. 
Using this fact, Algorithm \ref{alg:ppr} approximates the PPR vector in running time of order $\cO\left(\frac{1}{\epsilon\alpha}\right)$, by reaching at most $\frac{2}{\epsilon(1-\alpha)}$ vertices. The following proposition gives a guarantee on the approximation error for this algorithm in terms of the \textit{tolerance} parameter and the degrees of visited nodes.
%Algorithm \ref{alg:ppr} approximates the PPR vector in a running time linear to the vertices visited $\cO\left((\epsilon\alpha)^{-1}\right)$, thus nearly linear to the size of the output, with a guarantee on approximation error \citep{andersen2006local}. The next proposition bounds the approximation errors in terms of the \textit{tolerance} parameter and the degrees of visited nodes.
\begin{proposition}[Entrywise approximation error \citep{andersen2006local}]\label{prop:approx}
	Let $p$ be a PPR vector, and let $p^\epsilon\in {[0,1]}^N$ be an approximate PPR vector computed by Algorithm \ref{alg:ppr} with
	%		 the same teleportation constant $\alpha$ and
	a tolerance $\epsilon>0$. For any vertex $u$ that is sampled in Algorithm \ref{alg:ppr},
	$$\left|p_u-p^\epsilon_u\right|\le \epsilon d_u.$$
\end{proposition}

Proposition \ref{prop:approx} ensures that for any fixed graph, the approximate PPR vector is arbitrarily close to the exact PPR vector, as long as the tolerance $\epsilon>0$ is sufficiently small. 
%	Hence for notation simplicity, when it is clear, we drop the $\epsilon$ superscript that denoting an approximate PPR returned by Algorithm \ref{alg:ppr}.
Appendix \ref{apd:proof} contains a proof of this proposition for completeness. 
Given a seed node in the graph, Algorithm \ref{alg:lc} uses the approximate PPR vector from Algorithm \ref{alg:ppr} and returns a set of nodes with the largest corresponding values in the \textit{adjusted personalized PageRank} (aPPR) vector, which is defined as
$$p^*_v=\frac{p_v}{d_v}, \text{ for } v=1,2,...,N.$$ 
The aPPR vector was previously proposed in \citet{andersen2006local}.
Algorithm \ref{alg:ppr} and \ref{alg:lc} operate on undirected graphs. We will generalize them to directed graphs in Section \ref{sxn:pop} thanks to a simplified and interpretable form for the PPR vector.

\begin{algorithm}%[t]
	\caption{Approximate PPR Vector (undirected) \citep{andersen2006local}}
	\label{alg:ppr}
	\begin{algorithmic}
		\REQUIRE Undirected graph $G$, preference vector $\pi$, teleportation constant $\alpha$, and tolerance $\epsilon$.
		\STATE \textbf{Initialize} $p\leftarrow0$, $r\leftarrow\pi$, $\alpha'\leftarrow \alpha/(2-\alpha)$.
		\WHILE {$\exists u\in V$ such that $r_u\ge \epsilon d_u$}
		\STATE Uniformly sample a vertex $u$ satisfying $r_u\ge \epsilon d_u$.
		\STATE $p_u\leftarrow p_u+\alpha' r_u$.
		\FOR {$v:(u,v)\in E$}
		\STATE $r_v\leftarrow r_v+(1-\alpha')r_u/(2d_u)$.
		\ENDFOR
		\STATE $r_u\leftarrow (1-\alpha')r_u/2$.
		\ENDWHILE
	\end{algorithmic}
	\textbf{Return:} $\epsilon$-approximate PPR vector $p$.
\end{algorithm}

\begin{algorithm} %[t]
	\caption{PPR Clustering (undirected)}
	\label{alg:lc}
	\begin{algorithmic} [1]
		\REQUIRE Undirected graph $G$, seed node $v_0$, and the desired size of local cluster $n$.
		\STATE Calculate the approximate PPR vector $p$ (Algorithm \ref{alg:ppr}). \label{alg_a}
		
		\STATE Adjust the PPR vector $p$ by node degrees, $p^*_v\leftarrow p_v/d_v$. \label{alg_b}
		
		\STATE Rank all vertices according to the adjusted PPR vector $p^*$. \label{alg_c}
	\end{algorithmic}
	\textbf{Return:} local cluster -- $n$ top-ranking nodes.
\end{algorithm}

%	Note that in order to find all clusters, use exact PPR vector in step \ref{alg_a}, and divide the rankings into $K$ clusters (e.g. by k-means) in \ref{alg_c} instead. 
%	Algorithm \ref{alg:lc} gives consistent clusters in this case (Section \ref{sxn:pop}). 

%\subsection{Degree-corrected stochastic block model (DC-SBM)}
\subsection{Stochastic block model}
%	It is natural to formulate the local clustering problem under the stochastic block model (SBM), where 
%	In particular, given a seed node $v_0$ in block 1, we wish to find a small cluster containing all nodes in block 1. 

%	Given a seed node $v_0$ in block 1, we wish to understand when Algorithm \ref{alg:lc} returns a cluster containing all nodes in block 1. 
In the stochastic block model (SBM), each node belongs to one of $K$ blocks. The presence of each edge corresponds to an independent Bernoulli random variable, where the probability of an edge between any two nodes depends only on the block memberships of two nodes \citep{holland1983stochastic}. 
The formal definition is as follows.
\begin{definition}
	For a vertex set $V=\{1,2, ..., N\}$, let $z:\{1,2,...,N\}\rightarrow\{1,2,...,K\}$ partition the $N$ nodes into $K$ blocks, so $z(v)$ is the block membership of vertex $v$. Let $\B$ be a $K \times K$ matrix with all entries range in $[0,1]$. Under the SBM, the probability of an edge between $u$ and $v$ is $\B_{z(u)z(v)}$. That is, $A_{uv} \mid z(u),z(v) \overset {\normalfont\text{ind.}}{\sim} \text{Bernoulli}\left(\B_{z(u)z(v)}\right)$, for any $u,v \in \{1,2,...,N\}$.
\end{definition}

Under the ordinary SBM, nodes in the same block have the same expected degree. 
One extension is the degree-corrected stochastic block model (DC-SBM), which adds a series of parameters ($\theta_v > 0$ for every vertex $v$) to create more heterogeneous node degrees \citep{karrer2011stochastic}. 
Let $\B$ be a $K \times K$ matrix with $\B_{ij} > 0$ for any $i$ and $j$. 
Then the probability of an edge between $u$ and $v$ is $\theta_u\theta_v\B_{z(u)z(v)}$. 
That is, 
$$A_{uv} \mid z(u),z(v) \overset {\text{ind.}}{\sim} \text{Bernoulli}\left(\theta_u\theta_v\B_{z(u)z(v)}\right),$$ 
for $u,v \in \{1,2,...,N\}$. Since $\theta_v$'s are arbitrary to a multiplicative constant which can be absorbed into $\B$, \citet{karrer2011stochastic} suggest imposing the constraint  that the $\theta_v$'s sum to 1 within each block. That is, $\sum_{v:z(v)=i} \theta_v=1$ for all $i=1,2,..,K$. With this constraint, 
%	matrix $\B$ has an explicit interpretation;
$\B_{ij}$ represents the expected number of edges between block $i$ and $j$ if $i\neq j$, and twice of that if $i=j$. 
Throughout this paper, we presume $\B$ is positive definite
\footnote{This prevents scenarios where edges are unlikely within blocks and more likely between blocks. In such scenarios, local clustering needs to be reimagined cautiously. See Supplementary Materials \ref{sxn:parameter} for additional details about generalizations.}
and all blocks are connected (we ignore any blocks that are isolated from the seed).
The DC-SBM can be generalized to directed graphs by giving each node two parameters, $\thi_v$ and $\tho_v$, controlling its in-degree and out-degree respectively \citep{zhu2013oriented}. 
Then, the presence of an directed edge from $u$ to $v$, given the block memberships, corresponds to an independent Bernoulli random variable,
$$A_{uv} \mid z(u),z(v) \overset {\text{ind.}}{\sim} \text{Bernoulli}\left(\tho_u\thi_v\B_{z(u)z(v)}\right).$$
In order to make the model identifiable, we need to impose a structural constraint on $\thi$'s and $\tho$'s, that both of them sum up to 1 within each block,
\begin{equation*} \label{eq:zTz_dir}
\sum_{v:z(v)=i}\thi_v=\sum_{v:z(v)=i}\tho_v=1, \text{ for any } i=1,2,...,K.
\end{equation*}
Because the off-diagonal elements of $\B$ can be interpreted as the expected number of edges between blocks, we define the block in-degree and block out-degree to be the total number of incoming edges and outgoing edges respectively, that is,
$\tdi_j=\sum_{i=1}^K\B_{ij}$, and $\tdo_i=\sum_{j=1}^K\B_{ij}$.

\section{Population Analysis of PageRank}
\label{sxn:pop}

In this section, we analyze the PPR vector of the expected adjacency matrix under the DC-SBM.
This provides a simple representation of the PPR vector that motivates  (1) the bias adjustment and (2) the generalization of Algorithm \ref{alg:ppr} and \ref{alg:lc} to directed graphs.  

We use three distinct typefaces to denote three classes of objects. Calligraphic typeface is given to the population version of any observable quantities in random graphs, such as graph adjacency matrix and node degrees (e.g. Equation \eqref{eq:popA}). 
Normal typeface is given to unobserved model parameters, such as block membership and degree parameters $\theta_i$. 
Bold face is given to all block-level quantities and parameters like $\B$ and $\tdo_i$.

%	Then, Equation (\ref{eq:zTz_dir}) reads compactly in matrix form,
%	\begin{equation}
%	Z\T \Ti Z=Z\T \To Z=I_K. \label{eq:ZTZ_dir}
%	\end{equation}

Define the population graph adjacency matrix, 
\begin{equation}\label{eq:popA}
\A=\E\left(A\mid z(1), z(2), ..., z(N)\right),
\end{equation} 
to be the expectation of random adjacency matrix $A$. Let $Z\in\{0,1\}^{N\times K}$ be the block membership matrix with $Z_{vi}=1$ if and only if vertex $v$ belongs to block $i$, and define diagonal matrices $\Ti$ and $\To$ with entries $\thi$'s and $\tho$'s respectively. Then, under the directed DC-SBM with $K$ blocks and parameters $\{\B, Z, \Theta^\text{in}, \Theta^\text{out}\}$, $\A\in\R^{K\times K}$ can be compactly expressed as
$$\A=\To Z\B Z\T\Ti.$$

Accordingly, we define the population node degrees and the population transition matrix, $\di_u=\sum_{v\in V}\A_{uv}$, $\dou_v=\sum_{u\in V}\A_{uv}$, and $\cP=\left[\Do\right]^{-1}\A$,
where $\Di$ and $\Do$ are the diagonal matrices of the population node in-degrees $\di_u$'s and out-degrees $\dou_v$'s respectively.
Let $\p$ be the population PPR vector (i.e., the solution to equation $\p\T=\alpha\pi\T+(1-\alpha)\p\T\cP$) and let $\p^*=\left[\Di\right]^{-1}\p$ be the population aPPR vector.

In addition, define the \textit{block transition matrix} $\tP \in \R^{K\times K}$ as
\begin{equation}\label{eqn:blcoktrans}
\tP=\left[\tDo\right]^{-1}\B,
%		\tdi_i=\sum_{v:z(v)=i}\di_v,&\text{and}&\tdo_i=\sum_{v:z(v)=i}\dou_v.
\end{equation}
where $\tDi\in\R^{K\times K}$ and $\tDo\in\R^{K\times K}$ are diagonal matrices of the block in-degrees $\tdi_i$'s and out-degrees $\tdo_i$'s.

\subsection{A representation of PPR vectors} 

This section provides a simple and interpretable form for PPR vectors under the population DC-SBM.
To this end, we define the \textit{``block-wise'' PPR vector} $\tp\in\R^K$ to be the unique solution to linear system
\begin{equation} \label{ppr_block}
\tp\T=\alpha\boldsymbol{\pi}\T+(1-\alpha)\tp\T\tP,
\end{equation} 
where $\boldsymbol{\pi}=Z\T \pi\in\R^K$ is the block-wise preference vector and $\tP$ is the block transition matrix in Equation (\ref{eqn:blcoktrans}).
This treats the block connectivity matrix $\B$ as a weighted adjacency matrix of blocks and the block of seed node as a seed block.
To build up the relationship between PPR and  the block-wise PPR, the next theorem gives an explicit form for PPR vectors which also reveals the sources of bias for local clustering.

%Go get coffee, buy subway and cookies. -Michael

\begin{theorem} [Explicit form of PPR vectors] \label{explicit2}
	\underddcsbm 
	\begin{enumerate}[{\normalfont (a)}]
		\item the population PPR vector $\p \in \R^N$ has elements 
		$$\p_u = \thi_u \tp_{z(u)}$$
		%			$$\p= \Ti Z \tp ,$$
		where $\tp$ is the block-wise PPR vector in Equation (\ref{ppr_block}), 
		%			 Element-wise, $\p_u = \thi_u \tp_{z(u)}$.
		\item and the population aPPR vector $\p^* \in \R^N$ has elements
		\begin{equation}\label{eqn:appr}
		\p^*_u = \tp^*_{z(u)}
		%			\p^*=Z\tp^*,
		\end{equation}
		where $\tp^*=\left[\tDi\right]^{-1}\tp$.  
		%			Element-wise, $\p^*_u = \tp^*_{z(u)}$.
	\end{enumerate} 
\end{theorem}

Theorem \ref{explicit2} demonstrates that the PPR vector $\p$ decomposes into block-related information ($\tp$) and node specific information ($\Theta$).
%	, connected by block membership information ($Z$). 
%	A PPR vector is the expansion of the block-wise PPR vector with node degree parameters. 
%	In particular, the PPR values of all nodes in block $i$, $z(v)=i$, have form $\theta_v\tp_i$. 
Within each block, the PPR values are proportional to the node degree parameters $\theta_v$'s and sum up to the block-wise PPR value of the block.
The proof of Theorem \ref{explicit2} (Appendix \ref{apd:proof}) relies on a key observation (Appendix \ref{apd:lemmas}) that the powers of population transition matrix, $\cP^s$ for $s = 1, 2, \dots$, have a similarly simple form and the node specific information components (i.e., $z(v)$ and $\theta_v$) are invariant in $s$. 

%The PPR value has a simple structure induced from the block-wise PPR vector $\tp$. Separated by blocks, all nodes of the same block share one entry of $\tp$ which corresponds to that block, in proportion to their node degree parameters $\Theta$. 

In order to justify the adjustment (Step \ref{alg_b}) in Algorithm \ref{alg:lc}, we observe that the seed always has the highest population aPPR score. This turns out to be a key feature that facilitates the aPPR vector to recover a local cluster correctly, so we state it in the following lemma.
\begin{lemma}[The largest entry of aPPR vector] \label{lem:top}
	%	\underddcsbm
	Under the population DC-SBM, assume that the minimum expected degree is positive, that is, $\min_{v\in V}\cd_v>0$. 
	Then, for any fixed $\alpha>0$, the population aPPR vector $\p^*$ has the strictly largest entry corresponding to the seed node,
	$$\p^*_{v_0}>\p^*_v, \text{ for any } v\neq v_0.$$
	On the other hand, this is not generally true for a PPR vector.
\end{lemma}

When $\alpha=0$ (i.e., no teleportation), the PPR vector becomes the limiting distribution of a standard random walk and all entries of the aPPR vector are equal (Appendix \ref{apd:proof}). 
Lemma \ref{lem:top} (applied to block-wise PPR vectors) and Theorem \ref{explicit2} together identify two sources of bias for PPR vectors and suggest a justification for the degree adjustment, which we discuss in order: 
\begin{enumerate}[(i)]
	\item \label{rmk2} Both node degree heterogeneity ($\Theta$) and block size imbalance ($\tD$) confound the identification of local cluster by the PPR vector. 
	In particular, suppose vertex $v$ belongs to a block $z(v)=i$ other than 1. PPR vector assigns it a score $\theta_v\tp_i$, where $\tp_i$ is the block-wise PPR of block $i$, and $\theta_v$ is the parameter specifically controlling the degree of $v$. 
	Then, node $v$ may rank at the top, if $\theta_{v}$ is large enough. 
	Furthermore, Lemma \ref{lem:top} implies that $\tp_1$ is not necessarily the largest due to block degree heterogeneity. Specifically, if block $i$ has an exceedingly high block degree, it is likely that $\p$ fails to down-rank node $v$ vis-a-vis those nodes of block 1. 
	
	\item \label{rmk3} Adjusted personalized PageRank removes the node and the block degree heterogeneity simultaneously, and perfectly recovers the local cluster. 
	To see this, note that $\tp^*$ is the adjusted version of block-wise PPR vector.
	%		and that the stationary distribution of a random walk in a directed graph is characterized by the in-degree of nodes \citep{ghoshal2011ranking, lu2013respondent}.
	From Lemma \ref{lem:top}, $\tp^*_1$ is the largest entry of $\tp^*$. 
	From Equation \ref{eqn:appr}, the aPPR vector assigns any vertex $v$ a score $\tp^*_{z(v)}$. 
	Hence, nodes with the highest value of $\p^*$ belong to block 1, which is precisely the desired local cluster. 
	
\end{enumerate}

Note that the PPR vector can still be biased for local clustering even under the classic SBM. 
%	; an adjustment step is still needed. 
To see this, set the matrix $\Theta$ to the identity matrix in Theorem \ref{explicit2}. In this case, the heterogeneous block degrees still confound the PPR vector (Section \ref{sxn:simu2}); there is generally no guarantee for $\tp_1$ to appear on the top (due to Lemma \ref{lem:top}), unless there are further symmetry conditions. 
\citet{kloumann2017block} uses such one scenario.
%	\citet{kloumann2017block} presents such a scenario by assuming a symmetric block distribution; the asymptotic results do not hold without this strong assumption. 
As a byproduct of our analysis, we extend their results under the DC-SBM with the symmetric conditions (see Supplementary Materials \ref{sxn:ld} to the paper).

%We also remark on \ref{rmk3} that if we assume in addition that $\p^*$ is entrywise distinguishable, then \ref{rmk3} implies that one can even obtain the global clustering by simply performing a classic K-means algorithm on the aPPR vector. 

%hi, my name is jake. i am helping fan write, he is a very good writer. i love him.

%\subsection{Directed block model graphs} \label{apd:directed}
%In this section, we generalize the Theorem \ref{explicit2} onto directed graphs, where the PageRank methods are typically invoked.

\subsection{Local clustering on  directed graphs}
In light of the clean form of PPR vectors under the DC-SBM, one can modify Algorithm \ref{alg:ppr} and \ref{alg:lc} to operate on a directed graph accordingly. To this end, note that the transition matrix of a directed graph requires node out-degrees, hence Algorithm \ref{alg:ppr} examines only the edges leaving visited nodes. Consequently it suffices to replace $d_u$'s in Algorithm \ref{alg:ppr} by $\dout_u$'s (Algorithm \ref{alg:ppr_dir}).
Proposition \ref{prop:approx} applies to Algorithm \ref{alg:ppr_dir} as well, and one can approximate the PPR vector provided the out-degrees of visited nodes can be observed and the tolerance parameter $\epsilon>0$ is sufficiently small.

To perform local clustering on a directed graph, Algorithm \ref{alg:lc_dir} adjusts the approximate PPR vectors from Algorithm \ref{alg:ppr_dir} by node in-degrees, that is,
$$p^*_v=\frac{p_v}{\din_v}, \text{ for } v=1,2,...,N.$$
Another option is regularized adjustment, which produces the \textit{regularized} PPR (rPPR) vector,
$$p^\tau_v=\frac{p_v}{\din_v+\tau}, \text{ for } v=1,2,...,N,$$
where $\tau>0$ is the regularization parameter.
The regularized adjustment greatly stabilize the PPR clustering in practice, by removing nodes with extremely low in-degrees (see Section \ref{sxn:data} for more details). 
Adjusted PPR for directed graphs is a local algorithm so long as $\din$ is available with a local query, for example, the Twitter friendship graph. 

\begin{algorithm}
	\caption{Approximate PPR Vector (directed)}
	\label{alg:ppr_dir}
	\begin{algorithmic}
		\REQUIRE Directed graph $G$, preference vector $\pi$, teleportation constant $\alpha$, and tolerance $\epsilon$.\\
		\STATE \textbf{Initialize} $p\leftarrow0$, $r\leftarrow\pi$, $\alpha'\leftarrow\alpha/(2-\alpha)$.
		\WHILE {$\exists u\in V$ such that $r_u\ge \epsilon \dout_u$}
		\STATE Sample a vertex $u$ uniformly at random, satisfying $r_u\ge \epsilon \dout_u$.
		\STATE $p_u\leftarrow p_u+\alpha' r_u$.
		\FOR {$v:(u,v)\in E$}
		\STATE $r_v\leftarrow r_v+(1-\alpha')r_u/(2\dout_u)$.
		\ENDFOR
		\STATE $r_u\leftarrow (1-\alpha')r_u/2$.
		\ENDWHILE
	\end{algorithmic}
	\textbf{Return:} $\epsilon$-approximate PPR vector $p$.
\end{algorithm}

\begin{algorithm}
	\caption{PPR Clustering (directed)}
	\label{alg:lc_dir}
	\begin{algorithmic} [1]
		\REQUIRE Directed graph $G$, seed node $v_0$, the desired size of local cluster $n$, and an optional regularization parameter $\tau$.
		\STATE Calculate the approximate PPR vector $p$ (Algorithm \ref{alg:ppr_dir}). 
		
		\STATE Adjust the PPR vector $p$ with: \\
			Option (a): node in-degrees, $p^*_v\leftarrow p_v/\din_v$,\\
			Option (b): regularized node in-degrees, $p^\tau_v\leftarrow p_v/(\din_v+\tau)$.
		
		\STATE Rank all vertices according to the aPPR vector $p^*$ or $p^\tau$. 
	\end{algorithmic}
	\textbf{Return:} local cluster -- $n$ top-ranking nodes.
\end{algorithm}

\section{Personalized PageRank in Random Graphs} \label{sxn:spl}

This section establishes several concentration results for the local clustering algorithm using the adjusted PPR vector (Algorithm \ref{alg:lc} and \ref{alg:lc_dir}) under the DC-SBM. 
The results show that if the graph is generated from the DC-SBM, then PPR clustering returns the desired local cluster with high probability. 
Since in Algorithm \ref{alg:lc_dir}, the calculation for PPR vectors only relies on node out-degrees and the adjustment step solely utilizes node in-degrees, it is not difficult to distinguish $\din$ and $\dout$. Thus, we state the results in undirected graphs for simplicity. One can draw the analogous conclusions for directed graphs by tracing the proof step by step.

We first present a useful tool that controls the entrywise errors of a PPR vector in random graphs. Recall that $\p$ is the stationary distribution of probability transition matrix $\Q=\alpha \Pi+(1-\alpha)\cP$.
%	Note that $\Pi$ adds at most 1 to the rank of $\Q$, and because we presume $B$ is positive definite $\cP$ has exactly $K$ positive eigenvalues among other zeros (see Supplementary Materials \ref{sxn:parameter}). 
For any vector $x\in\R^n$, define the vector infinity norm as $\|x\|_\infty=\max_i|x_i|$.
The following theorem bounds the entrywise error of the stationary distribution of $\Q$.

\begin{theorem} [Concentration of the PPR vectors] \label{thm:concPPR}
	\dcsbmset Let $p$ and $\p$ be the PPR vector corresponding to random transition matrix $P$ and its population version $\cP$ respectively, with the same teleportation constant $\alpha$. Let $p^*, \p^* \in [0,1]^N$ be the adjusted PPR vector of $p$ and $\p$. Let $\delta$ be the average expected node degrees, that is, $\delta=\frac{1}{N}\sum_{v\in V}\cd_{v}$. 
	Assume that  $\rho=\frac{\max_{v\in V}\cd_{v}}{\min_{v\in V}\cd_{v}}$ is bounded by some finite constant and that
	\begin{equation}\label{eq:assm1}
	\delta>c_0(1-\alpha)^2\log{N},
	\end{equation}
	for some sufficiently large constant $c_0>0$. Then, with probability at least $1-\cO(N^{-5})$,
	\begin{eqnarray*}
		\frac{\|p-\p\|_\infty}{\|\p\|_\infty} \le c_1(1-\alpha)\sqrt{\frac{\log{N}}{\delta}},
		&\text{and}&
		\frac{\|p^*-\p^*\|_\infty}{\|\p^*\|_\infty} \le c_2(1-\alpha)\sqrt{\frac{\log{N}}{\delta}},
	\end{eqnarray*}
	for some sufficiently large constant $c_1, c_2>0$.
\end{theorem}

The proof of Theorem \ref{thm:concPPR} invokes the elementary eigenvector perturbation bound for asymmetric matrices, an analog to the celebrated Davis-Kahan $\sin\Theta$ theorem \citep{davis1970rotation}, and the novel leave-one-out technique due to \citet{chen2017spectral}. The detailed proof is given in the Supplementary Materials \ref{sxn:rand} to the paper. 

Theorem \ref{thm:concPPR} demonstrates that if the expected average degree $\delta$ exceeds $(1-\alpha)^2\log{N}$ to some sufficiently large extent, then with high probability, the random aPPR vector concentrates around the population aPPR vector in terms of all entries. 
In fact, the concentration statement holds for any valid preference vector $\pi$. Hence, the classic PageRank vector and some other variants also enjoy the entrywise error bounds, so long as they can be written as the solution to the linear system (\ref{eqn:ppr_def}).

Next, we introduce a separation measure of the DC-SBM. 
Recall that one can conduct a local clustering task by selecting nodes ranked by the adjusted PPR vector $p^*$. 
In the population version, it is equivalent to distinguishing between $\tp^*_1$ and $\tp^*_k$, for all $k=2,3,...,K$, which also characterizes the distance from the desired local cluster (block 1) to its complement set (the other blocks). 
Only if they are sufficiently separated, can the local cluster be identifiable in the sample. 
Due to Lemma \ref{lem:top}, we assume without loss of generality that the second block has the second highest value in the ``block-wise'' aPPR vector, that is, $\tp^*_1>\tp^*_2\geq\tp^*_k$ for $k=3,4,...,K$.
%	\begin{equation*}
%	\tp^*_1>\tp^*_2\geq...\geq\tp^*_K.
%	\end{equation*}
%	Then, it is simply the gap between $\tp^*_1$ and $\tp^*_2$. With this thinking,
Then, we define the \textit{separation measure} $\Delta_\alpha\in(0,1]$, 
$$\Delta_\alpha=\frac{\tp^*_1-\tp^*_2}{\tp^*_1},$$ 
which turns out to be crucial in determining the sample complexity required to guarantee the exact recovery. We remark that $\Delta_\alpha$ is an increasing function of the teleportation constant, hence the subscript $\alpha$.

With Theorem \ref{thm:concPPR} and the separation measure, we then give following corollary that bounds the accuracy of Algorithm \ref{alg:lc}, in terms of graph edge density.

\begin{corollary} [Exact recovery by adjusted PPR vector] \label{cor:main}
	For any seed nodes, let $C \subset V$ be the local cluster of $n$ nodes returned by Algorithm \ref{alg:lc} with teleportation constant $\alpha$ and tolerance $\epsilon$, and $\mathcal{C}\subset V$ be the nodes in the seed node's block. Assume that $\rho<c_0$, $\epsilon\le c_1(1-\alpha)\tp^*_1\sqrt{\log{N}/\delta}$, 
	%		 $\delta>c_1(1-\alpha)^2\log{N}$, % don't need this because $\Delta_\alpha\le1$.
	and that
	\begin{equation}\label{eq:assm2}
	\delta > 16 c_2\left(\frac{1-\alpha}{\Delta_\alpha}\right)^2 \log{N},
	\end{equation}
	for some sufficiently large constants $c_0,c_1,c_2>0$. If the desired size of the local cluster $n=\left|\mathcal{C}\right|$, then with probability at least $1-\cO(N^{-5})$, we have $C=\mathcal{C}$.
	%		$$C=\mathcal{C}.$$
\end{corollary}
The proof of Corollary \ref{cor:main} is presented in Appendix \ref{apd:proof}. We make a few remarks:
\begin{enumerate}[(i)]
	\item Corollary \ref{cor:main} demonstrates that Algorithm \ref{alg:lc} works under a sparse scenario, where the number of edges is exceedingly small in proportion to the number of possible edges in the network. 
	To reach the entrywise control of the aPPR vector and the sufficient separation of local cluster from others, the theorem calls for the expected node degree $\delta$ to grow with only a fraction (for any fixed teleportation constant $\alpha$) of the logarithm of the size of the network, $\logN$. In other words, Algorithm \ref{alg:lc} requires a sample complexity (the number of edges) of order 
	$$\left(\frac{1-\alpha}{\Delta_\alpha}\right)^2N\logN.$$
	\item The results show that $\alpha$ leverages between the sampling complexity and statistical performance of PPR clustering.  
	%		Moving beyond the graph density issue, the theorem suggests the conditions on teleportation constant for PPR clustering. 
	To see this, rearrange condition (\ref{eq:assm2}),
	$$\left(\frac{1-\alpha}{\Delta_\alpha}\right)^2 < {\frac{c'\delta}{\logN}},$$
	for some small enough constant $c'>0$. As $\alpha$ increases, the left hand side is decreasing to zero thus making the condition more likely to hold.  
	On the other hand, as $\alpha$ increases, the tolerance $\epsilon$ must decrease at rate $\cO(1-\alpha)$ in order to guarantee an entrywise control of $p^\epsilon$ analogous to the form in Theorem \ref{thm:concPPR} (Appendix \ref{apd:proof}). 
	More intuitively, if $\epsilon$ does not decrease, then as $\alpha$ goes to one, Algorithm \ref{alg:ppr} may terminate early without reaching all vertices in the desired local cluster. 
	In sum, Algorithm \ref{alg:ppr} and \ref{alg:ppr_dir} need at least $\cO\left(\frac{1}{\alpha(1-\alpha)}\right)$ queries (see Supplementary Materials \ref{sxn:parameter} for an example). This implies that one can approach the conditions in Corollary \ref{cor:main} by setting the teleportation constant sufficiently large, while the computational burden can increase as $\alpha\rightarrow 1$. 
	
	%		\item The theorem facilitates one of the highlights of PPR clustering; allowing the number of clusters $K$ to grow with the size of network $N$. The derivation of accuracy bounds in Theorem \ref{thm:concPPR} does not rely on restricting the model parameter $K$ to be fixed, thus making the PPR methods effective when $K$ grows with $N$, provided that the model assumptions (e.g. $\B$ is positive definite) and the conditions in Corollary \ref{cor:main} are satisfied.
\end{enumerate}

\section{Simulation Studies}\label{sxn:simu}
This section compares the PPR vector and the aPPR vector. The results show the effectiveness and robustness of aPPR vector in detecting a local cluster. 
Experiment 1 utilizes the DC-SBM with a power-law degree distribution and investigates the effects of heterogeneous node degrees. 
Experiment 2 uses the SBM with unequal block sizes to study the influences of heterogeneous block degrees. 
Experiment 3 generates networks from the SBM with equal block sizes and varying edge density to examine the efficacy of PPR methods in sparse graphs. 

In all simulations, we employee the block connectivity matrix $\B$ with homogeneous diagonal elements, $\B_{ii}=b_1$, and homogeneous off-diagonal elements, $\B_{ij}=b_2$ for any $i\neq j$.
Define the signal-to-noise ratio (SNR) to be the expected number of in-block edges divided by the expected number of out-block edges, that is, $b_1/(b_2(K-1))$, where $K$ is the number of blocks.
In particular, we set the SNR to $1.5$ and choose the teleportation constant of $\alpha=0.15$ throughout the section. 
Additional simulation results (illustrating the Theorem \ref{thm:concPPR}) are available from Supplementary Materials \ref{sxn:parameter}.

\subsection{Experiment 1}

\begin{figure} [t]%[!htbp]
	\centering
	\includegraphics[width=.4\textwidth]{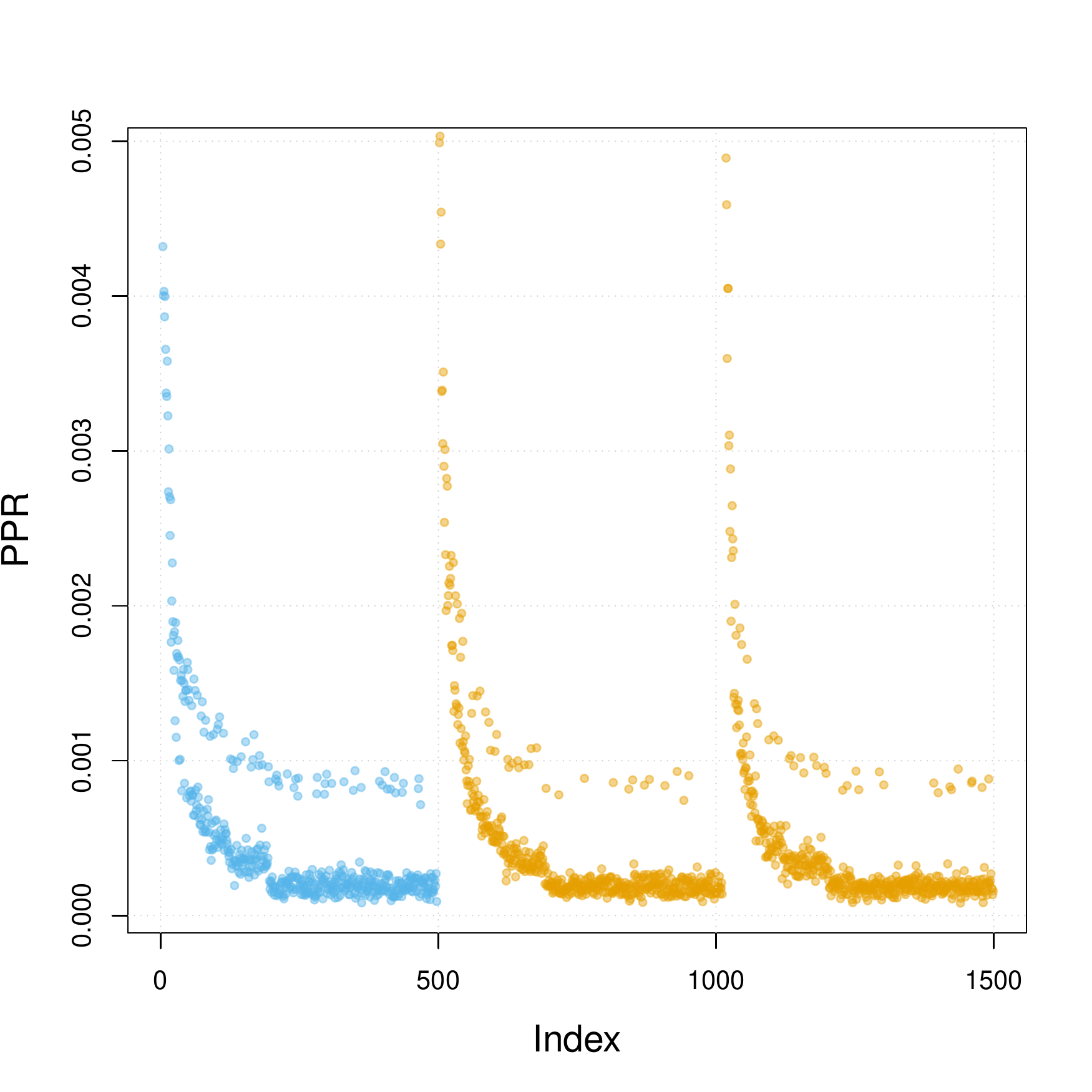} 
	\includegraphics[width=.4\textwidth]{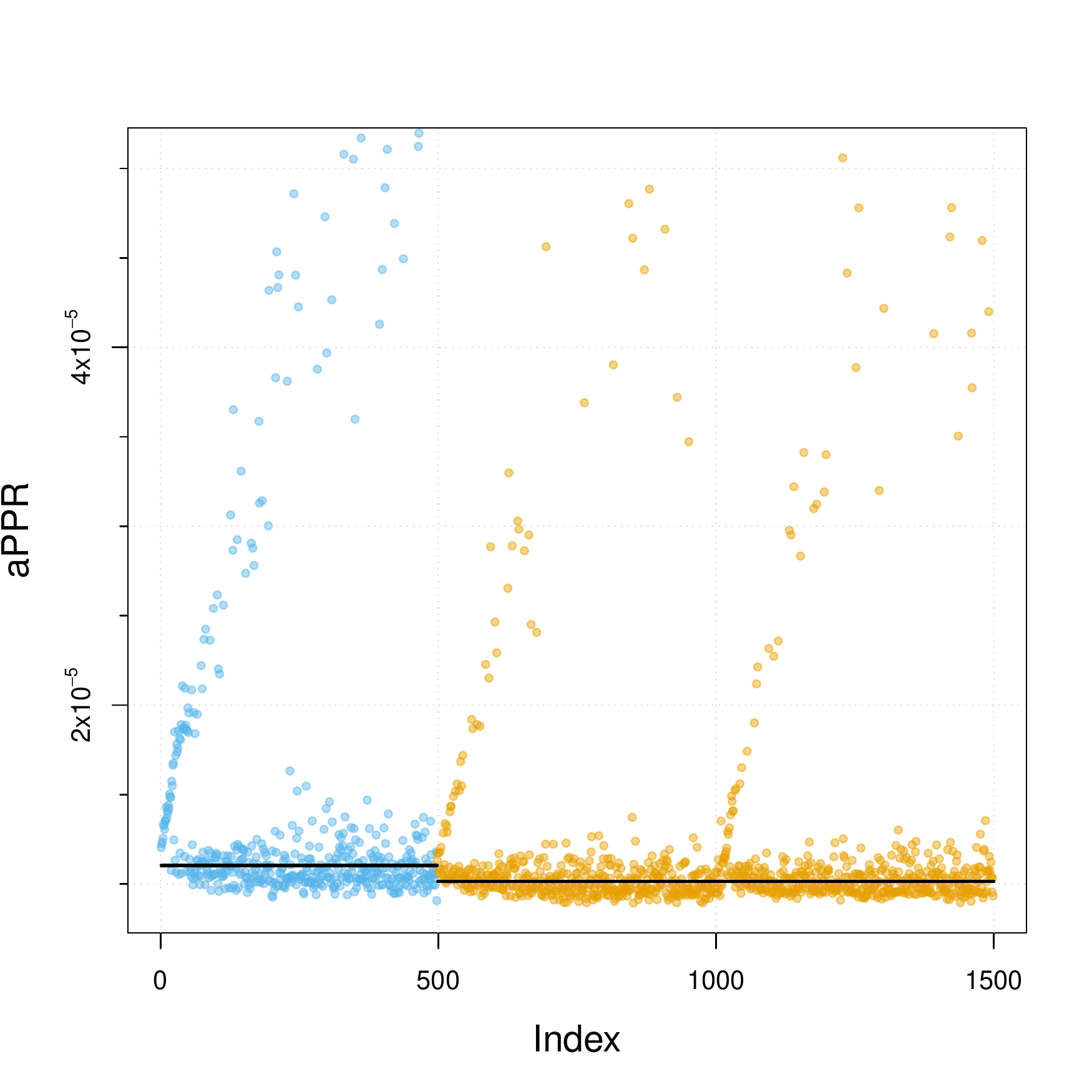} 
	
	\includegraphics[width=.4\textwidth]{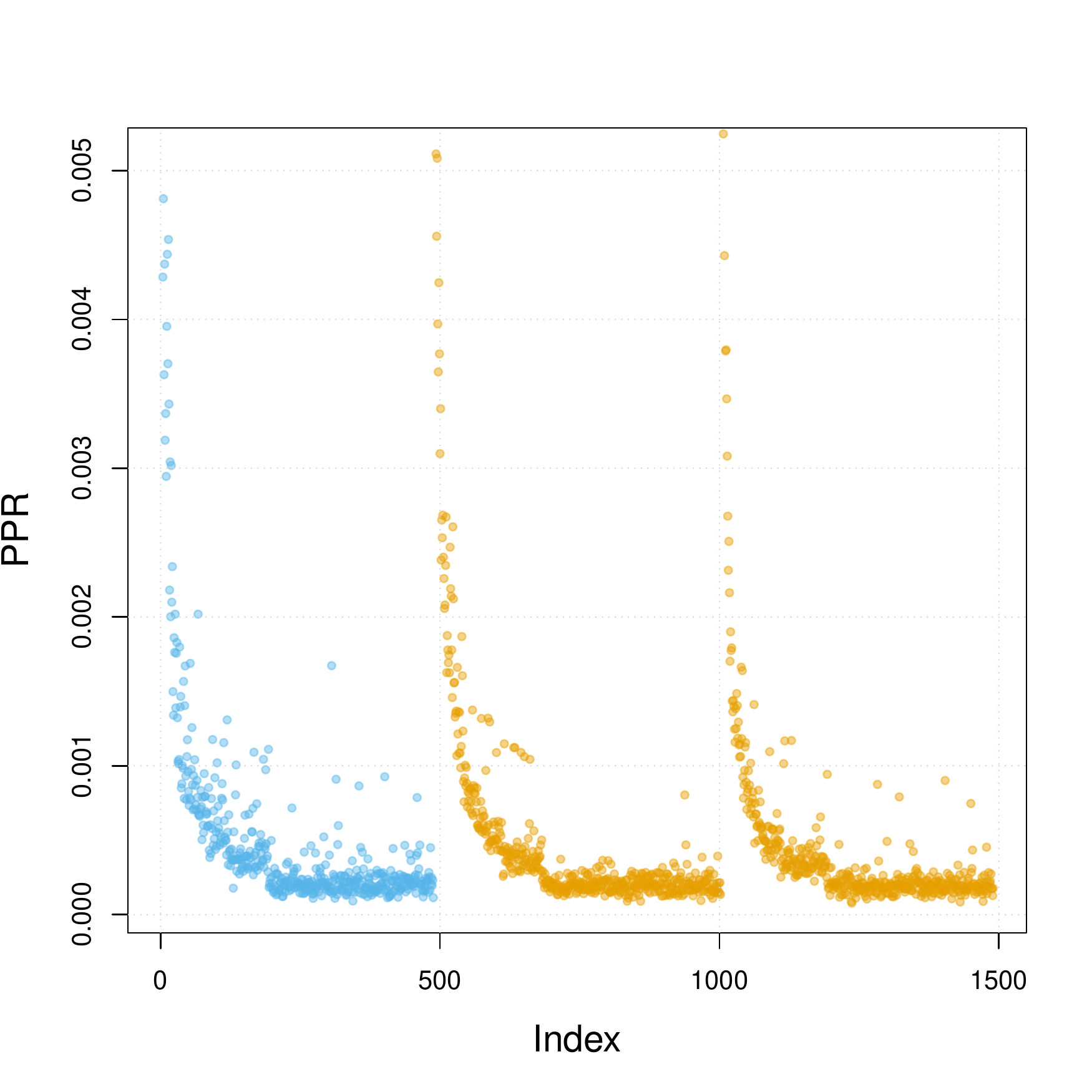} 
	\includegraphics[width=.4\textwidth]{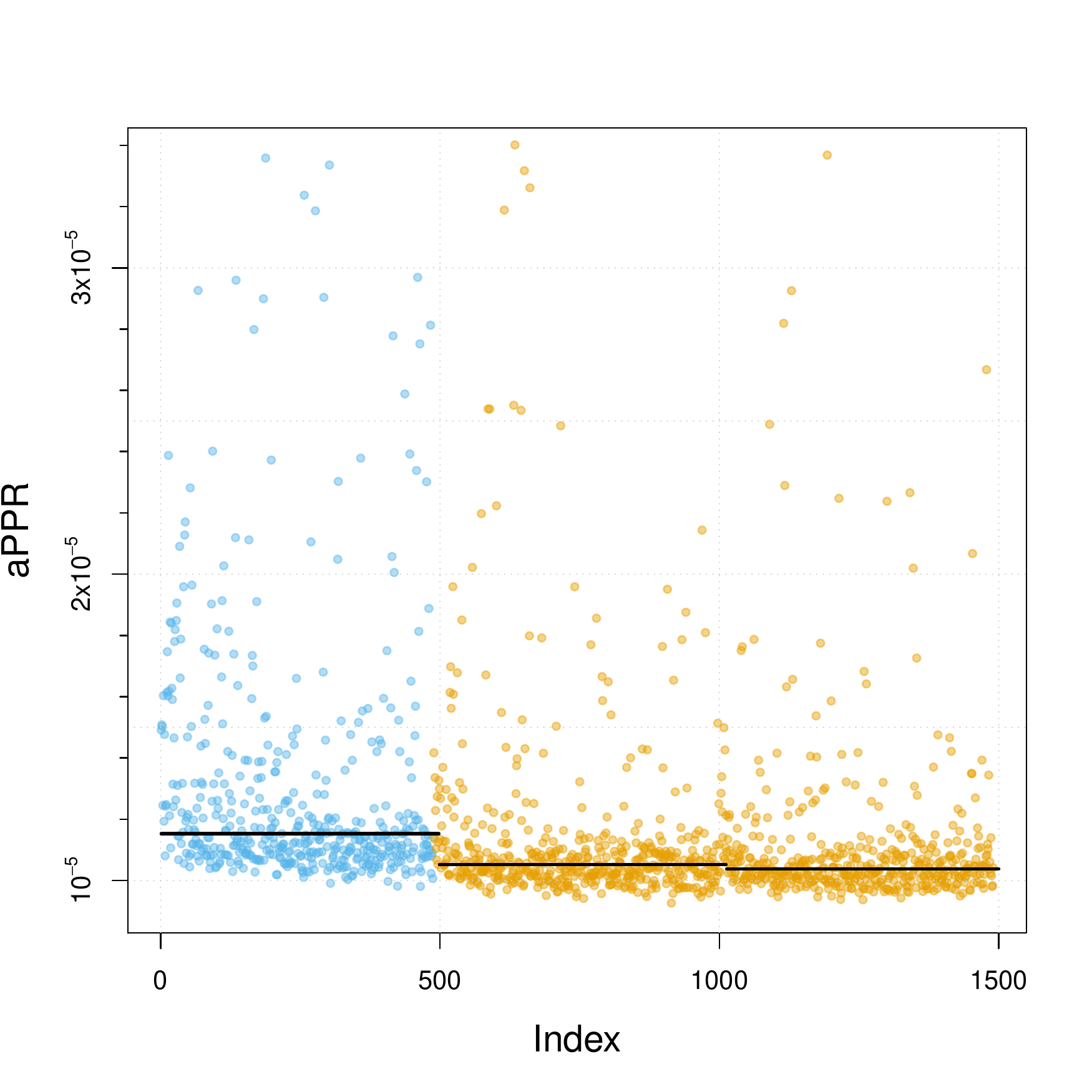} 
	\caption{Comparison of PPR (left two panels) and aPPR (right two panels) under the DC-SBM with one seed node (upper two panels) and ten seed nodes (bottom two panels). Local cluster is in blue and other clusters are colored in yellow. Solid horizontal lines on right panels indicate the median of aPPR values within each cluster.}
	\label{fig:simu1}
\end{figure}

This experiment illustrates how node degree heterogeneity affects the discriminant power in identifying local cluster using a PPR vector or an aPPR vector. The results also illustrate the advantages of having multiple seed nodes. The $\Theta$ parameters from the DC-SBM are drawn from the power law distribution with lower bound $x_{\min}=1$ and shape parameter 
%$\beta\in\left\{2, 2.25, 2.5, 2.75, 3, 3.25\right\}$
$\beta=2.5$. 
%The smaller $\beta$ indicates the greater node degree heterogeneity. 
A random networks were sampled from the DC-SBM with $K=3$, $N=1500$ and equal block sampling proportions, 
$$z(v)\overset{\text{i.i.d.}}{\sim}\text{Multinomial}\left(\frac{1}{3}, \frac{1}{3}, \frac{1}{3}\right),$$
for vertex $v=1,2,...,N$, whose expected average degree ($\delta$) is set to $105$. The PPR vector is calculated with one or ten seeds randomly chosen from block one. 

Figure \ref{fig:simu1} plots PPR values (left two panels) and aPPR values (right two panels) of a random graph generated from the DC-SBM, excluding seed node(s). 
The upper two panels in Figure \ref{fig:simu1} contrast PPR and aPPR when there is only one seed node and the bottom two panels compare two vectors when ten seed nodes are used.
The vertices from the local block in the SBM are colored in blue and the others are in yellow.  The nodes are ordered first by block, then by node degree parameters $\theta$ (left is larger). 
A horizontal line is drawn for each block indicating the median of the aPPR values within that block.

With one seed node (upper two panels), the scatter plots has two clouds within each block. The upper cloud contains the immediate neighbors of the seed node. This separation disappears when multiple seed nodes are used (bottom two panels).
To see the effect of node heterogeneity, the skewed distribution of PPR values in each block demonstrates its bias towards high degree node inside and outside of the seed nodes block in the SBM. In contrast, aPPR values are evenly distributed within blocks, verifying that aPPR vector removes the effects of node degree heterogeneity. 

\subsection{Experiment 2} \label{sxn:simu2}
This experiment compares PPR and aPPR under the SBM with block degree heterogeneity. 
A number of random networks were sampled from the SBM with $K=3$, $N=900$, and geometric block sampling proportions, 
\begin{equation} \label{eq:blocksize}
z(v)\overset{\text{i.i.d.}}{\sim}\text{Multinomial}\left(1, b, b^2\right),
\end{equation}
where $b\in\left\{1.0, 1.2, 1.4, 1.6, 1.8, 2.0\right\}$. When $b$ is larger, the population of nodes in each block becomes more unbalanced and thus inducing greater block degree heterogeneity. 
The block connectivity matrix $\B$ is configured as described in the beginning of this section.
The expected average degree ($\delta$) is set to 70. 
For each sampled network, the size of the first block is assumed known to Algorithm \ref{alg:lc}.
The PPR vector is calculated exactly in place of the approximation PPR vector (Step \ref{alg_a}), with one seed randomly chosen from the first block. 

The top panels of Figure \ref{fig:simu2} displays the PPR vector on an example network with $b=1.4$, demonstrating its preference toward the high degree block (the third block) over local cluster. 
Given the size of the first block, we measure the accuracy by the proportion of vertices belonging to the first block in the returned cluster.
The bottom left panel of Figure \ref{fig:simu2} shows the accuracy of PPR and aPPR for six different values of $b$ (i.e., the geometric ratio in distribution (\ref{eq:blocksize})) where each point is the average of 100 sampled network. The comparison demonstrates that the adjusted PPR vector corrects the bias of PPR caused by block heterogeneity.  
Moreover, block degree heterogeneity degrades the performance of both PPR and aPPR.  Note that aPPR outperforms PPR even when $b=1$; this is likely due to the fact that even when nodes have equal expected degrees in the SBM, the actual node degrees will be heterogeneous due to the randomness in the sampled graph.  In a finite graph, this variability is enough to give aPPR an advantage over PPR.  Asymptotically, this advantage should fade away \citep{kloumann2017block}. 
%	 and is more accurate when there is no block degree heterogeneity.

\begin{figure}[!t]
	\centering
	\includegraphics[width=0.4 \textwidth]{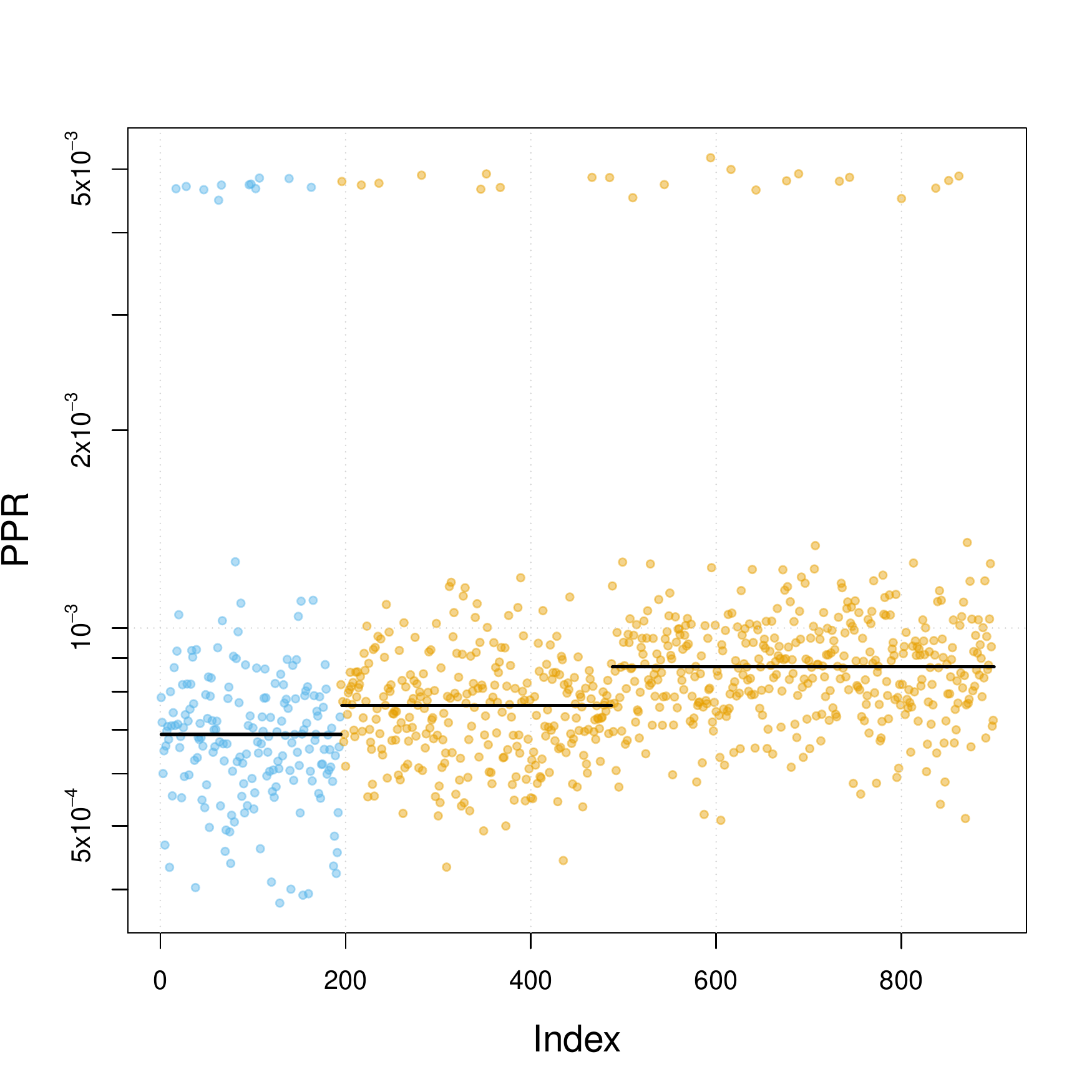} 
	\includegraphics[width=0.4 \textwidth]{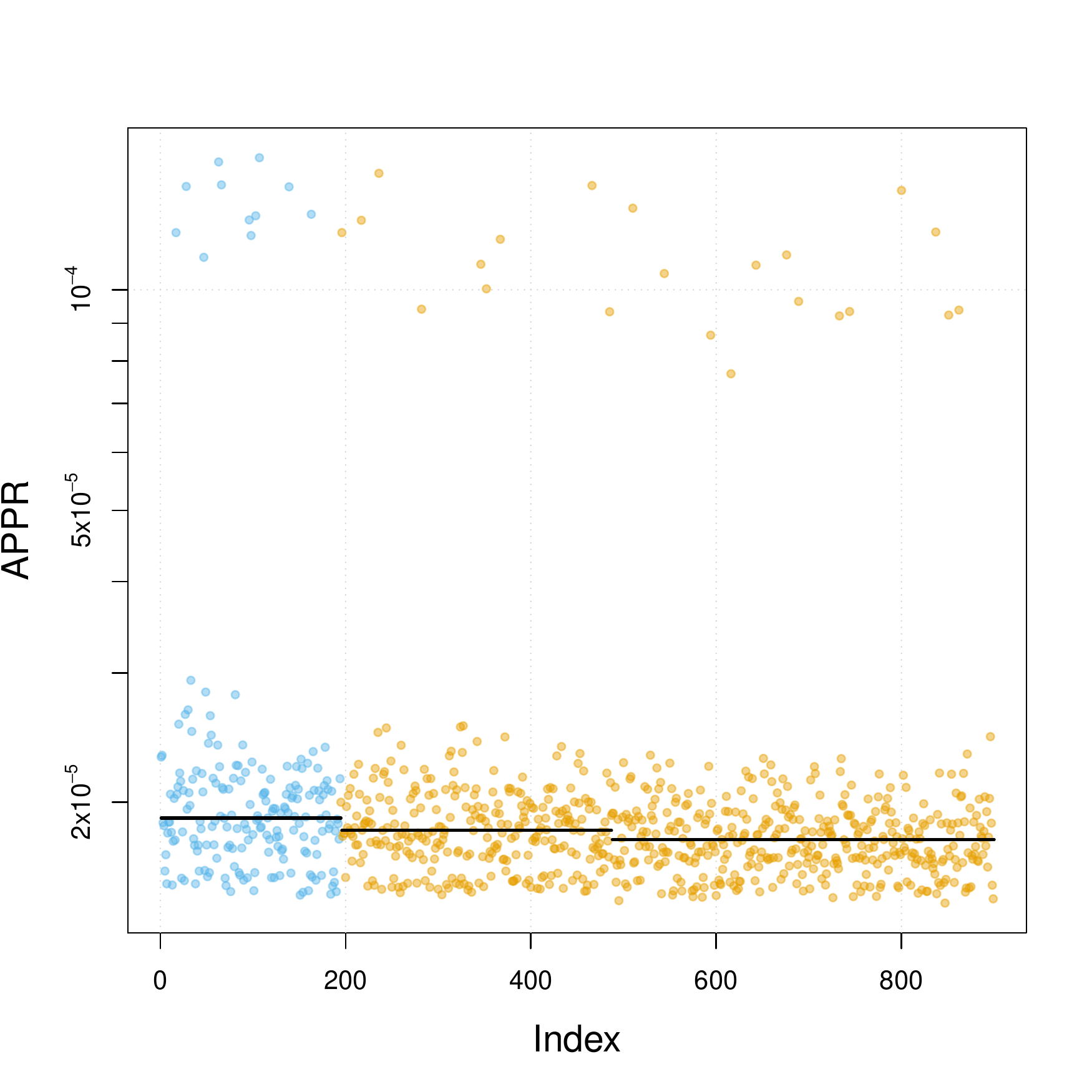} 
	\includegraphics[width=0.4 \textwidth]{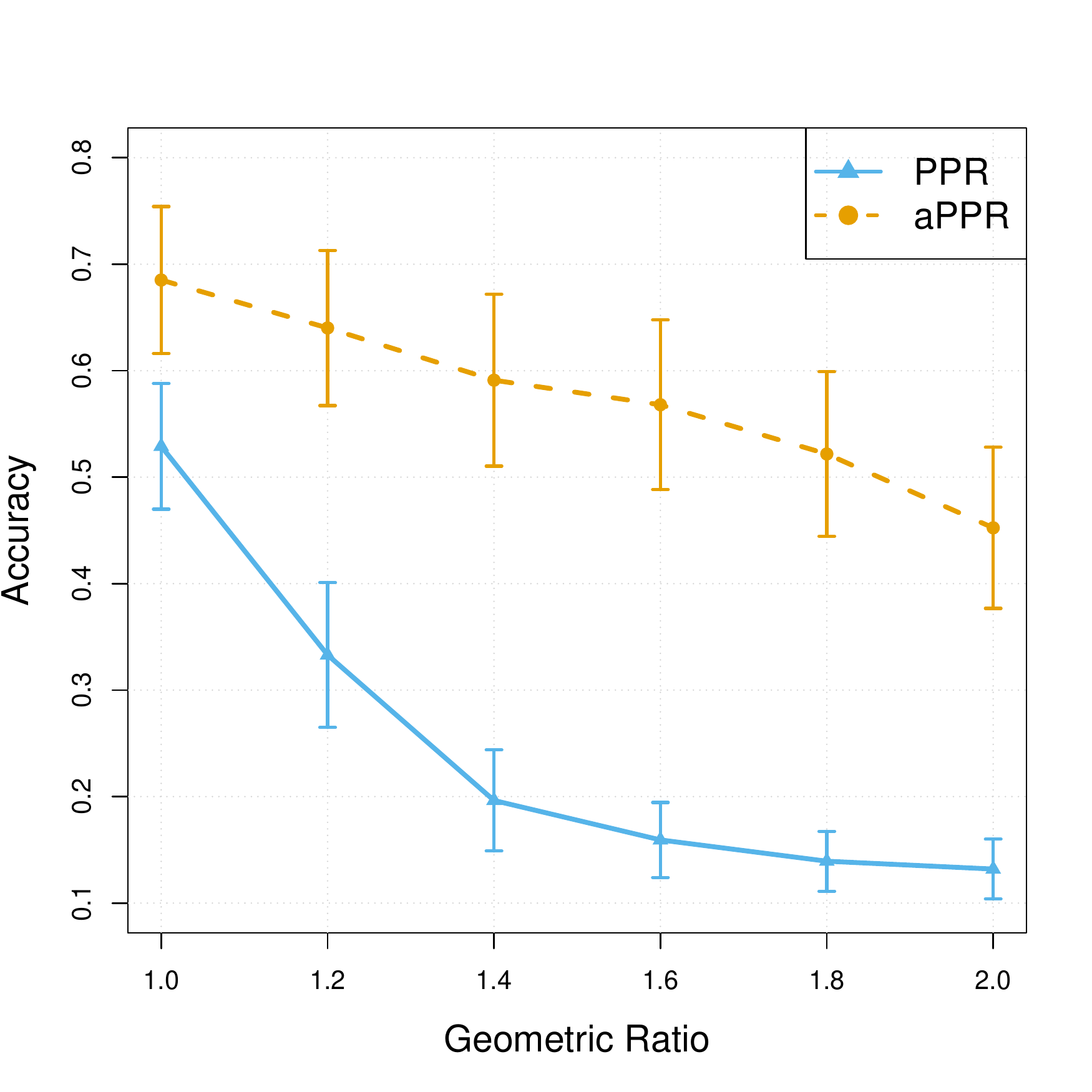}
	\includegraphics[width=0.4 \textwidth]{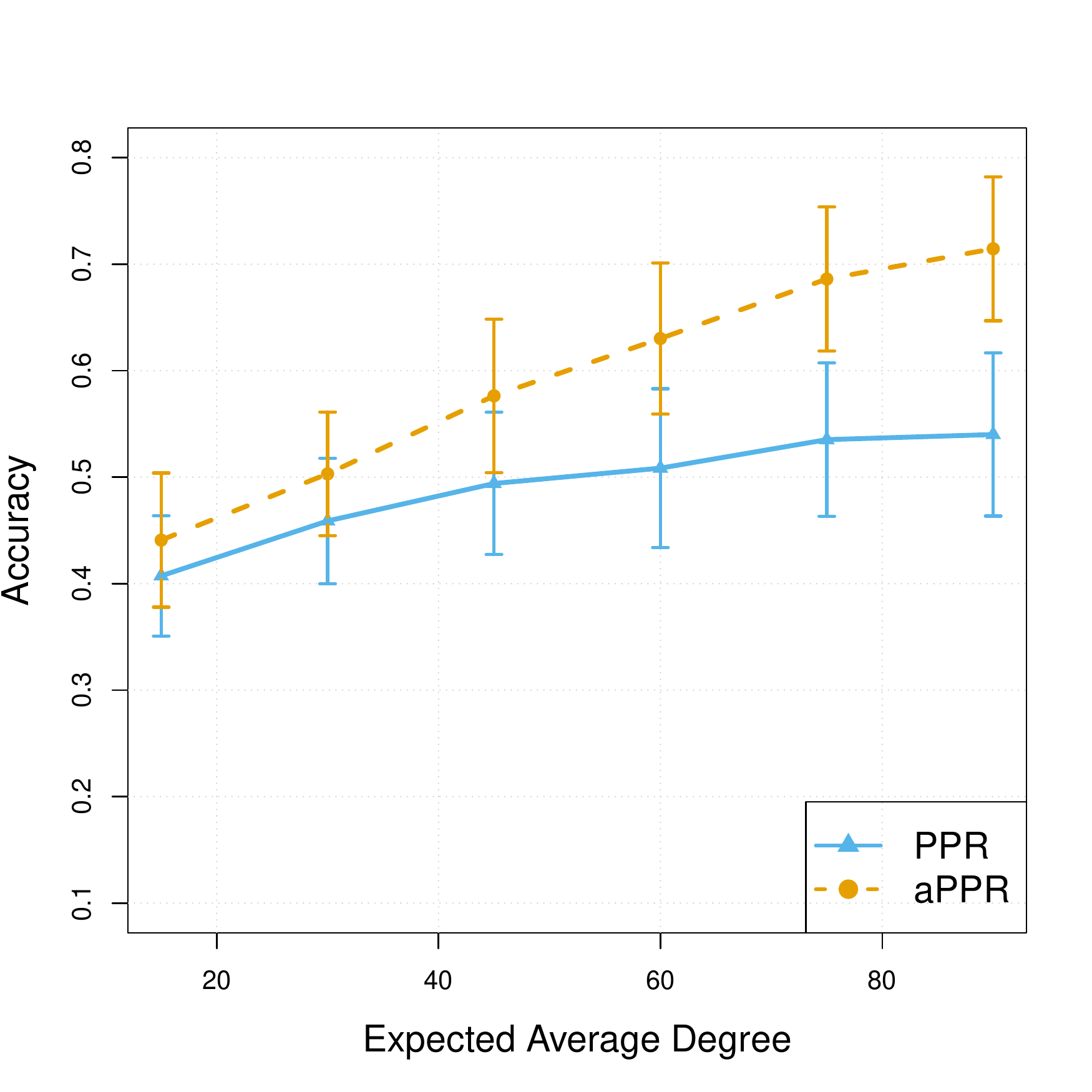}
	\caption{(Top) Simulated network generated from the classic SBM of 3 blocks with block degree heterogeneity. Three horizontal lines indicate the median of PPR and aPPR values within each cluster. 
		(Bottom Left) Comparison of performance for PPR (triangles with solid line) and aPPR (circles with dashed line) under the SBM with different levels of block degree heterogeneity. (Bottom Right) Comparison of performance for PPR and aPPR under the four-parameters SBM with different sparsity. Error bars are drawn using standard deviation.}
	\label{fig:simu2}
\end{figure}

\subsection{Experiment 3}
This experiment investigates the performance of PPR and aPPR under the SBM where there is no heterogeneity in the expected node degrees or block degrees. 
A number of random networks were sampled from the four-parameter stochastic block model, $\text{SBM}(K=3, N=900, b_1=0.6, b_2=0.2)$ \citep{rohe2011spectral}.
Under the four-parameter SBM, each of $K$ blocks has equal size in expectation, $N/K$, and the probability of a connection between two nodes is $b_2$ if they are in two separate blocks, or $b_1$ if in the same one.
%Here, the $\B$ matrix has the same structure as previously described, and the SNR is still three.
In addition, the expected average degree varies, $\delta \in \left\{15, 30, 45, 60, 75, 90\right\}$. 
For every setting, the results are averaged over 100 samples of the network. The PPR vector is calculated with one seeds randomly chosen from block one.
The bottom right panel of Figure \ref{fig:simu2} contrasts the accuracy of PPR and aPPR against six different values of expected average degree, showing that when the sampled graph has minimal degree heterogeneity, the adjusted PPR vector has only slightly higher accuracy than the PPR vector.

\section{A Sample of Twitter} \label{sxn:data}
In this section, we provide a more detailed case study to illustrate the properties of different PPR vectors. We obtain a local cluster of nodes around the seed node @NBCPolitics (NBC Politics) in the Twitter friendship graph. In the Twitter graph, the nodes are called handles or accounts (e.g. @NBCPolitics) and if Twitter handle $i$ follows Twitter handle $j$, then we define this as a directed edge $(i, j)$ pointing from $i$ to $j$. Affiliated with NBC news, NBC Politics specializes in political news coverage and has over 470k followers on Twitter (in-degree) and follows 145 handles (out-degree) as of December 2018. A brief look through @NBCPolitics' following list reveals that it follows a wide range of accounts, from TV programs, reporters and editors affiliated with NBC, to media accounts and journalists of other news outlets as well as politicians.

Data on following and handle profile information were collected through the Standard Twitter Search API. We queried the Twitter friendship graph starting from the seed node @NBCPolitics, using Algorithm \ref{alg:ppr_dir} with teleportation constant $\alpha=0.15$ and termination parameter $\epsilon=10^{-7}$, ending up with 5840 surrounding handles. Through this exercise, we intend to illustrate the properties and applications of local clustering using PPR, aPPR and rPPR vectors, where we set the regularization parameter $\tau$ to 100.

%	We observe a positive correlation between the PPR values and in-degrees (Figure \ref{fig:nbc} Left). In fact, the Kendall's rank correlation $\tau=0.365$ and the Spearman's rank correlation $\rho=0.536$, both differ significantly from zero (with p-values $<10^{-16}$ under the null), agreeing with the results in previous sections \citep{kendall1938new}.
%Eventually, the ranking of the handles is determined by PPR, regularized aPPR or aPPR values. 

\begin{table}[t]
	\caption{\label{tab:ppr} Top 30 handles of PPR with seed node @NBCPolitics and the teleportation constant $\alpha=0.15$ in December 2018.} 
	\centering
	\fbox{
		\begin{threeparttable}
			\scriptsize
			\begin{tabular}{rlHrl}
				%				\hline
				\noalign{\vskip .3mm}  
				& \textbf{Name} & \textbf{Friend} & \textbf{Followers} & \textbf{Description} \\ 
				\hline
				\noalign{\vskip .7mm}  
				1 & Melania Trump & TRUE & 11242283 & This account is run by the Office of First Lady Melania Trump... \\ 
				2 & The White House & TRUE & 17625630 & Welcome to @WhiteHouse! Follow for the latest from President... \\ 
				3 & Chuck Todd & TRUE & 2032038 & Moderator of @meetthepress and @nbcnews political director; ... \\ 
				4 & NBC News & TRUE & 6280551 & The leading source of global news and info for more than 75 ... \\ 
				5 & NBC Nightly News & TRUE & 962290 & Breaking news, in-depth reporting, context on news from ... \\ 
				6 & Andrea Mitchell & TRUE & 1737764 & NBC News Chief Foreign Affairs Correspondent/anchor, Andrea ... \\ 
				7 & Savannah Guthrie & TRUE & 881669 & Mom to Vale \& Charley, TODAY Co-Anchor, Georgetown Law. ... \\ 
				8 & Joe Scarborough & TRUE & 2521215 & With Malice Toward None \\ 
				9 & MSNBC & TRUE & 2261911 & The place for in-depth analysis, political commentary and ... \\ 
				10 & Rachel Maddow MSNBC & TRUE & 9498076 & I see political people... \\ 
				11 & Breaking News & TRUE & 9223158 &  \\ 
				12 & NBC News First Read & TRUE & 53847 & The first place for news and analysis from the @NBCNews Poli... \\ 
				13 & TODAY & TRUE & 4276453 & America's favorite morning show $|$ Snapchat: todayshow \\ 
				14 & Meet the Press & TRUE & 566713 & Meet the Press is the longest-running television show in history ... \\ 
				15 & The Wall Street Journal & TRUE & 16188842 & Breaking news and features from the WSJ. \\ 
				16 & Pete Williams & TRUE & 70062 & NBC News Justice Correspondent. Covers US Supreme Court, ... \\ 
				17 & Mark Murray & TRUE & 97571 & Mark Murray is the senior political editor for NBC News, ... \\ 
				18 & POLITICO & TRUE & 3695835 & Nobody knows politics like POLITICO. Got a news tip for us? ... \\ 
				19 & Katy Tur & TRUE & 587474 & MSNBC anchor @2pm, NBC News correspondent, author of NYT ... \\ 
				20 & Bill Clinton & TRUE & 10697521 & Founder, Clinton Foundation and 42nd President of the United... \\ 
				21 & Kasie Hunt & TRUE & 381704 & @NBCNews Capitol Hill Correspondent. Host, @KasieDC, Sundays... \\ 
				22 & TIME & TRUE & 15584815 & Breaking news and current events from around the globe. Host... \\ 
				23 & Kelly O'Donnell & TRUE & 195765 & White House Correspondent @NBCNews Veteran of Cap Hill \& ... \\ 
				24 & John McCain & TRUE & 3181773 & Memorial account for U.S. Senator John McCain, 1936-2018. To... \\ 
				25 & Peter Alexander & TRUE & 283522 & @NBCNews White House Correspondent / Weekend @TODAYshow ... \\ 
				26 & Hallie Jackson & TRUE & 359099 & Chief White House Correspondent / @NBCNews / @MSNBC Anchor ... \\ 
				27 & Kristen Welker & TRUE & 182244 & @NBCNews White House Correspondent. Links and retweets ... \\ 
				28 & Carrie Dann & TRUE & 37119 & .@NBCNews / @NBCPolitics. RTs not endorsements. \\ 
				29 & Willie Geist & TRUE & 807536 & Host @NBC \#SundayTODAY, Co-Host @Morning\_Joe, “Sunday ... \\ 
				30 & Morning Joe & TRUE & 563650 & Live tweet during the show! Links to must-read op-eds ... \\ 
				\hline
			\end{tabular}
			\begin{tablenotes}[flushleft]
				\footnotesize
				\item Through the PPR vector, the top 30 handles returned to @NBCPolitics include NBC's news related programs and celebrity reporters, comparable mainstream media outlets, as well as prominent political and public figures and institutions. Such results line up with its status as a mainstream political news source, demonstrating clustering effectiveness. Those Twitter handles tend to have millions of followers, showing the PPR vector's bias toward high in-degree.
%				\vspace{1mm}
			\end{tablenotes}
		\end{threeparttable}
	}
\end{table}

We first present the results of PPR. As Table \ref{tab:ppr} shows, the top 30 handles (except @NBCPolitics) with the highest PPR values are a combination of (i) NBC's news related programs such as NBC News, TODAY and Meet the Press; (ii) NBC's political reporters, anchors and editors, from well-known figures like Chuck Todd and Andrea Mitchell to less-known ones like Pete Williams (justice correspondent) and Mark Murray (senior political editor); (iii) other mainstream news outlets such as The Wall Street Journal, POLITICO, and TIME; and (iv) prominent public figures and politicians like Melania Trump, Bill Clinton and John McCain. In light of NBC's status as a mainstream news outlet and the political focus of @NBCPolitics, such results make sound sense. It must also be noted that all the top 30 handles are direct friends of @NBCPolitics's and have at least tens of thousands of followers. The median follower count is 1.4 million, suggesting high in-degrees. In fact, the pattern observed in the top 30 extends to the top 200 handles with the highest PPR values, which include NBC's own programs, journalists, editors and staff; fellow mainstream media outlets and their staff; and prominent public figures, politicians and government institutions (see Supplementary Materials \ref{sxn:top200}). The median in-degree of top 200 handles is around 184k, though there are four handles with less than one thousand followers. One important thing to notice is that among the top 200 handles, the first 139 are all directly followed by @NBCPolitics, with handles having high in-degrees generally ranked higher than those having low in-degrees (although @NBCPolitics follows 145 handles, 6 of them might have privacy protection that has prevented us from accessing their information). The remaining handles on the list, although not directly followed by @NBCPolitics, include five handles associated with NBC, from its news anchor Lester Holt to its News International President. However, the majority of those indirectly followed by @NBCPolitics are mainly high profile political and public figures (like President Trump, Vice President Pence, Hillary Clinton, and Stephen Colbert), government organizations (like WhiteHouse Office of Cabinet Affairs and National Security Council), and mainstream news outlets (like New York Times, CNN and AP) and well-known journalists (like John Dickerson and Anderson Cooper). We can thus conclude that the PPR vector is biased toward popular accounts followed directly by the seed node or indirectly by its friends, reflecting the popular Twitter handles followed by them. This property of the PPR vector can be harnessed by researchers interested in identifying the upstream of a handle, i.e., those Twitter elites who are followed by and might influence the seed node and by extension its followers.

\begin{table}[t]
	\caption{\label{tab:appr}Top 30 handles of aPPR with seed node @NBCPolitics and the teleportation constant $\alpha=0.15$ in December 2018.}
	\centering
	\fbox{
		\begin{threeparttable}
			\scriptsize
			\begin{tabular}{rlHrl}
				%					\hline
				\noalign{\vskip .3mm}  
				& \textbf{Name} & \textbf{Friend} & \textbf{Followers} & \textbf{Description} \\ 
				\hline
				\noalign{\vskip .7mm}  
				1 & Stephanie Palla & TRUE & 198 & Enroll America National Regional Director... \\ 
				2 & Jennifer Sizemore & TRUE & 386 &  \\ 
				3 & Alissa Swango & TRUE & 441 & Director of Digital Programming at @natgeo. All things ... \\ 
				4 & Making a Difference & TRUE & 670 & @NBCNightlyNews' popular feature profiles ordinary ... \\ 
				5 & Ron Whittemore & FALSE &   1 &  \\ 
				6 & Svante Stockselius & FALSE &   3 &  \\ 
				7 & Greg Martin & TRUE & 1161 & Political Booking Producer at @nbcnews @todayshow \\ 
				8 & Area Man & FALSE &   1 & I am Area Man. I pwn your news feed. \\ 
				9 & CELESTIA ROBINSON & FALSE &   2 &  \\ 
				10 & NBC Field Notes & TRUE & 1390 & NBC News correspondents and http://t.co/1eSopOQt8s ... \\ 
				11 & rob adams & FALSE &   2 &  \\ 
				12 & JL & FALSE &   2 &  \\ 
				13 & David Kelsey & FALSE &   1 &  \\ 
				14 & Hank Morris & FALSE &   1 &  \\ 
				15 & Jesse Marks & FALSE &   1 &  \\ 
				16 & Brayden Rainey & FALSE &   1 &  \\ 
				17 & child of the tiger & FALSE &   3 & yet another activist twitter, fighting all those fun... \\ 
				18 & Julie Swango & FALSE &   4 &  \\ 
				19 & Author Dianne Kube & FALSE &   7 & Dianne Kube is an Author with a passion, for family,... \\ 
				20 & Consider the Source & FALSE &   7 &  \\ 
				21 & Adam Edelman & TRUE & 2341 & Political reporter @nbcnews. Wisconsin native, ... \\ 
				22 & Phil McCausland & TRUE & 2519 & @NBCNews Digital reporter focused on the rural-urban... \\ 
				23 & Corky Siemaszko & TRUE & 2538 & Senior Writer at NBC News Digital (former NY Daily ... \\ 
				24 & Sam Petulla & TRUE & 2588 & Editor @cnnpolitics • Usually looking for datasets. ... \\ 
				25 & Ken Strickland & TRUE & 2693 & NBC News Washington Bureau Chief \\ 
				26 & Mike Mullen & FALSE &   7 &  \\ 
				27 & Elyse PG & TRUE & 2697 & White House producer @nbcnews $|$@USCAnnenberg alum ... \\ 
				28 & A. Johnson & FALSE &   2 & Change your thoughts \& you change your world. -Normal... \\ 
				29 & Steve Fenton & FALSE &   4 &  \\ 
				30 & Dobe Pitty Mami & FALSE &  13 &  \\ 
				\hline
			\end{tabular}
			\begin{tablenotes}[flushleft]
				\footnotesize
				\item Through the aPPR vector, the top 30 handles returned to @NBCPolitics include some relevant handles (NBC's news team and their counterparts in other mainstream news organizations) and many obscure ones (handles with few followers and no profile descriptions). This results from the aPPR vector's bias toward extreme low degree and introduces noise to the clustering results.
			\end{tablenotes}
%			\vspace{1mm}
		\end{threeparttable}
	}
\end{table}

In contrast, the aPPR vector up-weights handles that are much less popular (i.e., those with low in-degrees). As shown in Table 2, the 30 handles with the highest aPPR values include NBC's reporters, writers, editors, producers, and programs, all of whom have a few hundred to a few thousand followers. The 30 handles also include those unaffiliated with NBC, such as director of a non-profit (Enroll America), director of digital programming at National Geographic, and @CNNPolitics' editor. All of them are professionally related to the seed node. This testifies to the applicability of aPPR for locating an idiosyncratic local cluster around a seed node. However, more than half (17) of the 30 handles are obscure and not directly followed by @NBCPolitics. The reason they appear on the list is probably that they have just one and at most a dozen followers (recall that aPPR divides by in-degree). In fact, 160 of the top 200 handles are not direct friends of @NBCPolitics; the median in-degree of the top 200 handles is merely 8 (Supplementary Materials \ref{sxn:top200}). Those handles might have ended up on the list due to a combination of luck and, more importantly, their extremely low in-degrees. In this regard, “noise” can be introduced by the aPPR vector because it prioritizes handles with extremely low in-degrees that are possibly several degrees separated from the seed node. 

\begin{table}[t]
\caption{\label{tab:rppr}Top 30 handles of rPPR with seed node @NBCPolitics and the teleportation constant $\alpha=0.15$ in December 2018.} 
\centering
\fbox{
\begin{threeparttable}
\scriptsize		
\begin{tabular}{rlHrl}
	%					\hline
	\noalign{\vskip .3mm}  
	& \textbf{Name} & \textbf{Friend} & \textbf{Followers} & \textbf{Description} \\ 
	\hline
	\noalign{\vskip .7mm}  
	1 & Stephanie Palla & TRUE & 198 & Enroll America National Regional Director http://t.co/X6jJIE... \\ 
	2 & Jennifer Sizemore & TRUE & 386 &  \\ 
	3 & Alissa Swango & TRUE & 441 & Director of Digital Programming at @natgeo. All things food.... \\ 
	4 & Making a Difference & TRUE & 670 & @NBCNightlyNews' popular feature profiles ordinary people do... \\ 
	5 & Greg Martin & TRUE & 1161 & Political Booking Producer at @nbcnews @todayshow \\ 
	6 & NBC Field Notes & TRUE & 1390 & NBC News correspondents and http://t.co/1eSopOQt8s reporters... \\ 
	7 & Adam Edelman & TRUE & 2341 & Political reporter @nbcnews. Wisconsin native, Bestchester ... \\ 
	8 & Phil McCausland & TRUE & 2519 & @NBCNews Digital reporter focused on the rural-urban divide.... \\ 
	9 & Corky Siemaszko & TRUE & 2538 & Senior Writer at NBC News Digital (former NY Daily News ... \\ 
	10 & Sam Petulla & TRUE & 2588 & Editor @cnnpolitics • Usually looking for datasets. You can ... \\ 
	11 & Ken Strickland & TRUE & 2693 & NBC News Washington Bureau Chief \\ 
	12 & Elyse PG & TRUE & 2697 & White House producer @nbcnews $|$@USCAnnenberg alum $|$ LA kid ... \\ 
	13 & Hasani Gittens & TRUE & 3002 & Level 29 Mage. Senior News Ed. @NBCNews. Sheriff of Nattahna... \\ 
	14 & Scott Foster & TRUE & 3464 & Senior Producer, Washington @NBCNEWS @TODAYshow \\ 
	15 & Zach Haberman & TRUE & 3693 & Lead Breaking News Editor, @NBCNews. Previously had other jobs... \\ 
	16 & Emmanuelle Saliba & TRUE & 4004 & Head of Social Media Strategy @Euronews $|$ Launched \#THECUBE ... \\ 
	17 & Alex Johnson & TRUE & 4371 & News, data and analysis for @NBCNews; data geek; ... \\ 
	18 & Savannah Sellers & TRUE & 4637 & News junkie. Host of NBC's "Stay Tuned" on Snapchat. Storyte... \\ 
	19 & NYC Clothing Bank & FALSE & 154 & We distribute new, never-worn clothing and merchandise... \\ 
	20 & Shaquille Brewster & TRUE & 5362 & @NBCNews Producer/Politics $|$ @HowardU Alum$|$ Journalist $|$ Pol... \\ 
	21 & Joey Scarborough & TRUE & 6277 & NBC News Social Media Editor. New York Daily News Alum. RTs ... \\ 
	22 & Jane C. Timm & TRUE & 6478 & @nbcnews political reporter and fact checker. More fun than ... \\ 
	23 & Anthony Terrell & TRUE & 6827 & Emmy Award winning journalist. Political observer. Covered ... \\ 
	24 & NBC News Videos & TRUE & 7838 & The latest video from http://t.co/xPyvMOTEF6 \\ 
	25 & Libby Leist & TRUE & 7946 & Executive Producer @todayshow \\ 
	26 & Voices United & FALSE & 310 & Voices United is a non profit educational organization ... \\ 
	27 & Social Headlines & FALSE & 344 & Daily roundup of top social media and networking stories. \\ 
	28 & James Miklaszewski & FALSE & 337 & Writer, Photographer, Editor, Director, Producer, Newshound ... \\ 
	29 & Courtney Kube & TRUE & 9494 & NBC News National Security \& Military Reporter... \\ 
	30 & Bob Corker & TRUE & 10042 & Serving Tennesseans in the U.S. Senate \\ 
	\hline
\end{tabular}
\begin{tablenotes}[flushleft]
	\footnotesize
	\item Through the rPPR vector, the top 30 handles returned to @NBCPolitics include much fewer low in-degree and obscure ones and many more moderately connected nodes that are relevant to @NBCPolitics, including its reporters and editors and media professionals from other organizations.
%	\vspace{1mm}
\end{tablenotes}
\end{threeparttable}
}
\end{table}

To reduce noise, we applied a regularization step to the aPPR vector to remove those ``distant''  and small nodes while preserving the close and relevant ones. In Table 3, the majority of the top 30 handles with the highest regularized aPPR (i.e., rPPR) values have three- or four-digit numbers of followers. Similar to the aPPR results, they include NBC's news crew. But the difference is that the overwhelming majority (18) of the top 30 handles work at NBC. Some handles who work for other news organizations (e.g., Sam Petulla at @cnnpolitics and Emmanuelle Saliba at @Euronews) might have previously worked at NBC or have close connection with its news team. Even the four handles that are not directly followed by @NBCPolitics are interesting -- they are non-profit organizations (NYC Clothing Bank and Voices United) and news-related individual or organization (James Miklaszewski and Social Headlines). This pattern can also be observed in the top 200 handles, 72 of whom are directly followed by @NBCPolitics. The overwhelming majority of those directly followed by it are affiliated with NBC, comprising its day-to-day news team, who enjoy much less publicity than the celebrity reporters. The remaining 128 of them, who are not directly followed by @NBCPolitics, actually also include 20 NBC's journalists and staff, such as Ray Farmer (NBC News photographer) and Jim Miklaszewski (chief Pentagon correspondent for NBC News). Others are non-profits like Vets Helping Heroes and professionals from other news organizations or companies such as WSJ, NFL Network, and Microsoft, who might have worked for NBC or have close connection with it. Although there still appear to be obscure handles with few followers, they decrease significantly in number -- the median in-degree of the top 200 handles is 340 (Supplementary Materials \ref{sxn:top200}), a precipitous drop from that of the top PPR handles yet not too small as compared to that of the top aPPR handles. We thus conclude that the regularized aPPR vector returns a local cluster with little noise, reflecting a seed node's close circles, either directly or indirectly related. 

In order to evaluate the influence of the desired cluster size $n$ on the results based on different PPR vectors, we compare the local clusters of PPR, aPPR, and rPPR by varying sample size. 
Define the \textit{in-and-out ratio} of local cluster $C\subset V$ as the proportion of edges inside $C$ among all edges connected to $C$,
$$\frac{2\times\sum_{u,v\in C}A_{uv}}{\sum_{u\in C} \din_u+\dout_u}.$$
A higher in-and-out ratio indicates a more internally connected sample.
Figure \ref{fig:nbc} (Right) shows the effectiveness of aPPR and rPPR in producing a compact local cluster. 
When the sample size is bigger than 100, the connectedness of the local cluster produced by rPPR stabilizes; the greater the sample size, the more densely connected a cluster aPPR would produce. 
However, PPR is easily susceptible to the inclusion of popular nodes. 
In this case, a sharp drop of in-and-out ratio for PPR when the sample size reaches around 140 is caused by inclusions of highly popular accounts @POTUS (President Trump) and @realDonaldTrump (Donald J. Trump). 

The PPR clustering is fairly robust to the choice of teleportation constant, despite the size of local cluster. 
To illustrate this, we also performed the same pipeline of analysis with the seed @NBCPolitics while varying the value of $\alpha$ (e.g., $0.05$, $0.25$, and $1/3$) in parallel.
We observed that those local clusters returned by Algorithm \ref{alg:lc_dir} all share a great portion of members in common.
For example, there are 280 (93.3\%) overlapping members between two targeted samples of size $n=300$, using $\alpha=0.15$ and $0.25$ respectively. 
These suggest a low sensitivity to the teleportation constant (see Supplementary Materials \ref{sxn:parameter}).

The left panel of Figure \ref{fig:nbc} depicts the behaviors of PPR, aPPR and rPPR.  
%Each dot is a node that was sampled.  The axes are node degree and the node's value in the PPR vector.  The top 200 accounts from PPR, aPPR and rPPR are all represented. 
Each handle queried in this sampling is displayed as a dot, with y-axis representing the PPR value and x-axis the number of followers (i.e., in-degree). Top handles with the highest PPR values are above blue dashed line, which tend to concentrate on the right end of the x-axis and thus are biased toward high in-degrees. Top handles with highest aPPR values are dots to the left of the yellow dotdash line, which gather on the left end of the x-axis and thus in favor of low in-degrees. Regularized aPPR, by purple dots, excludes the very low degree nodes and very high degree nodes.
As the empirical results show, these three vectors can be thought of as lenses through which we view the local structure of a given Twitter handle with varying foci, rendering high, moderate, and low in-degree blocks and serving different needs and purposes.

\begin{figure}[!t]
	\centering
	\includegraphics[width= .49\textwidth]{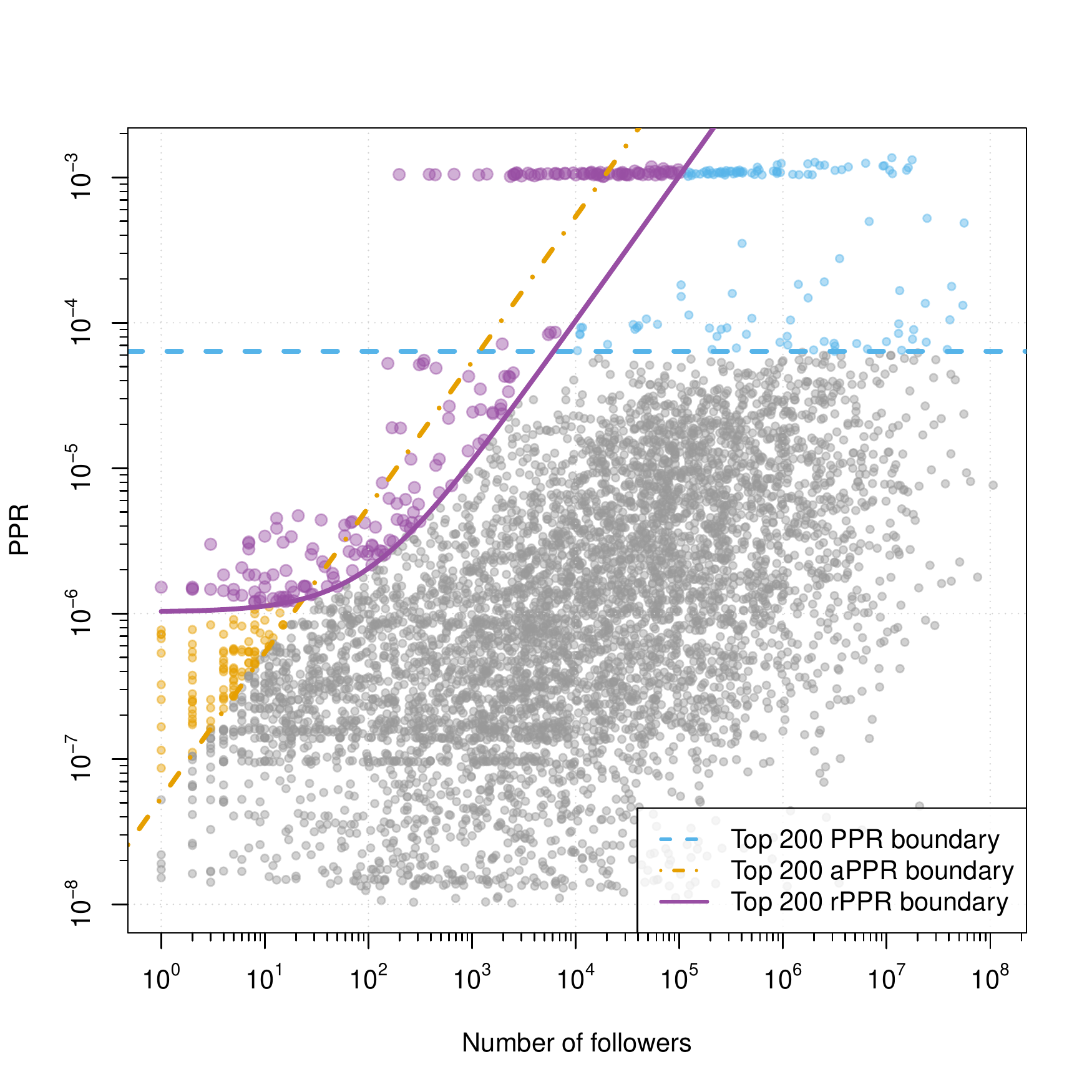}
	\includegraphics[width= .49\textwidth]{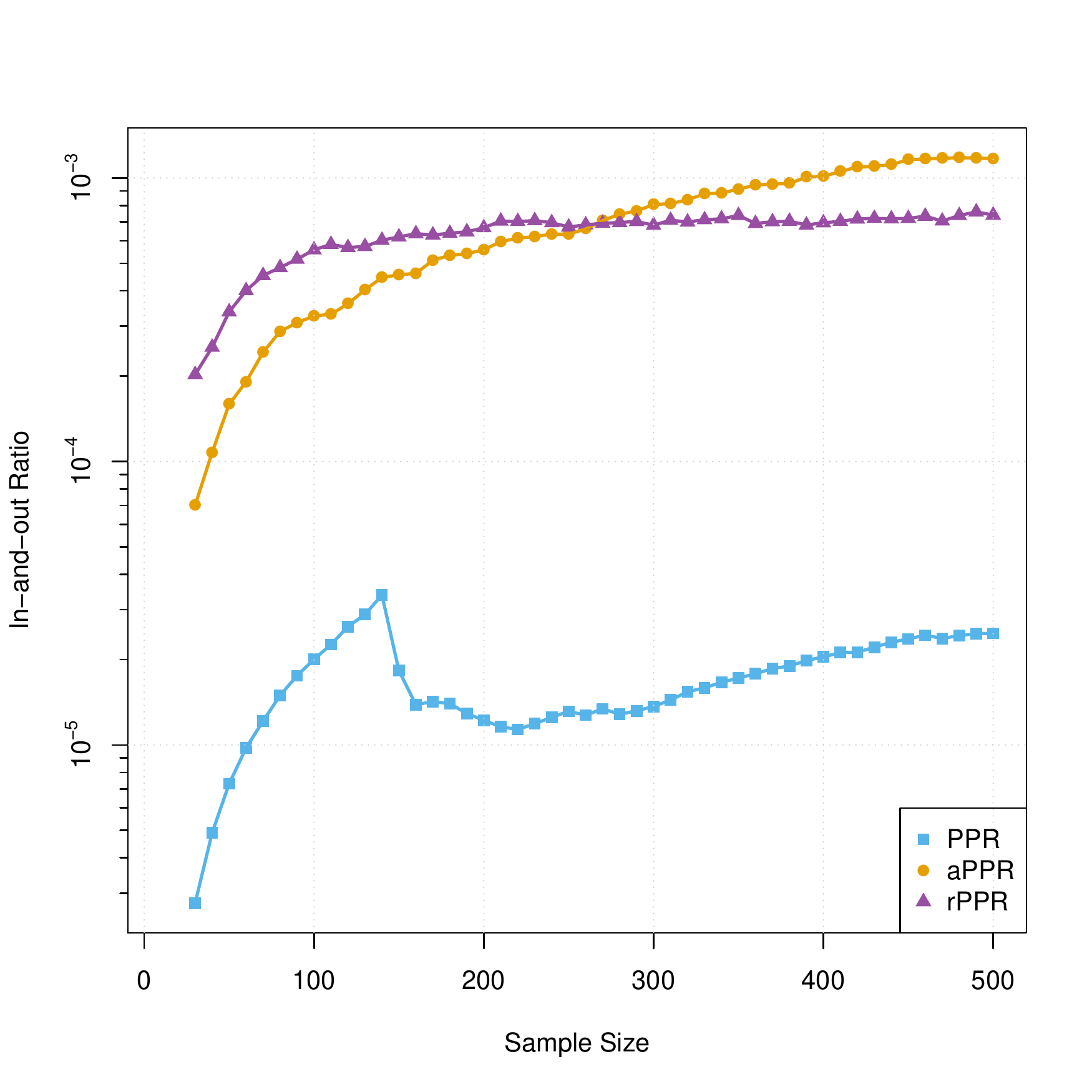}
	\caption{Left: an illustration of 5840 Twitter handles examined by Algorithm \ref{alg:ppr_dir} and three samples of size 200 by PPR, aPPR, and rPPR. Each dot represents a user in Twitter. The blue dashed line delimits the top 200 handles by PPR vector; vertices above the line are PPR's sample. Similarly, the yellow dotdash line determines the sample returned by Algorithm \ref{alg:lc_dir} given $n=200$; vertices above this boundary correspond to aPPR's sample. In particular, dots in purple stand for the sample of rPPR; the purple solid line shows the boundary of this sample. Right: The in-and-out ratio of local clusters identified by PPR, aPPR, and rPPR, as the sample sizes vary. A higher in-and-out ratio indicates a more internally connected cluster.}
	\label{fig:nbc}
\end{figure}

\section{Discussion}

This paper studies the PPR vector under the degree-corrected stochastic block model and PPR clustering in massive block model graphs. 
We establish some consistency results for this method, and examine its performance through analysis of Twitter friendship graph. 
As shown in the results, the PPR vectors with and without adjustment have distinct properties and can be used to effectively sample a massive graph for various purposes. However, there are limitations worthy of future investigations. 

In Section \ref{sxn:pop}, we provide a representation of the PPR vector under the DC-SBM and its extension into directed graphs. 
The result does not impose extra structural restrictions on the model parameters, except that $\B$ corresponds to a strongly connected ``block-wise'' graph. 
We consider a positive definite connectivity matrix particularly so that it is intuitive to conceive the notion of local cluster.
In practice (and many of our experiments, see Supplementary Materials \ref{sxn:parameter}), however, a PPR-type algorithm appears to continue working for a broader range of $\B$ (e.g., singular or indefinite), provided that the teleportation constant is sufficiently large (e.g. $\alpha>0.1$). 
It is unclear yet what is the minimum constraint needed on $\B$ in order for the PPR clustering to function. 
In addition, DC-SBM does have its limits. For example, the model fails to capture either mixed block membership or popularity features which are potentially informative in real world networks. 
The behavior of a PPR vector under other extensions of stochastic block model, such as mixed membership stochastic block model and popularity-adjusted block model, remains unknown \citep{airoldi2008mixed, sengupta2018block}.
Future studies on the PPR vector under these models could shed further light on the PPR clustering and offer more practical guidelines on their application.

%	In the discussion \ref{rmk3} in Section \ref{sxn:pop}, if we further assume that the ``block-wise'' PPR vector has several distinguishable entries in addition to local cluster, then one can obtain the identification of nearby local clusters surrounding the block to which the seed node belongs. 
%	Furthermore, if the ``block-wise'' PPR vector is entrywise distinguishable, a PageRank method is then ideal for an ordinary (global) clustering task. 
%	An investigation of when these conditions hold would be useful as it leads to a broader application of the PPR clustering. 
%	It is worth noting that for this purpose, the preference vector is less important while the teleportation constant is more critical, as it plays an essential role in modulating the adjusted PPR vector.

In Section \ref{sxn:spl}, we proved the consistency of the PPR clustering, requiring the average expected node degree to grow in order of $\logN$, which hits the boundary between the theoretical guarantees and the realistic observation. 
In contrast, scale-free networks such as the preferential attachment model \citep{barabasi1999emergence} have finite expected node degrees. 
Future investigations into variants of PPR that could possibly overcome this limitation yet ensure a fine local cluster discovery would be particularly interesting and useful.

In Section \ref{sxn:data}, we introduce the regularized version of adjusted PPR (rPPR) vector, with a series of empirical evidence showing its efficacy in targeted sampling. 
While the results appear promising, theoretical guarantees for this technique remain unexplored.
In order for some mathematical analyses, one may resort to the techniques used in \citet{le2016optimization}.
It is previously shown that the regularized graph Laplacian (or transition matrix) enjoys ``nice'' finite sample properties, which facilitate the consistency of many regularized spectral methods. 
It thus is reasonable conjecture that rPPR vectors are also suitable for local clustering.

An \textsf{R} implementation of the PPR clustering is available at author's GitHub (\url{https://github.com/RoheLab/aPPR}).

\section*{Acknowledgments}
This research is supported by NSF Grants DMS-1612456 and DMS-1916378 and ARO Grant W911NF-15-1-0423. 
Thank you to Yuling Yan and E. Auden Krauska for the helpful comments.
Thank you to Alex Hayes for kindly advising on the software development.

%\newpage
\appendix

\section{Technical Proofs}\label{apd:proof}

\subsection{Proof of Proposition \ref{ppr_sol}}

\begin{proof} %[of Proposition \ref{ppr_sol}]
We apply Perron-Frobenius theorem for the first part \citep{perron1907theorie, frobenius1912matrizen}, and complete the proof by construction.
\begin{enumerate}[(a)]
	\item First, notice that $Q$ is a Markov transition matrix by modifying $G=(V,E)$ a little. To this end, (i) shrink the weights of every existing edge by factor $1-\alpha$, and (ii) add an edge weighted $\alpha$ between seed node $v_0$ and all nodes in the graph. Then $Q$ represents the new graph $G'(V, E')$, which is strongly connected by construction. Hence $Q$ is irreducible.
	
	The PPR vector $p$ is all-positive. 
	To see this, note that the equation $p\T=p\T Q$ implies that $p$ is a stationary distribution for the standard random walk on $G'$.
	Since $G'$ is strongly connected, it follows that the stationary distribution must be all-positive.
	
	From the Perron-Frobenius theorem, the only all-positive eigenvector of a non-negative irreducible matrix is associated with the leading eigenvalue, which is 1 in our case.
	Since the leading eigenvalue of non-negative irreducible matrix is simple, we conclude that $p$ is unique.
	\item We finish the proof by constructing an explicit form of the PPR vector. Let 
	$R_\alpha=\alpha\sum_{s=0}^\infty(1-\alpha)^sP^s.$
	The infinite sum converges for $\alpha\in(0,1]$. Then, $p=R_\alpha\T \pi$ satisfies the definition of personalized PageRank vector,
	\begin{eqnarray*}
		\alpha \pi\T+(1-\alpha) \pi \T R_\alpha P &=&\alpha \pi\T+(1-\alpha)\pi\T\left(\alpha\sum_{s=0}^\infty(1-\alpha)^sP^s\right)P\\
		&=&\alpha \pi\T+\alpha\sum_{s=1}^\infty(1-\alpha)^s\pi\T P^s\\
		&=&\pi\T R_\alpha.
	\end{eqnarray*}
\end{enumerate}

Since the solution is unique, we have $p= R_\alpha\T\pi$.
\end{proof}

\subsection{Proof of Proposition \ref{prop:approx}}

\begin{proof}
	Algorithm \ref{alg:ppr} maintains two vectors, $p^\epsilon$ and $r$, by transporting probability mass from $r$ to $p^\epsilon$ at each updating step.
	Note that the termination criterion implies that $r_u<\epsilon d_{u}$ for any $u$ sampled, thus it suffices to prove that 
	$$\left|p_u - p^\epsilon_u\right| \le r_u.$$
	
	For a fixed $\alpha$, let $p(x)$ be the PPR vector with preference vector $x\in\R^N$ satisfying $x_i\geq0$ and $\|x\|_1\le1$. Then $p(\pi)$ is the exact PPR vector as in Equation (\ref{eqn:ppr_sum}). 
	Since $p(x)\T P=p(x\T P)$, we have \citep{jeh2003scaling} 
	\begin{equation}\label{eqn:ppr_gen}
	p(x)=\alpha x+(1-\alpha)p(P\T x).
	\end{equation}
	%$p^\epsilon(\pi)$ (or simply $p^\epsilon$ when it is clear) as the approximate PPR vector in Algorithm \ref{alg:ppr}. 
	
	We argue that $p^\epsilon + p(r)$ is invariant in updating steps.
	%		\begin{equation} \label{eqn:invariant}
	%		\end{equation}
	To see this, suppose $(p^\epsilon)'$ and $r'$ are the results of performing one update on $p^\epsilon$ and $r$ after sampling node $u$.
	We have 
	\begin{eqnarray*}
		(p^\epsilon)'&=&p^\epsilon +\alpha r_ue_u,\\
		r'&=&r-r_ue_u+(1-\alpha)r_uP\T e_u.
	\end{eqnarray*}
	where $e_u$ is the unit vector on the direction of $u$.
	Then, 
	\begin{eqnarray*}
		p(r)&=&p(r-r_ue_u)+p(r_ue_u)\\
		&\overset{\text{(i)}}{=}&p(r-r_ue_u) + \alpha r_ue_u+(1-\alpha)p\left(r_u P\T e_u\right)\\
		&\overset{\text{(ii)}}{=}&p\left(r-r_ue_u + (1-\alpha)r_uP\T e_u\right)+ \alpha r_ue_u\\
		&=&p(r')+\left(p^\epsilon\right)'-p^\epsilon,
	\end{eqnarray*}
	where (i) is applying Equation (\ref{eqn:ppr_gen}) at $x=r_ue_u$ and (ii) comes from the linearity of PPR vector in the preference vector.
	
	The desired result follows from recognizing that $p^\epsilon + p(r)$ is initially 
	$\vec{0} + p(\pi)$
	and that when the algorithm terminates, $\left[p(r)\right]_u\leq r_u$ for any sampled $u$.
\end{proof}

\textsc{Remark.}
If $\epsilon d_1>1$, Algorithm \ref{alg:ppr} terminates after the first round and simply output $p=\vec{0}$. 
Under this circumstance, Proposition \ref{prop:approx} still holds, because $\left|p_u - p^\epsilon_u\right|\le \left|p_u\right|+ \left|p^\epsilon_u\right|\le 1$.

%\textsc{Remark.} \citet{andersen2006local} shows that using random walk transition $P$ with teleportation constant $\alpha$ is equivalent to using lazy random walk transition $W=\left(I+P\right)/2$ with teleportation constant $\alpha'$, provided $$\alpha'=\frac{\alpha}{2-\alpha}.$$

\subsection{Lemmas for the DC-SBM} \label{apd:lemmas}
\begin{lemma}[Properties of the DC-SBM] \label{lem:basic}
	\underddcsbm 
	\begin{enumerate}[\normalfont(a)]
		%		\item $Z\T \Theta Z=I_K$, \label{ZThetaZ}
		%		\item  $\tD=Z\T \D Z$, \label{DB}
		%		\item $\cd_{v}=\theta_v\td_{z(v)}$. \label{D_DB}
		\item $\tDi=Z\T \Di Z$, and $\tDo=Z\T \Do Z$ \label{DB}, and
		\item $\di_v=\thi_v\tdi_{z(v)}$, and $\dou_v=\tho_v\tdo_{z(v)}$. \label{D_DB}
	\end{enumerate}
\end{lemma}

\begin{proof}
	\ref{DB} is an alternative way of writing the definition. For \ref{D_DB}, we prove the first equation. Recall that for any $i$, $\sum_{u:z(u)=i}\tho_u=1$, then by definition,
	$$\di_{v}=\sum_u \tho_u\thi_v B_{z(u)z(v)}=\thi_v\sum_{j=1}^K\left(\B_{jz(v)}\sum_{u:z(u)=j} \tho_u\right)=\thi_v\tdi_{z(v)}.$$
\end{proof}

\textsc{Remark.}
%\begin{remark}
Since $Z\T \Ti Z=I_K$, \ref{DB} implies $\left[\Di\right]^{-1}\Ti Z=Z\left[\tDi\right]^{-1}$.
%\end{remark}

%	\subsection{Proof of Lemma \ref{explicit1}} 

\begin{lemma} [Explicit form of $\cP$ and its powers] \label{explicit1}
	\underddcsbm the population graph transition is the product
	$$\cP= Z\tP Z\T\Ti.$$ 
	and its matrix powers are
	$$\cP^k= Z \tP^k Z\T\Ti.$$
\end{lemma}

\begin{proof}% [Proof of Lemma \ref{explicit1}] 
	By definition and Lemma \ref{lem:basic}\ref{D_DB}, for any $u, v\in V$, 
	\begin{equation*}
	\cP_{uv}=\left(\tho_u\tdo_{z(u)}\right)^{-1}\tho_u\thi_v\B_{z(u)z(v)}
	=\thi_v\B_{z(u)z(v)}/\tdo_{z(u)}=\thi_v\tP_{z(u)z(v)}.
	\end{equation*}
	For the powers of $\cP$, noticing that $Z\T \Ti Z=I_K$, 
	\begin{equation*}
	\cP^2=Z\tP Z\T \Ti Z\tP Z\T\Ti =Z\tP^2Z\T\Ti.
	\end{equation*}
	The desired result follows from the principle of induction on $k$-th power.
\end{proof}

\subsection{Proof of Theorem \ref{explicit2}} \label{pf:explicit2}
\begin{proof}
	By Proposition \ref{ppr_sol} and Lemma \ref{explicit1}, we have
	%	\begin{eqnarray*}
	%		\p&=&\alpha\sum_{s=0}^\infty(1-\alpha)^s\cP^s\pi\\
	%		&=&\alpha\sum_{s=0}^\infty(1-\alpha)^s\Theta Z\tP^sZ\T \pi\\
	%		&=&\Theta Z\left(\alpha\sum_{s=0}^\infty(1-\alpha)^s \tP^s\boldsymbol{\pi}\right)\\
	%		&=&\Theta Z \tp.
	%	\end{eqnarray*}
	\begin{eqnarray*}
		\p&=&\alpha\sum_{s=0}^\infty(1-\alpha)^s\left(\cP^s\right)\T\pi\\
		&=&\alpha\sum_{s=0}^\infty(1-\alpha)^s\Ti Z\left(\tP^s\right)\T Z\T \pi\\
		&=&\Ti Z\left(\alpha\sum_{s=0}^\infty(1-\alpha)^s \left(\tP^s\right)\T\boldsymbol{\pi}\right)\\
		&=&\Ti Z \tp,
	\end{eqnarray*}
	In addition, it follows from Lemma \ref{lem:basic}\ref{DB} that 
	%	$$\p^*=\D^{-1}\Theta Z \tp=Z\tD^{-1}\tp=Z\tp^*,$$
	$$\p^*=\left[\Di\right]^{-1}\p=\left[\Di\right]^{-1}\Ti Z \tp=Z\left[\tDi\right]^{-1}\tp=Z\tp^*.$$
	This completes the proof.
\end{proof}

\subsection{Proof of Lemma \ref{lem:top}} \label{pf:top}
\begin{proof} %[of Lemma \ref{lem:top}]
	
	For any $\alpha>0$, the PPR vector with seed node $v_0=1$ is the solution to the equation $\p\T=\p\T\Q$, where $\Q=\alpha\Pi+(1-\alpha)\cP$. Define a sequence of probability distribution $\p^s\in\R^N$ such that $\p^s= \left(\Q^s\right)\T \p^0$, where $\p^0$ is an arbitrary initial probability distribution. Then, $\lim_{s\rightarrow\infty}\p^s=\p$. For simplicity, we assume $\p^0$ is close to $\p$, that is, for any $\varepsilon>0$ and $s\ge0$,
	\begin{equation} \label{eqn:concentrate_ps}
	\|\p^s-\p\|_\infty<\varepsilon/2.
	\end{equation}
	This can be achieved by finding an integer $S(\varepsilon)$ large enough and setting $\p^0=\p^{S}$.
	
	We first claim that 
	\begin{equation} \label{eqn:decreasing}
	\max_{u\neq 1}\frac{\p^{s+1}_u}{\cd_u} \le  (1-\alpha)\max_{u\in V} \frac{\p^s_u}{\cd_u}.
	\end{equation}
	In fact, for any $u\neq1$,
	%$$	\p^{s+1}_u=\alpha\1_{\{u=1\}}+(1-\alpha)\sum_{v\in V}\frac{\A_{vu}}{\cd_v}\p^s_v\le(1-\alpha)\left(\sum_{v\in V}\A_{vu}\right)\max_{v\in V} \frac{\p^s_v}{\cd_v}=(1-\alpha)\cd_u\max_{v\in V} \frac{\p^s_v}{\cd_v}.$$
	\begin{eqnarray*}
		\p^{s+1}_u&=&\alpha\1_{\{u=1\}}+(1-\alpha)\sum_{v\in V}\frac{\A_{vu}}{\cd_v}\p^s_v\\
		&\le&(1-\alpha)\left(\sum_{v\in V}\A_{vu}\right)\max_{v\in V} \frac{\p^s_v}{\cd_v}\\
		&=&(1-\alpha)\cd_u\max_{v\in V} \frac{\p^s_v}{\cd_v}.
	\end{eqnarray*}
	
	We then show $\frac{\p^s_1}{\cd_1}>\frac{\p^s_v}{\cd_v}$ for any $v\neq1$ by contradiction. 
	Suppose otherwise that $\frac{\p^s_1}{\cd_1}\le\max_{u\neq1}\frac{\p^s_u}{\cd_u}$, then Equation (\ref{eqn:concentrate_ps}) implies for any $s'$,
	\begin{eqnarray*}
		\frac{\p^{s'}_1}{\cd_1}\le\frac{\p^s_1+\varepsilon}{\cd_1}\le\max_{u\neq 1}\frac{\p^s_u}{\cd_u}+\frac{\varepsilon}{\cd_1}\le \max_{u\neq 1}\frac{\p^{s'}_u+\varepsilon}{\cd_u}+\frac{\varepsilon}{\cd_1}\le \max_{u\neq 1}\frac{\p^{s'}_u}{\cd_u}+\frac{2\varepsilon}{\cd_{\min}},
	\end{eqnarray*}
	where $\cd_{\min}=\min_{v\in V}\cd_v$. 
	Hence, 
	%\max \left\{\frac{\p^{s'}_1}{\cd_1}, \max_{u\neq1} \frac{\p^{s'}_u}{\cd_u}\right\}\
	$\max_{u\in V}\frac{\p^{s'}_u}{\cd_u}\le \max_{u\neq 1}\frac{\p^{s'}_u}{\cd_u}+\frac{2\varepsilon}{\cd_{\min}}$.
	In addition, applying Equation (\ref{eqn:decreasing}) recursively we have 
	\begin{eqnarray*}
		\max_{u\in V}\frac{\p^s_u}{\cd_u}&=&\max_{u\neq 1}\frac{\p^s_u}{\cd_u}\\
		&\le&(1-\alpha)\max_{u\in V}\frac{\p^{s-1}_u}{\cd_u}\\
		&\le&(1-\alpha)\left(\max_{u\neq1}\frac{\p^{s-1}_u}{\cd_u}+\frac{2\varepsilon}{\cd_{\min}}\right)\\
		&\le&(1-\alpha)^s\max_{u\in V}\frac{\p^{0}_u}{\cd_u} + \frac{2\varepsilon}{\cd_{\min}}\sum_{t=1}^{s-1}(1-\alpha)^t.
	\end{eqnarray*}
	The inequality means that if $\cd_{\min}>0$ is fixed, $\p^s_u$ can be arbitrarily small when $s\rightarrow\infty$, which contradicts the fact that $\p$ is a probability distribution. This completes the proof.
	
	\textsc{Remark.}
	When the teleportation constant is zero, the PPR vector becomes the stationary probability distribution of a standard random walk, %which is proportional to node degrees,
	\begin{equation*}\label{station}
	%\pi^*=
	\left(\frac{\cd_1}{\sum_i\cd_i},\frac{\cd_2}{\sum_i\cd_i},..., \frac{\cd_N}{\sum_i\cd_i}\right).
	%\pi^*=\frac{1}{\sum_i\cd_i}\left(\cd_1,\cd_2,..., \cd_N\right).
	\end{equation*}
	After adjusting by node degrees, every entry becomes identical (${1/\sum_i\cd_i}$). 
	The lemma is intuitive, recognizing that the teleportation introduces a particular favor of the seed node. 
	
	\textsc{Remark.}
	When the edges are weighted (non-negative), the stationary distribution of a random walk is still proportional to node degrees, if one defines the degree as sum of edge weights incident to the node \citep{lovasz1993random}.
	Note also that the stationary distribution of a random walk in a directed graph is characterized by the in-degree of nodes \citep{ghoshal2011ranking, lu2013respondent}. The conclusion and a modified proof apply to directed or weighted graphs.
\end{proof}

\subsection{Proof of Corollary \ref{cor:main}}

\begin{proof} 
	The algorithm ranks all vertices according to ${p^\epsilon}^*$, and the population local cluster can be explicitly written as
	$$\mathcal{C}=\{v\in V:\p_v^*=\tp^*_1\}.$$
	It suffices to show that
	$${p^\epsilon_v}^*>{p^\epsilon_u}^*, \text{ for } \forall v\in \mathcal{C},u\in V \backslash \mathcal{C},$$ 
	where ${p^\epsilon}^*_v=p^\epsilon_v/d_v$.
	To this end, we apply triangle inequality and get
	\begin{eqnarray*}
		\frac{{p^\epsilon_v}^*-{p^\epsilon_u}^*}{\|\p^*\|_\infty}
		&\ge&\frac{\p^*_v-\p^*_u}{\|\p^*\|_\infty}-\frac{|p^*_v-\p^*_v|}{\|\p^*\|_\infty}-\frac{|p^*_u-\p^*_u|}{\|\p^*\|_\infty}-\frac{\left|{p^\epsilon_u}^*-p^*_u\right|}{\|\p^*\|_\infty}-\frac{\left|{p^\epsilon_v}^*-p^*_v\right|}{\|\p^*\|_\infty} \\
		&\ge&\Delta-\frac{2\|p^*-\p^*\|_\infty}{\|\p^*\|_\infty}-\frac{2\|{p^\epsilon}^*-p^*\|_\infty}{\|\p^*\|_\infty}.
	\end{eqnarray*}
	Since $\Delta_\alpha\le1$, assumption (\ref{eq:assm2}) contains condition (\ref{eq:assm1}) in Theorem \ref{thm:concPPR}, which together with Proposition \ref{prop:approx} implies that
	\begin{eqnarray*}
		\frac{\|p^*-\p^*\|_\infty}{\|\p^*\|_\infty}<\frac{1}{4}\Delta, && \frac{\|{p^\epsilon}^*-p^*\|_\infty}{\|\p^*\|_\infty}<\frac{1}{4}\Delta,
	\end{eqnarray*}
	if 
	%$\frac{\Delta^2\delta}{\log{N}}$ 
	$\Delta^2\delta/\log{N}$ 
	is large enough. These collectively imply $p^*_v>p^*_u$ as desired.
\end{proof}

\newpage
\clearpage
\beginsupplement

\begin{center}
	{\LARGE Supplementary Materials}
\end{center}

%\vspace{1cm}

\begin{abstract}
	This document provides several supplementary materials to ``Targeted sampling from massive block model graphs with personalized PageRank''. 
	Section \ref{sxn:rand} contains a proof for the entrywise error control (Theorem \ref{thm:concPPR}). 
	Section \ref{sxn:parameter} gives additional information about some model parameters, including $\B$, $\cP$, $\alpha$, and $N$. 
	Section \ref{sxn:ld} extends the results in \citet{kloumann2017block} to the DC-SBM from a linear discriminant analysis perspective. 
	Section \ref{sxn:top200} supplies three targeted Twitter samples about the seed @NBCPolitics described in the paper. 
	The PPR clustering is implemented in \textsf{R} and all source codes are available at author's GitHub (\url{https://github.com/RoheLab/aPPR}).
\end{abstract}

\section{A proof for the entrywise error control} \label{sxn:rand}
We start with a few lemmas to prepare for the proof of Theorem \ref{thm:concPPR}. 
For completeness, Section \ref{sxn:ineq} lists a few inequalities that are used throughout the proofs.

\subsection{Some definitions and lemmas}
In this section, we introduce a few notations used in Lemma \ref{lem:Davis} and list a few properties of vector norm and matrix norm \citep{bremaud2013markov}.

For any strictly positive probability distribution vector $p\in\R^N$, the inner product space indexed by $p$ is a real vector space $\R^N$ endowed with the inner product 
$$\langle x,y\rangle_p=\sum_{v=1}^Np_vx_vy_v.$$
The corresponding vector norm and the induced matrix norm are defined respectively as 
$$\|x\|_p=\sqrt{\langle x,x\rangle_p} \text{ \ and \ } \|A\|_p=\sup_{\|x\|_p=1}\|A\T x\|_p.$$
\begin{lemma}
	\label{norm}
	If $0\le p_{\min}\le p_v\le p_{\max}$ for all $v=1,2,...,N$, then the following inequalities hold
	$$\sqrt{p_{\min}}\|x\|_2\le \|x\|_p\le\sqrt{p_{\max}}\|x\|_2 \text{ \ and \ } \sqrt{\frac{p_{\min}}{p_{\max}}}\|A\|_2\le\|A\|_p\le \sqrt{\frac{p_{\max}}{p_{\min}}}\|A\|_2.$$
\end{lemma}

The following lemma provides concentration of the node degrees in a graph generated from the DC-SBM.
\begin{lemma} [Degree concentration]
	\label{degree}
	\dcsbmset Let $d_{\min}$ and $d_{\max}$ be the smallest and the largest node degree observed. 
	Let $\delta$ be the average expected node degree, and define $\rho=\cd_{\max}/\cd_{\min}$. 
	If $\delta\ge c_0(1-\alpha)\logN$ for some sufficiently large constant $c_0>0$, then with probability at least $1-\cO(N^{-10})$, it holds that
	\begin{equation}\label{deg1}
	\frac{\delta}{2\rho}\le d_{\min}\le d_{\max} \le \frac{3\rho\delta}{2}.
	\end{equation}
\end{lemma}
\begin{proof}
	Note that the definition of $\rho$ immediately implies that
	$$\frac{\delta}{\rho}\le\cd_{\min}\le\cd_{\max}\le\delta\rho.$$
	The lemma follows from the standard Chernoff's bound, hence is omitted.
\end{proof}

The following useful lemma concerns the eigenvector perturbation for probability transition matrices, promoted from the celebrated Davis-Kahan $\sin\Theta$ Theorem \citep{davis1970rotation}.

\begin{lemma}[Eigenvector perturbation]\label{lem:Davis}
	Suppose that $Q$, $\hat Q$, and $\Q$ are probability transition matrices with stationary distributions $p$, $\hat p$, and $\p$ respectively. Assume that $\Q$ represents a reversible Markov chain. Then, 
	$$\|p-\hat p\|_\p\le\frac{\|(Q-\hat Q)\T p\|_\p}{1-\max\{\lambda_2(\Q), -\lambda_N(\Q)\}-\|\hat{Q}-\Q\|_\p}.$$
\end{lemma}
The proof the Lemma \ref{lem:Davis} can be found in \citet{chen2017spectral} Section 3, thus omitted.

\subsection{Proof of Theorem \ref{thm:concPPR}}
\begin{proof}
	The proof processes as follows. 
	We first bound the entrywise error rate of $p$, 
	$$\frac{\|p-\p\|_\infty}{\|\p\|_\infty}\le c_0\sqrt{\frac{\logN}{\delta}},$$
	by invoking the novel leave-one-out techniques \citep{chen2017spectral},
	The entrywise error bounds of $p^*$ follows immediately.
	
	Recall that both $p$ and $\p$ are stationary distribution, which means
	$$p=Q\T p \text{ \ and \ } \p=\Q\T \p.$$
	Due to this, for any $w=1,2,...,N$, we can decompose
	\begin{eqnarray*}
		p_w-\p_w&=&Q\T_{\cdot w}p-\Q\T_{\cdot w}\p\\
		&=&\underbrace{(Q_{\cdot w}-\Q_{\cdot w})\T \p}_{:=I_1^w}+\underbrace{Q\T_{\cdot w}(p-\p)}_{:=I_2^w},
	\end{eqnarray*}
	where $Q_{\cdot w}$ denotes the $w$-th column of $Q$.
	
	\begin{enumerate} [(a)]
		\item We start with the first term $I^w_1$. Note that
		\begin{eqnarray*}
			I_1^w&=&(1-\alpha)\sum_{v=1}^N\left[\frac{A_{vw}}{d_v}-\frac{\A_{vw}}{\cd_v}\right]\p_v\\
			&=&\underbrace{(1-\alpha)\sum_{v=1}^N\left[(A_{vw}-\A_{vw})\frac{1}{\cd_v}\right]\p_v}_{:=I_{11}^w}+\underbrace{(1-\alpha)\sum_{v=1}^NA_{vw}\left(\frac{1}{d_v}-\frac{1}{\cd_v}\right)\p_v}_{:=I_{12}^w}.
		\end{eqnarray*}
		
		Recall that $A_{vw}$'s correspond to independent Bernoulli random variables, we can easily bound the first term using Bernstein's inequality (Lemma \ref{bernstein}), with probability at least $1-\cO(N^{-8})$,
		\begin{eqnarray*}
			|I_{11}^w|&\le&(1-\alpha)\left|\sum_{v=1}^N(A_{vw}-\A_{vw})\right|\frac{\|p\|_\infty}{\delta}\\
			&\le&(1-\alpha)\left(\sqrt{16\logN\sum_{v=1}^N\A_{vw}}+\frac{16\logN}{3}\right)\frac{\|p\|_\infty}{\delta}\\
			&\overset{\text{(i)}}{\le}&(1-\alpha)\left(4\sqrt{\frac{\rho\logN}{\delta}}+\frac{16\logN}{3\delta}\right)\|p\|_\infty,
		\end{eqnarray*}
		where (i) follows from the fact that $\rho\delta\le\cd_{\max}$.
		
		Note that the second term is 
		$$I_{12}^w=(1-\alpha)\sum_{v=1}^N \1_{(v,w)\in E}\left(\frac{1}{d_v}-\frac{1}{\cd_v}\right)\p_v,$$
		to which we can apply the Hoeffding's inequality (Lemma \ref{hoeffding}) and obtain
		$$\bP\left(|I_{12}^w|\le\rho(1-\alpha)\sqrt{\frac{\rho\logN}{\delta}}\|\p\|_\infty\right)\ge1-2N^{-8}.$$
		
		In sums, we have high probability event
		\begin{equation}\label{I_1}
		|I^w_1|\le(1-\alpha)\left((4+\rho)\sqrt\rho+3\sqrt{\frac{\logN}{\delta}}\right)\sqrt{\frac{\logN}{\delta}}\|p\|_\infty.
		\end{equation}
		\item 
		The statistical dependency between $p$ and $Q$ introduces difficulty in sharply bounding $I_2^w$. Nevertheless, we can invoke the leave-one-out techniques to decouple the dependency. To this end, we define, for each $w=1,2,...,N$, a new transition matrix $\Qw=\alpha\Pi+(1-\alpha)\Pw$ that bridges between $Q$ and $\Q$. $\Pw$ has almost the same entries as $P$ except for replacing those in $w$-th row or column by their expectations; that is, for any $u \neq v$, 
		$$\Pw_{uv}=\left\{
		\begin{array}{ll}
		P_{uv},&u\neq w \text{ and } v\neq w,\\
		\cP_{uv},&u=w \text{ or } v=w,
		\end{array}\right.$$
		and for any $u=1,2,...,N$,
		$$\Pw_{uu}=1-\sum_{v:v\neq u}\Pw_{uv},$$
		in order to ensure that $\Pw$ and $\Qw$ are transition matrices. In addition, define $\pw$ to be the stationary distribution corresponding to $\Qw$. As demonstrated in \citet{chen2017spectral}, $\pw$ helps us well approximate $p$, yet it is statistically independent of $Q_{\cdot w}$.
		
		Now we decompose $I_2^w$ as follows:
		\begin{eqnarray*}
			I_2^w&=&\sum_{v=1}^N Q_{vw}(p_v-\p_v)\\
			&=&\underbrace{\sum_{v=1}^N Q_{vw}\left(p_v-\pw_v\right)}_{:=I_{21}^w}+ \underbrace{\sum_{v=1}^N Q_{vw}\left(\pw_v-\p_v\right)}_{:=I_{22}^w}.
		\end{eqnarray*}
		
		\item 
		In this part, we focus on the first term $I_{21}^w$, where we would need another intermediate quantity to facilitate our estimation. To be specific, consider the leave-one-out version of $Q$ conditioning on the graph $G=(V,E)$, $\QwG=\alpha\Pi+(1-\alpha)\PwG$, which is almost the same as $Q$ except for replacing the non-zero entries in $w$-th row or column by their expectations. Concretely, for and $u\neq v$,
		$$\PwG_{uv}=\left\{
		\begin{array}{ll}
		P_{uv}, & u\neq w \text{ and } v\neq w,\\
		\1_{(u,v)\in E}\cP_{uv}, & u= w \text{ or } v= w,
		\end{array}\right.$$
		and for any $u=1,2,...,N$, define
		$$\PwG_{uu}=1-\sum_{v:v\neq u}\PwG_{uv},$$
		so that $\PwG$ is a probability transition matrix.
		
		With $\QwG$ in mind, we now apply Cauchy-Schwarz inequality on $I_{21}^w$ to reach
		\begin{eqnarray*}
			|I_{21}^w|&=&\left|\sum_{v=1}^N Q_{vw}\left(p_v-\pw_v\right)\right|\\
			&\le&\left(\sum_{v=1}^N Q_{vw}^2\right)^\frac{1}{2}\left\|p-\pw\right\|_2\\
			&\overset{\text{(i)}}{\le}&\sqrt{\alpha+\frac{1}{d_{\min}}}\sqrt{\frac{\p_{\max}}{\p_{\min}}}\left\|p-\pw\right\|_\p\\
			&\overset{\text{(ii)}}{\le}&\sqrt{\alpha+\frac{1}{d_{\min}}}\sqrt{\frac{\p_{\max}}{\p_{\min}}}\frac{1}{\gamma}\left\|\left(Q-\Qw\right)\T\pw\right\|_2\\
			&\overset{\text{(iii)}}{\overset{\text{w.h.p.}}{\le}}&\sqrt{\alpha+\frac{2\rho}{\delta}}\frac{\sqrt{\kappa}}{\gamma }\left(\underbrace{\left\|(Q-\QwG)\T\pw\right\|_2}_{:=I^w_{211}}+\underbrace{\left\|(\QwG-\Qw)\T\pw\right\|_2}_{:=I^w_{212}}\right).\\
		\end{eqnarray*}
		where (i) follows from Lemma \ref{norm} and the fact that $P_{vw}\le \frac{1}{d_{\min}}$, (ii) comes from Lemma \ref{lem:Davis}, and (iii) results from Lemma \ref{degree} and the triangle inequality, and recognizing $\kappa = \p_{\max}/\p_{\min}$ (from the proof of Proposition \ref{ppr_sol}, it is bounded), and ``w.h.p.'' is short for ``with high probability''. Note that $\Pi$ adds at most 1 to the rank of $\Q$, and because we presume $B$ is positive definite $\cP$ has exactly $K$ positive eigenvalues among other zeros (Section \ref{sxn:sp_transition}). Here, $\gamma=1-\max\{\lambda_2(\Q), -\lambda_N(\Q)\}-\|\QwG-\Q\|_\p$ is the spectral gap and is lower bounded by some positive constant (due to \citet{khanna2017approximation}). Then, it boils down to controlling $I^w_{211}$ and $I^w_{212}$.
		
		For $I^w_{211}$, the $w$-th entry inside the vector norm is 
		\begin{eqnarray*}
			\left[\left(Q-\QwG\right)\T\pw\right]_w&=&\left[\left(Q-\Q\right)\T\pw\right]_w\\
			&=&(1-\alpha)\sum_{v=1}^N(P_{vw}-\cP_{vw})\pw_v.
		\end{eqnarray*}
		Note that $\pw_v$ is statistically independent of $P_{\cdot w}$. Then, by Hoeffding's inequality (Lemma \ref{hoeffding}) and Lemma \ref{degree}, we have with probability at least $1-2N^{-8}$,
		\begin{equation}
		\label{eqn:I2111}
		\left[\left(Q-\QwG\right)\T\pw\right]_w\le4\rho(1-\alpha)\sqrt{\frac{\rho \logN}{\delta}}\left\|\pw\right\|_\infty.
		\end{equation}
		As for any $u\neq w$, applying Hoeffding's inequality again yields
		\begin{eqnarray*}
			\left[\left(Q-\QwG\right)\pw\right]_u&=&(1-\alpha)\sum_{v=1}^N\left(P_{vu}-\cP_{vu}\right)\pw_v\\
			&=&(1-\alpha)\left(P_{uu}-\PwG_{uu}\right)\pw_u\\
			&&+(1-\alpha)\left(P_{uw}-\PwG_{uw}\right)\pw_w\\
			&=&-(1-\alpha)\left(P_{uw}-\PwG_{uw}\right)\pw_u\\
			&&+(1-\alpha)\left(P_{uw}-\PwG_{uw}\right)\pw_w.
		\end{eqnarray*}
		Recognizing that 
		\begin{eqnarray*}
			\left|P_{uw}-\PwG_{uw}\right|&=&\left\{
			\begin{array}{ll}
				A_{uw}d^{-1}_{uu}-\A_{uw}\cd^{-1}_{uu},&(u,w)\in E,\\
				0,&(u,w)\notin E,
			\end{array}\right.
		\end{eqnarray*}
		we apply again the Hoeffding's inequality (Lemma \ref{hoeffding}) together with (\ref{deg1}), and obtain with probability at least $1-\cO\left(N^{-8}\right)$,
		\begin{equation}
		\label{eqn:I2112}
		\left|\left[\left(Q-\QwG\right)\T\pw\right]_u\right|\le\left\{
		\begin{array}{ll}
		4\rho(1-\alpha)\frac{\sqrt{\logN}}{\delta}\left\|\pw\right\|_\infty,&(u,w)\in E,\\
		0,&(u,w)\notin E.
		\end{array}\right.\\
		\end{equation}
		Combining (\ref{eqn:I2111}) and (\ref{eqn:I2112}) yields 
		\begin{eqnarray*}
			I^w_{211}&\le&4\rho(1-\alpha)\left(1+\sqrt{\sum_{u:u\neq w}\1_{(u,w)\in E}}\right)\sqrt{\frac{\rho \logN}{\delta}}\left\|\pw\right\|_\infty\\
			&\overset{\text{(i)}}{\overset{\text{w.h.p.}}{\le}}& 8\rho^2\sqrt{\rho}(1-\alpha)\sqrt{\frac{\logN}{\delta}}\left\|\pw\right\|_\infty,
		\end{eqnarray*}
		where (i) follows from the high probability event that $d_{\max}\le3\rho\delta/2$.
		
		Regarding $I^w_{212}$, since $(\QwG-\Qw)\p=\vec 0$, we can rewrite this as
		$$I^w_{212}=\left\|\left(\QwG-\Qw\right)\T\left(\pw-\p\right)\right\|_2.$$
		Similarly, note that $\Pw_{vw}-\PwG_{vw}=\frac{\A_{vw}}{\cd_v}\1_{(w, v)\notin E}$, we apply Bernstein's inequality on $w$-th term inside the vector norm to obtain that with probability at least $1-2N^{-8}$,
		\begin{eqnarray*}
			\left[\left(\QwG-\Qw\right)\left(\pw-\p\right)\right]_w&=&(1-\alpha)\sum_{v=1}^N\left(\PwG_{vw}-\Pw_{vw}\right)\left(\pw_v-\p_v\right)\\
			&=&(1-\alpha)\sum_{v=1}^N\frac{1}{\cd_v}\left(\pw_v-\p_v\right)\1_{(w, v)\notin E}\\	&\overset{\text{w.h.p.}}{\le}&(1-\alpha)\left(4\rho\sqrt{\frac{\rho\logN}{\delta}}+\frac{16}{3}\frac{\logN}{\delta}\right)\left\|\pw-\p\right\|_\infty.
		\end{eqnarray*}
		For any $u\neq w$, the $u$-th term inside vector norm is
		\begin{eqnarray*}
			\left[\left(\QwG-\Qw\right)\left(\pw-\p\right)\right]_u&=&(1-\alpha)\sum_{v=1}^N\left(\PwG_{vu}-\Pw_{vu}\right)\left(\pw_v-\p_v\right)\\
			&=&(1-\alpha)\left(\PwG_{vv}-\Pw_{uu}\right)\left(\pw_u-\p_u\right)\\
			&&+(1-\alpha)\left(\PwG_{vw}-\Pw_{uw}\right)\left(\pw_w-\p_w\right)\\
			&=&-(1-\alpha)\left(\PwG_{vw}-\Pw_{uw}\right)\left(\pw_u-\p_u\right)\\
			&&+(1-\alpha)\left(\PwG_{vw}-\Pw_{uw}\right)\left(\pw_w-\p_w\right).
		\end{eqnarray*}
		Recognizing that 
		$$\Pw_{uw}-\PwG_{uw}= \A_{uw}\cd^{-1}_{u}\1_{(u,w)\notin E},$$
		we have from (\ref{deg1}) that
		$$\left|\left[\left(\QwG-\Qw\right)\T\left(\pw-\p\right)\right]_u\right|\le 2\A_{uw}\cd^{-1}_{u}\1_{(u,w)\notin E}(1-\alpha)\left\|\pw-\p\right\|_\infty.$$
		Hence, we have with probability at least $1-\cO\left(N^{-8}\right)$,
		\begin{eqnarray*}
			I^w_{212}&\le&(1-\alpha)\left(4\rho\sqrt{\rho}\sqrt{\frac{\logN}{\delta}}+\frac{16}{3}\frac{\logN}{\delta}+2\sqrt{\sum_{u:u\neq w}\frac{\1_{(u,w)\notin E}}{\D_{uu}^2}}\right)\left\|\pw-\p\right\|_\infty\\
			&\overset{\text{(i)}}{\le}&(1-\alpha)\left(4\rho\sqrt{\rho}\sqrt{\frac{\logN}{\delta}}+\frac{16}{3}\frac{\logN}{\delta}+2\rho\sqrt{\frac{\rho}{\delta}}\right)\left\|\pw-\p\right\|_\infty,
		\end{eqnarray*}
		where (i) follows from the high probability event that $d_{\max}\le3\rho\delta/2$.
		Combining the above two bounds, we have with probability at least $1-\cO\left(N^{-8}\right)$ that
		\begin{eqnarray*}
			I^w_{21}&\le&\sqrt{\alpha+\frac{2\rho}{\delta}}\frac{\sqrt{\kappa}}{\gamma }\left(I^w_{211}+I^w_{212}\right)\\
			&\le&c(1-\alpha)\left(8\rho^2\sqrt{\frac{\rho\logN}{\delta}}\|\pw\|_\infty+\left(2\rho\sqrt\frac{\rho}{\delta}+4\rho\sqrt{\frac{\rho\logN}{\delta}}+\frac{16}{3}\frac{\logN}{\delta}\right)\|\pw-\p\|_\infty\right)\\
			&\overset{\text{(i)}}{\le}&8c\rho^2(1-\alpha)\sqrt{\frac{\rho\logN}{\delta}}\|\p\|_\infty\\
			&&+c(1-\alpha)\left(8\rho^2\sqrt{\frac{\rho\logN}{\delta}}+ 2\rho\sqrt\frac{\rho}{\delta}+4\sqrt{\frac{\rho\logN}{\delta}}+\frac{16}{3}\frac{\logN}{\delta}\right)\|\pw-\p\|_\infty\\
			&\overset{\text{(ii)}}{\le}&8c\rho^2(1-\alpha)\sqrt{\frac{\rho\logN}{\delta}}\|\p\|_\infty+\frac{c}{2}\|\pw-\p\|_\infty\\
			&\overset{\text{(iii)}}{\le}&16c\rho^2(1-\alpha)\sqrt{\frac{\rho\logN}{\delta}}\|\p\|_\infty+c\|p-\p\|_\infty.\\
		\end{eqnarray*}
		where $c=\sqrt{\alpha+\frac{2\rho}{\delta}}\frac{\sqrt{\kappa}}{\gamma }$, and (i) follows from the triangle inequality $\left\|\pw\right\|_\infty\le\left\|\pw-\p\right\|_\infty+\|\p\|_\infty$, and (ii) holds as long as $\delta>c_0(1-\alpha)^2\logN$ for some $c_0>0$ sufficiently large, and (iii) comes from the triangle inequality $\left\|\pw-\p\right\|_\infty\le\left\|\pw-p\right\|_2+\|p-\p\|_\infty$.
		
		\item Now it is left to estimate the last item $I^w_{22}$. Note that
		\begin{eqnarray*}
			I^w_{22}&=&\sum_{v=1}^N\1_{(v,w)\in E}Q_{vw}\left(\pw_v-\p_v\right)\\
			&=&\sum_{v=1}^N\left[\alpha\1_{\{w=1\}}+(1-\alpha)\frac{1}{d_v}\1_{(v,w)\in E}\right]\left(\pw_v-\p_v\right)\\
			&=&\underbrace{\alpha\sum_{v=1}^N\1_{\{w=1\}}\left(\pw_v-\p_v\right)}_{:=I^w_{221}}+\underbrace{(1-\alpha)\sum_{v=1}^N\frac{\1_{(v,w)\in E}}{\cd_v}\left(\pw_v-\p_v\right)}_{:=I^w_{222}}\\
			&&+\underbrace{(1-\alpha)\sum_{v=1}^N\left(\frac{1}{d_v}-\frac{1}{\cd_v}\right)\1_{(v,w)\in E}\left(\pw_v-\p_v\right)}_{:=I^w_{223}}.\\
		\end{eqnarray*}
		Since both $\pw$ and $\p$ are distribution vector, $I^w_{221}=0$. Then, due to Hoeffding's inequality (Lemma \ref{hoeffding}), 
		\begin{eqnarray*}
			|I^w_{222}|&\le&4\rho(1-\alpha)\sqrt{\frac{\rho\logN}{\delta}}\left\|\pw-\p\right\|_\infty,\\
			|I^w_{223}|&\le&2\rho(1-\alpha)\sqrt{\frac{\rho\logN}{\delta}}\left\|\pw-\p\right\|_\infty,
		\end{eqnarray*}
		with probability at least $1-\cO\left(N^{-8}\right)$.
		Thus, we reach the high probability event
		\begin{eqnarray*}
			|I^w_{22}|&\le&6\rho(1-\alpha)\sqrt{\frac{\rho\logN}{\delta}}\left\|\pw-\p\right\|_\infty.
		\end{eqnarray*}
		In sums, we reach with probability at least $1-\cO\left(N^{-8}\right)$,
		\begin{eqnarray}\label{I_2}
		|I^w_{2}|&\le&16\rho^2(1-\alpha)\frac{\sqrt{\kappa\rho}}{\gamma}\sqrt{\alpha+\frac{2\rho}{\delta}}\sqrt{\frac{\logN}{\delta}}\|\p\|_\infty\nonumber\\
		&&+\left(\frac{\sqrt{\kappa}}{\gamma}\sqrt{\alpha+\frac{2\rho}{\delta}}+6\rho(1-\alpha)\sqrt{\frac{\rho\logN}{\delta}}\right)\|p-\p\|_\infty.
		\end{eqnarray}
		
		\item Collecting the preceding bounds (\ref{I_1}) and (\ref{I_2}) together, we conclude that with high probability
		\begin{eqnarray*}
			\|p-\p\|_\infty&=&\max_w |p_w-\p_w|\\
			&\le&c_2 (1-\alpha)\sqrt{\frac{\logN}{\delta}}\|\p\|_\infty + c_3 \|p-\p\|_\infty,
		\end{eqnarray*} 
		as long as $\delta/[(1-\alpha)\logN]$ is sufficiently large, which controls the entrywise error of $p$,
		\begin{equation}\label{pbound}
		\frac{\|p-\p\|_\infty}{\|\p\|_\infty}\le c_1(1-\alpha)\sqrt{\frac{\logN}{\delta}},
		\end{equation}
		for some sufficiently large constant $c_1,c_2,c_3>0$.
		
		\textsc{Remark.}
		$c_2$ and $c_3$ are controlled by constants $\rho, \kappa, \gamma$, which are thereby driven from the model parameters $\B$, $\Theta$, $K$, and $Z$.

		\item Finally, we accomplish the proof by observing that
		\begin{eqnarray*}
			\frac{\|p^*-\p^*\|_\infty}{\|\p^*\|_\infty}&\le&\frac{2\max \left(d_{\min}^{-1}, \cd_{\min}^{-1}\right)}{\cd_{\min}^{-1}} \frac{\|p-\p\|_\infty}{\|\p\|_\infty}\\
			&\le&\frac{4\|p-\p\|_\infty}{\|\p\|_\infty}.
		\end{eqnarray*}
	\end{enumerate}
	Above observation together with the inequality (\ref{pbound}) allow us to control the entrywise error of $p^*$ as claimed, with probability at least $1-\cO\left(N^{-5}\right)$,
	$$\frac{\|p^*-\p^*\|_\infty}{\|\p^*\|_\infty} \le c_2(1-\alpha)\sqrt{\frac{\logN}{\delta}},$$
	for some sufficiently large constant $c_2>0$.
\end{proof}

\subsection{Concentration inequalities} \label{sxn:ineq}
The following is a standard concentration inequality used throughout the paper.
\begin{lemma}[Hoeffding's inequality]\label{hoeffding}
	Let $\{X_i\}_{1\le i\le n}$ be a sequence of independent random variables where $X_i\in[a_i,b_i]$ for each $1\le i \le n$, and $S_n=\sum_{i=1}^n X_i$. Then, 
	$$\bP(|S_n-\E S_n|\ge t)\le 2\exp\left\{-\frac{t^2}{\sum_{i=1}^n(b_i-a_i)^2}\right\}.$$
\end{lemma}

The next lemma is a special case of Chernoff's bound.
\begin{lemma}[Chernoff's bounds]\label{chernoff}
	Let $\{X_i\}_{1\le i\le n}$ be a sequence of independent random variables, whose sum is $S_n$, each having probability $p_i$ of being equal to $a_i$, otherwise 0. Define $\mu=\sum_i p_ia_i$. Then, for any $\epsilon>0$,
	$$\bP\left(X_i\ge (1+\epsilon)\mu\right)\le (1+\epsilon)^{-\epsilon\mu},$$
	$$\bP\left(X_i\le (1-\epsilon)\mu\right)\le (1-\epsilon)^{\epsilon\mu}.$$
\end{lemma}

For the use of this paper, we only invoke a simpler version of Bernstein inequality.
\begin{lemma}[Bernstein's inequality]\label{bernstein}
	Let $\{X_i\}_{1\le i\le n}$ be a sequence of independent random variables with $|X_i|\le B$ for each $1\le i \le n$, and $S_n=\sum_{i=1}^n X_i$ and $T_n=\sum_{i=1}^n X_i^2$. Then, with probability at least $1-2n^{-a}$,
	$$|S_n-\E [S_n]|\le\sqrt{2a\log n\E [T_n]} +\frac{2a}{3}B\log n$$
	for any $a\ge2$.
\end{lemma}
The proofs of Lemma \ref{hoeffding}, \ref{chernoff}, and \ref{bernstein} can be found in \citet{boucheron2013concentration}, hence are omitted.

%\newpage
\section{Additional information on model parameters}\label{sxn:parameter}

\subsection{Block connectivity matrix $\B$} 
In the paper, we assume that the block connectivity matrix $\B$ corresponds to a strongly connected graph at block level and is positive definite. 
These assumptions asserts the efficacy of PPR clustering and is primarily a technical assumption sufficient for our theoretical results. 
In fact, we require $\B$ to represent to a strongly connected graph because this enables the block-wise PPR vector to have the largest value corresponding to the block of seed(s) (Lemma \ref{lem:top} in the paper). 
On the other hand, we impose the positive definiteness on $\B$ because this allows us to intuitively define the notion of local cluster, yet our statistical theory (i.e., the entrywise control of sample PPR vector) does not explicitly rely on such positive-definiteness per se.
It is not clear yet whether these constraints are \textit{necessary} in order for PPR clustering to function; possible generalizations of them are of research interest. 

We list a few concrete examples showing that 
(i) if we break the strongly-connectivity assumption, the PPR clustering can fail, despite a reasonable teleportation constant, $\alpha=0.15$, but (ii) PPR clustering often works as hoped even when $\B$ is not positive-definite. 
Throughout, we assume that the first block is targeted and 
consider directed graphs with three underlying blocks ($K=3$). 
The first two instances of $\B$ demonstrates the necessity of the strongly-connectivity constraint, which ensures the block-wise aPPR vector to possess the largest first element.  
The third and forth instances, on the other hand, indicate that $\B$ need not to be positive definite. 

\subsubsection{Violating the strongly-connective assumption}
\paragraph{Hierarchy case.} 
Let the block connectivity matrix 
$$\B=\begin{bmatrix} 
p & p & p \\
0 & p & p \\
0 & 0 & p \end{bmatrix}$$
for some constant $p>0$. $\B_{ij}$ is the number (or the probability) of edges from the $i$-th block to the $j$-th block in population. 
Then, the directed graph represented by $\B$ is not strongly connected, as block 3 has no path to the first block. 
In fact, this graph (specified by upper triangular $\B$) has a hierarchical structure, where the third block is in the center (or the highest hierarchy) of the graph, and the member of first block are essentially satellite from outside.
Particularly, edges only come from outsiders to insiders.

We now perform the PPR clustering on the first cluster. 
The block-wise transition matrix is 
$$\tP=\begin{bmatrix} 
1/3 & 1/3 & 1/3 \\
0 & 1/2 & 1/2 \\
0 & 0 & 1 \end{bmatrix}.$$
Then, both $\B$ and $\tP$ are positive definite, with eigenvalues of $(p, p, p)$ and $(1, 1/2, 1/3)$ respectively.
To ease the calculation, we set $p=3$.
Then the block-wise PPR vector is approximately 
$$\tp=(0.209, 0.103, 0.688),$$ 
and the block-wise aPPR vector is approximately (after adjusting by column sums of $\B$) 
$$\tp^*=(0.0698, 0.0172, 0.0764).$$
As shown, neither block-wise PPR vector nor aPPR vector properly recognize the local cluster 1. 

\paragraph{Adding a small amount of circulation.}
If we add a small quantity to the left bottom element of above $\B$ matrix, then the block connectivity matrix corresponds to a connected graph.
To illustrate, we assign a small value to it, $\B_{31}=0.1$, then the new block connectivity matrix becomes
$$\B'=\begin{bmatrix} 
p & p & p \\
0 & p & p \\
0.1 & 0 & p \end{bmatrix}.$$
To explore the PPR vector, we set $p=3$ once again. 
In this case, $\B'$ has one real eigenvalue ($\approx 4.069$) and two imaginary eigenvalues.
The block-wise PPR vector is approximately 
$$\tp=(0.235, 0.115, 0.650),$$ 
and the block-wise aPPR vector is approximately (after adjusting by column sums of $\B'$) 
$$\tp^*=(0.0755, 0.0192, 0.0723).$$
In this case, the PPR clustering works like a charm.

\subsubsection{Violating the positive-definite assumption} 
Consider again the $K=3$ design with equally distributed block size. 
We present two examples breaking the positive-definite assumption on $\B$, where the PPR cluster still operates properly.

\paragraph{Indefinite case.}
Given some constants $r>p>0$, define $$\B=\begin{bmatrix} 
p & r & r \\
r & p & r \\
r & r & p \end{bmatrix}.$$

In this case, the random graphs generated from such configuration of $\B$ have a unique characteristic: two vertices with different block memberships are more likely to connect than those pairs belonging to the same block. 
Note that the three eigenvalues of $\B$ are $p+2r$, $p-r$, and $p-r$. Hence, $B$ is an indefinite matrix (so does the block-wise transition matrix $\tP$).

Interestingly, the PPR clustering continues working under this circumstance. 
For simplicity, setting $p=3$ and $r=9$, and we articulate the block-wise PPR vector and aPPR vector. 
In fact, the block-wise PPR vector is approximately 
$$\tp=(0.386, 0.306, 0.306).$$ 
Since $\B$ has homogeneous column sums, it follows that the first element in the block-wise aPPR vector is also the largest, suggesting the effectiveness of PPR clustering.
The same conclusion hold when we set $p=3$ and $r=99$ (or $999$).

\paragraph{Singular case.} 
Suppose $p>0$ and let $$\B=\begin{bmatrix} 
0 & p & 0 \\
p & 0 & p \\
0 & p & 0 \end{bmatrix}.$$

In this case, nodes in block 1 only connect with those nodes in block 2, and the nodes in block 3 only have edges with block 2's members.
Note that $\B$ is singular because three of its eigenvalues are 1, -1, and 0.
So does the block-wise transition matrix.
However, the PPR clustering remain effective. 
In fact, the block-wise PPR vector and aPPR vector are
\begin{eqnarray*}
	\tp=(0.345, 0.459, 0.195) &\text{and}& \tp^*=\frac{1}{p}(0.345, 0.230, 0.195).
\end{eqnarray*}

In both cases (when $\B$ is not positive-definite), the block-wise aPPR vector correctly assigns the largest value to the first element and thus is still effective for targeted sampling. 
These examples suggest a potentially greater applicability of the PPR clustering under the block model graph.

\subsubsection{Comments}
Putting together above demonstrations, we briefly comment on $\B$ and the PPR clustering. 
(i) The strongly-connectivity assumption is essential for the PPR clustering to be consistent.
(ii) The efficacy of PPR clustering is conditioning on the fact that teleportation constant is sufficiently large. 
If we assign an extremely small to it, e.g. $\alpha=0.001$, the PPR clustering collapses. 
(iii) Beyond community-like graphs (where $\B$ is positive-definite), the PPR clustering has potential for working on a more general block model graphs.
%(iv) Finally, we noticed from above examples that $\B$ has at least one \textit{real} and positive eigenvalue in order for a consistent PPR clustering.

\subsection{Spectral analysis on graph transition $\cP$} \label{sxn:sp_transition}
In this section, we present a spectral analysis of graph transition matrix, which demonstrates that (1) under the population DC-SBM, a graph transition matrix $\cP$ has exactly $K$ positive eigenvalues, and $N-K$ zero eigenvalues, and (2) in a random graph generated from the DC-SBM, the graph transition matrix $P$ is close to its population, with respect to spectral norm.

\begin{lemma} [Eigen-decomposition for $\cP$ and $\tP$] 
	\underdcsbm let $\cP\in\R^{N\times N}$ be the population graph transition matrix and $\tP\in\R^{K \times K}$ be the block-wise transition matrix. Then, $\cP$ and $\tP$ have the same $K$ positive eigenvalues. The remaining $N-K$ eigenvalues of $\cP$ are all zeros. Denote the $K$ positive eigenvalues of both matrices as $\lambda_1 \ge \lambda_2 \ge \cdots ... \lambda_K\ge 0$, and let $\mathcal{X}\in \mathbb R^{N\times K}$ and $\mathcal{Y}\in \mathbb R^{K\times K}$ contain the left eigenvector of $\cP$ and $\tP$ respectively, corresponding to $\lambda_i$ in their $i$-th column. Then, there exists a orthogonal matrix $U\in\R^{K \times K}$, such that
	\begin{enumerate}[\normalfont(a)]
		\item $\mathcal X\T=\mathcal{D}^{-1/2}\Theta^{1/2}ZU;$ and
		\item $\mathcal Y\T=\tD^{-1/2}U.$
	\end{enumerate} 
\end{lemma}

\begin{proof}
	We follow the proof of Lemma 3.3 in \citet{qin2013regularized}. Define $\textbf{L}=\tD^{-1/2}\B\tD^{-1/2}$, then $\tP= \tD^{-1/2}\textbf{L}\tD^{1/2}$.
	By model assumption, $\tP\succ 0$. 
	%		Let $U\in \mathbb R^{K\times K}$ be an orthogonal matrix whose $k$-th column is the eigenvector of $\tP$ corresponding to $\lambda_k$, $k=1,2,...,K$, then $\tP=U\Lambda U\T$. 
	
	Define the graph Laplacian $\mathcal{L=D^{-1/2}AD^{-1/2}}$, then by Lemma \ref{lem:basic}\ref{D_DB}, 
	$$\mathcal{L}_{uv}=\frac{\mathcal{A}_{uv}}{\sqrt{\cd_u\cd_v}}=\frac{\theta_u\theta_v\B_{z(u)z(v)}}{\sqrt{\cd_u\cd_v}}=\frac{\B_{z(u)z(v)}\sqrt{\theta_u\theta_v}}{\sqrt{\td_{z(u)}\td_{z(v)}}}=[\textbf{L}]_{z(u)z(v)}\sqrt{\theta_u\theta_v},$$
	or equivalently, $$\mathcal{L}=\Theta^{1/2}Z\textbf{L} Z\T  \Theta^{1/2}.$$
	
	Then 
	$$\mathcal{X}\T\Lambda\mathcal{X}'=\D^{-1/2}\Theta^{1/2}ZU\Lambda U\T Z\T \Theta^{1/2}\D^{1/2}=\D^{-1/2}\mathcal{L}\D^{1/2}=\D^{-1}\A=\mathcal{P},$$
	and
	$$\mathcal{Y}\T\Lambda\mathcal{Y}'=\tD^{-1/2}U\Lambda U\T \tD^{1/2}=\tD^{-1/2}\textbf{L} \tD^{1/2}=\tP,$$
	where $\mathcal X' =U\T Z\T \Theta^{1/2}\mathcal{D}^{1/2}$ and $\mathcal Y'=U\T \tD^{1/2}$ are right eigenvectors if $\cP$ and $\tP$ respectively.
	Recognizing that $\mathcal{X}\T\mathcal{X}'=\mathcal{Y}\T\mathcal{Y}'=I_K$ completes the proof.
\end{proof}

\begin{lemma} \label{P&L}
	Let $L$ be a symmetric matrix, let $D$ be a diagonal matrix, and let $P=D^{-1/2}LD^{1/2}$. If $x$ is an eigenvector of $L$ corresponding to eigenvalue $\lambda$, then
	\begin{enumerate}[\normalfont(a)]
		\item $D^{-1/2}x$ is a right eigenvector of $P$ with eigenvalue $\lambda$, \label{eigen_trans}
		\item $\|P\T P\|=\|L\|^2$.\label{norm_trans}
	\end{enumerate}
\end{lemma}

\begin{proof} 
	Let $y=D^{-1/2}x$, then
	$$Py=D^{-1/2}LD^{1/2}y=D^{-1/2}Lx=\lambda D^{-1/2}x=\lambda y.$$
	Part \ref{eigen_trans} of the lemma follows. To see \ref{norm_trans}, observe that $y$ is also an eigenvector of $P\T P$ with eigenvalue $\lambda^2$.
\end{proof}

Lemma \ref{P&L} implies that $P$ has the same spectral norm of graph Laplacian $L$. Since $L$ concentrates to $\mathcal{L}$ (see for example \citet{qin2013regularized} for a proof), we have under a random graph generated from the DC-SBM, the graph transition matrix $P$ concentrates to its population $\cP$ with respect to spectral norm.

\subsection{Teleportation constant $\alpha$} 
In the paper, we state that a sufficiently large teleportation constant $\alpha$ enables the entrywise control of sample PPR vector, thus facilitating the PPR clustering in a random graph. 
Here, from a practical perspective, we further illustrate the sensitivity of PPR clustering to $\alpha$, with the Twitter friendship network. 
To this end, we investigate the targeted sampling returned by four configurations of the teleportation constant, $\alpha\in\{0.1, 0.15, 0.25, 1/3\}$, where NBC Politics (@NBCPolitics) is the seed.
The tolerance parameter is fixed, $\epsilon=10^{-7}$, in four targeted sampling.

Table \ref{tab:size} lists the number of Twitter users we examined and the total number of users we ``reached'' (as of August 2019) in four attempts. 
Here, we examine a user by retrieving its friend list (after which it gets a positive $p_u$ value in Algorithm \ref{alg:ppr_dir}), and reach a user once it appears in a user's friend list (at which point, it possesses a positive $r_v$ value in Algorithm \ref{alg:ppr_dir}).
Given the same tolerance parameter, varying the teleportation constant largely affects the number of nodes examined/reached.
This demonstrates the role of teleportation constant in leveraging between the seed preference and the standard random walk. 

\begin{table}
	\caption{Number of nodes examined/reached by Algorithm \ref{alg:ppr_dir} with seed node @NBCPolitics and different teleportation constants, and a fixed tolerance parameter $\epsilon=10^{-7}$, as in August 2019.} 
	\label{tab:size}
	\centering
	\begin{tabular}{ccc}
		\hline
		\noalign{\vskip .3mm}  
		$\alpha$ & Examined & Reached \\ 
		\hline
		\noalign{\vskip .7mm}  
		0.1 & 7,445 & 342,454 \\ 
		0.15 & 5,919 & 272,985 \\ 
		0.25 & 4,860 & 228,561 \\ 
		1/3 & 3,984 & 193,848 \\ 
		\hline
	\end{tabular}
\end{table}

Despite the fact that different $\alpha$'s result in substantial difference in network coverage, when the algorithm stops, the estimated local clusters appear to share a vast majority in common. 
To demonstrate this immediately, we inspect the local clusters of size $n=300$ returned by Algorithm \ref{alg:lc_dir} with four $\alpha$'s and quantify to what degree do they overlap each other.
Table \ref{tab:pairOverlap} shows the percentage of common members between each pair of four returned local clusters.
As shown, most pairs have about 90\% overlapping members, indicating that PPR clustering is fairly robust against the teleportation constant.

The stability of PPR clustering continues to show when we vary the cluster size, $n=100,150, ..., 700$. 
Figure \ref{fig:teleport} shows the proportion of common members across \textit{all} four local clusters, returned by PPR, aPPR, and rPPR (with the regularizer $\tau=10$). 
Overall, the PPR clustering produces a fairly consistent local cluster, with around 80\% of members overlapping across four different strengths of teleportation (see Supplementary Materials).

\begin{table}
	\caption{Percentage of pairwise overlapping among three local clusters around @NBCPolitics with different teleportation constants, $\alpha\in\{0.1, 0.15, 0.25, 1/3\}$, as in August 2019.} 
	\label{tab:pairOverlap}
	\centering
	\begin{tabular}{c|cccc}
		\hline
		%		\noalign{\vskip .3mm}  
		$\alpha$ & 0.1 & 0.15 & 0.25 & 1/3 \\ 
		\hline
		%		\noalign{\vskip .7mm}  
		\hline
		0.1 & 100\% & 92.7\% & 89.3\% & 87.7\% \\ 
		0.15 &  & 100\% & 93.3\% & 90\% \\ 
		0.25 &  & & 100\% & 92\% \\ 
		1/3 &  & &  & 100\% \\ 
		\hline
	\end{tabular}
\end{table}

\begin{figure} %[!htbp]
	\centering
	\includegraphics[width=.9\textwidth]{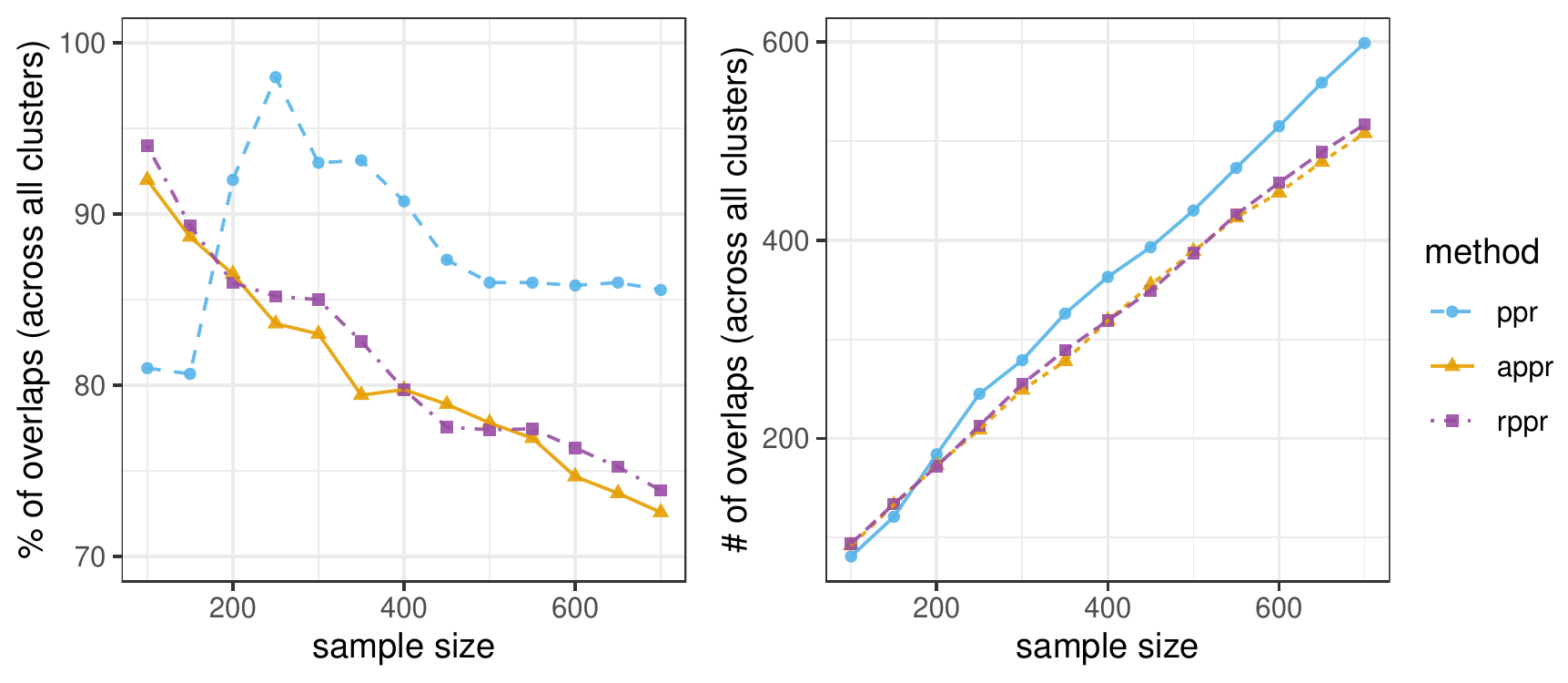} 
	\caption{Sensitivity to the teleportation constant, $\alpha=\{0.1, 0.15, 0.25, 1/3\}$. Shown are the percentage (left) and number (right) of common members across all four local clusters returned by three PPR clustering methods. The targeted sample size increases from 100 to 700 with the increment of 50.}
	\label{fig:teleport}
\end{figure}

We conclude that in practice, PPR clustering (i) is mainly influential to the number of nodes examined in the targeted sampling and (ii) has fairly robust performance with respect to the choice of teleportation constant.

\subsection{The graph size $N$} 
\begin{figure}
	\centering
	\includegraphics[width=0.5\linewidth]{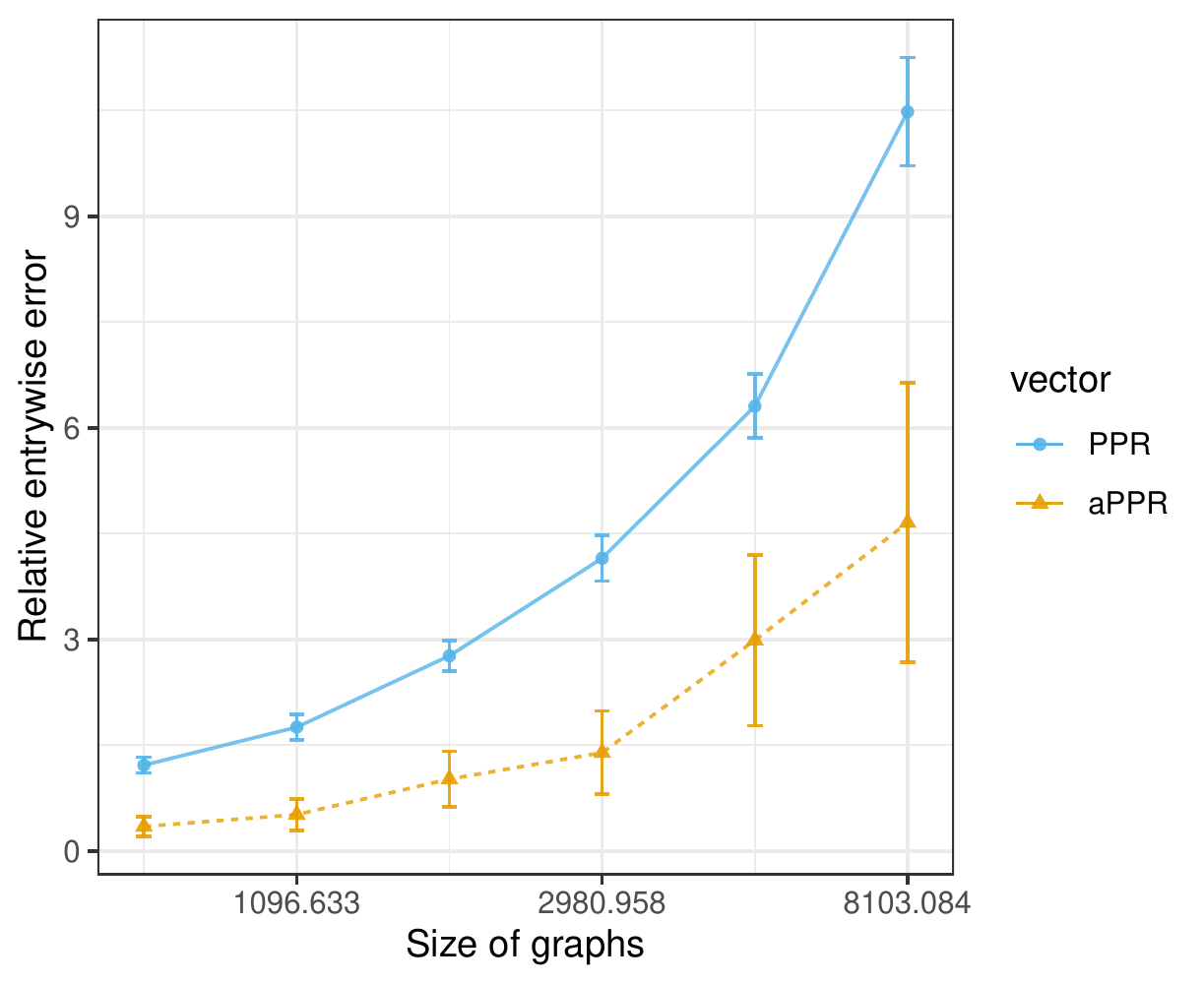}
	\caption{
		Entrywise error rate versus the graph sizes. Shown are relative entrywise error (REE) corresponding to different underlying graph sizes, averaged over 30 replicates. 
		For each dot, an error bar indicates the standard error. 
		The RER for aPPR vector is scaled down by a factor of 240 to improve visualization.
		The ticks in x-axis are transformed through logarithm with the natural base.
	}
	\label{fig:graphsize}
\end{figure}

In the paper, we provide an entrywise error control for the PPR vector and the aPPR vector (Theorem \ref{thm:concPPR}), assuming the edge density is sufficiently large (i.e., inequality (\ref{eq:assm1}) in the main paper).
Simulation 3 in Section \ref{sxn:simu} demonstrates the relationship between the expected degree ($\delta$) and the error rate of PPR clustering, as promised by the theorem. 
Here, we provide another simulation to illustrate Theorem \ref{thm:concPPR}. 
Specifically, we further investigate the affect of graph size ($N$) on the relative entrywise error (REE) of the PPR vector ($\frac{\|p-\p\|_\infty}{\|\p\|_\infty}$), given some edge density ($\delta$). 

We generate 30 replicates of networks of size $N=e^x$, where $x\in\{6.5, 7, 7.5, 8,8.5, 9\}$, from the four-parameter stochastic block model, $\text{SBM}(K=3, N, b_1=9, b_2=3)$.
The average expect degree is set to $\delta=125$.
Both PPR vectors and aPPR vectors are calculated for every network, with teleportation constant $\alpha = 0.15$ and $10$ seeds randomly selected from the first block.
Figure \ref{fig:graphsize} depict the REE with respect to different graph sizes (scaled by a logarithm transformation with the natural base).
As shown, with $\delta$ fixed (not growing at the rate of $\logN$), the REE increases as the graph expands, so does the variance of REE for both PPR and aPPR vectors, matching the results in Theorem \ref{thm:concPPR}.

%\newpage
\section{Connection to linear discriminant analysis} \label{sxn:ld}

In this section, we give another representation of PPR vectors in the landing probability space, which builds upon \citet{kloumann2017block}. This assorts PPR to a greater functional regime. Then, we extend the previous result that links the PPR vectors with linear discriminant functions under the DC-SBM. In particular, when every block has the same degree (volume), where $\tD$ becomes a scalar matrix, the PPR vector is asymptotically equivalent to the optimal linear discriminant function.

%\subsection{Landing probability space} \label{ld_pre}
First, we briefly introduce linear discriminant analysis in landing probability space, which the PPR vector also lives in.
Consider a random walk on the graph starting from a seed node. 
Define the \textit{landing probability} $r_s^v$ to be the probability that the random walk ends up at $v\in V$ after exactly $s$ steps. The \textit{landing probability space} is the space of landing probability of any nodes.

A \textit{linear discriminant} (LD) analysis keeps the first $S$ landing probability on each node, $r^v=(r^v_0, r^v_1, ..., r^v_S)\in\R^S$, and divides vertices into two sets by thresholding on the linear discriminant score vector $l\in\R^N$, whose $v$-th entry is defined to be inner product
$$l_v = \langle\omega, r^v\rangle$$ 
with some weights $\omega\in \R^S$. 
For example, let $\omega= r^{v_1} - r^{v_2}$, where $v_1$, $v_2$ are empirical centroids of two node sets. 
Then $l_v$ increases as $v$ slides from $v_1$ to $v_2$, and thresholding $(\|v_1\|^2-\|v_2\|^2)/2$ allocates vertices to nearest centroid.

\textsc{Remark.}
The landing probability of the $s$-th step, $ r_s=\left(r^1_s, r^2_s, ..., r^N_s\right) \in\R^N$, is defined as $\left(P^s\right)\T\pi$. It follows from proposition \ref{ppr_sol} that PPR vector $p=\sum_{s=0}^{\infty} \phi_s r_s$ with $\phi_s=\alpha(1-\alpha)^s$. 
Keeping the first $S$ terms yields an LD score vector with the weights $\omega_{PPR}=(\phi_0, \phi_1, ..., \phi_{S-1})$. 

%\subsection{Population analysis}
We then perform population (expectation) analysis for PPR in the landing probability space. 
Define the population \textit{block landing probability} $\W^k_{s}$ to be the probability that a random walk from $v_0$ ends up in block $k$ after exactly $s$ steps, where $k=1,2,...,K$ and $s=0,1,...,S-1$. 
Given that $v_0$ is in block 1, $\W^\cdot_{0}=(1,0,...,0)\T $. 
Using the first $S$ steps block landing probabilities, the next lemma gives an explicit form of LD vectors.

\begin{lemma} [Explicit form of LD vectors] \label{lem:explicit3}
	\underdcsbm assume all blocks have the same degrees. Let $\cl(k)$ be the linear discriminant score vector between block 1 and block $k$. Then,  
	\begin{enumerate}[\normalfont(a)]
		\item $\W^\cdot_s=\tP\T\W^\cdot_{s-1}$,  $s=1,2,...,S-1$; and \label{transW}
		\item  $\cl(k)=\Theta Z\textbf{l}(k)$, $k=2,...,K$, where $\textbf{l}(k)=\W \W\T (e_1-e_k)$. \label{explicit3b}
	\end{enumerate} 
	Here, $e_k$ is the elementary unit vector on the direction of $k$-th block.
\end{lemma}
\begin{proof} 
	We prove \ref{transW} using following quantities. 
	Let $E^k_s$ be the number of paths from $v_0$ to block $k$ with exact length $s$, and let $\cE^k_s$ be the expected number of paths from $v_0$ to block $k$ with exact length $s$.
	Recall from \ref{sxn:pop} that $\B_{ij}$ represents the expected number of edges between block $i$ and $j$ if $i\neq j$, or twice of that if $i=j$.
	Then, 
	$$\cE^k_s=\sum_{j=1}^K \B_{kj}\cE^j_{s-1}.$$
	
	%Let $V_k$ be the set of nodes from block $k$, that is, $V_k=\{v\in V:z(v)=k\}$.
	To see $\W^\cdot_s=\tP\T\W^\cdot_{s-1}$, observe that 
	$$
	\W^k_s=\frac{\cE^k_s}{\sum_{i=1}^K\cE^i_s}\\
	= \frac{\sum_{j=1}^K\B_{kj}\cE^j_{s-1}}{\sum_{i=1}^K \sum_{j=1}^K\B_{ij} \cE^j_{s-1}}\\
	= \frac{\sum_{j=1}^K\B_{kj}\cE^j_{s-1}}{ \sum_{j=1}^K\td_{j} \cE^j_{s-1}}\\
	= \sum_{j=1}^K \tP_{kj}\W^k_{s-1}.
	$$
	The last equality comes from the assumption that all blocks have the same degrees, which means $\td_i$ is constant.
	
	Now, we prove part \ref{explicit3b} of the lemma. Let $R \in \R^{N\times S}$ collect all landing probabilities $r^v_s$ of the first $S$ steps, where $v=1,2,...,N$ and $s=0,1,...,S-1$. Without loss of generality, assume the seed node corresponds to the first row. 
	Define $\cR=\E(R)\in[0,1]^{N\times S}$ to be the population version of $R$. 
	Then the population landing probability is explicitly 
	$$\cR^v_s=\frac{\cd_v}{\td_{z(v)}}\W_s^{z(v)}=\theta_v\W^{z(v)}_s,$$
	or compactly, 
	$$\cR=\Theta Z \W.$$
	
	In linear discriminant, the weights vector $\omega$ is the geometric difference between centroid of block 1 and $k$, which can be written as 
	$$\left(\sum_{v:z(v)=1}\cR^v_1-\sum_{v:z(v)=k}\cR^v_1, \sum_{v:z(v)=1}\cR^v_2-\sum_{v:z(v)=k}\cR^v_2,..., \sum_{v:z(v)=1}\cR^v_S-\sum_{v:z(v)=k}\cR^v_S\right),$$
	or compactly
	$$\omega=\cR\T Z(e_1-e_k).$$
	By Lemma \ref{lem:basic}, the linear discriminant score vector reads
	\begin{eqnarray*}
		\langle\cR\cdot\omega\rangle&=&\cR\cR\T Z(e_1-e_k)\\
		&=&\Theta Z \W\W\T  (e_1-e_k),
	\end{eqnarray*}
	for $k=2,...,K$.
	Setting $\textbf{l}(k)=\W \W\T (e_1-e_k)$ completes the proof. 
\end{proof}

Recall from Theorem \ref{explicit2} that $\p=\Theta Z\tp$. The LD score vector $\cl$ has a similarly simple form that separates the block-related information ($\W$) and the node specific information ($\Theta$ and $Z$). 
Lemma \ref{lem:explicit3} provides a population (expectation) representation of PPR in the landing probability space. To facilitate its application in random graphs, the next lemma provides a control of the landing probabilities on a random block model graph.

\begin{lemma}[Concentration of landing probabilities]\label{lem:lp}
	Let $G=(V,E)$ be a graph of $N$ vertices generated from the DC-SBM with $K$ blocks and parameters $\{B, Z, \Theta\}$. Let $R^\cdot_s \in [0,1]^{N}$ be the landing probabilities of the $k$-th step, and $\cR^\cdot_s=\E (R)$ be its expectation. Then, for any $\epsilon>0$ and any vertex $u=1,2,...,N$,
	$$\bP\left(R^u_s\ge(1+\epsilon)\cR^u_s\right)\le (1+\epsilon)^{-\epsilon N r},$$
	$$\bP\left(R^u_s\ge(1-\epsilon)\cR^u_s\right)\le (1-\epsilon)^{\epsilon N r},$$
	where $r=\min_{v\in V} \theta_u\theta_v\tP_{z(u)z(v)}\W^{z(v)}_{s-1}$.
\end{lemma}
\begin{proof}
	Note that $R^u_s=\sum_{v\in V}X_{uv}$, where $$X_{uv}=\frac{\W^{z(v)}_{s-1}}{\td_{z(v)}}\1_{\{A_{uv}=1\}}$$
	are independent random variables having probability $\theta_u\theta_v\B_{z(u)z(v)}$ of being equal to ${\W^{z(v)}_{s-1}/\tD_{z(v)z(v)}}$. 
	Then, 
	$$\E\left[R^u_s\right]=\sum_{v\in V}\frac{\W^{z(v)}_{s-1}}{\td_{z(v)}}\theta_u\theta_v\B_{z(u)z(v)}=\theta_u\sum_{k=1}^K \tP_{z(u)k}\W^k_{s-1}=\cR^u_s.$$
	We can apply Chernoff's bounds (Lemma \ref{chernoff}) on $R^u_s$ and obtain bounds for any fixed $u$,
	$$\bP\left(R^u_s\ge(1+\epsilon)\cR^u_s\right)\le (1+\epsilon)^{-\epsilon \cR^u_s},$$
	and
	$$\bP\left(R^u_s\le(1-\epsilon)\cR^u_s\right)\le (1-\epsilon)^{\epsilon \cR^u_s}.$$
	Recognizing that $\cR^u_s \ge Nr$ completes the proof.
\end{proof}

Lemma \ref{lem:lp} provides an entrywise concentration bound for landing probabilities. 
The next theorem equates PPR and LD vectors when blocks are equally distributed. 
Together, they asserts the asymptotically equivalence between PPR and LD vectors, in symmetric block model graphs.

\begin{theorem} [Equivalence between PPR and LD vectors] \label{asymPPRLD}
	\underdcsbm assume $B_{ii}=b_1$ for all $i$, and $\B_{ij}=b_2$ for $i\neq j$ $(b_1>b_2>0)$. Let $\lambda_2$ the second largest eigenvalue of $\cP$. Let $\p$ be the personalized PageRank vector, and let $\cl(k)$ be the linear discriminant score vector between block 1 and block $k$, $k=2,...,K$. If the teleportation constant $\alpha=1-\lambda_2$, then
	$$\p \propto \cl(k).$$
\end{theorem}
\begin{proof}
	From Section \ref{sxn:sp_transition} and Lemma \ref{lem:explicit3}\ref{transW}, the block landing probability is precisely  
	$$\W^k_s=\sum_{j=1}^K \lambda_j^s U_{kj}U_{1j},$$ 
	where $\lambda_k$ is the $k$-th eigenvalues of $\cP$ and $U$ is the orthogonal matrix used in Lemma \ref{sxn:sp_transition}. 
	
	Note that $\B$ has eigenvalues $\lambda_1=1$ and $\lambda_2=\frac{b_1-b_2}{b_1+b_2}$, with complexity indices 1 and $k-1$ respectively. 
	In addition, we know the orthogonal matrix above precisely as well,
	$$U=\begin{bmatrix} 
	\frac{1}{\sqrt{N}} & \frac{1}{\sqrt{2}} & \frac{1}{\sqrt{2}} &\cdots&\frac{1}{\sqrt{2}} \\
	\frac{1}{\sqrt{N}} & -\frac{1}{\sqrt{2}} & 0 &\cdots&0 \\
	\frac{1}{\sqrt{N}} & 0& -\frac{1}{\sqrt{2}} &\cdots&0 \\
	\vdots & \vdots& \vdots& &\vdots \\
	\frac{1}{\sqrt{N}} & 0& 0 &\cdots&-\frac{1}{\sqrt{2}} 
	\end{bmatrix}.$$ 
	Then it follows from Lemma \ref{lem:explicit3}\ref{explicit3b} that the LD weight vector is 
	\begin{eqnarray*}
		\omega_{LD}&=&\cR\T Z(e_1-e_k)\\
		&=&\W\T (e_1-e_k)\\
		&=&\W^1_\cdot-\W^k_\cdot\\
		&=&\sum_{j=1}^K \lambda_j^s (U_{1j}-U_{kj})U_{1j}\\
		&=&\frac{K}{2}\begin{bmatrix} 
			1\\\lambda_2\\\lambda_2^2\\\vdots\\\lambda_2^{S-1}
		\end{bmatrix}.
	\end{eqnarray*}
	
	On the other hand, the weight vector of PPR on landing probability space is $\omega_{PPR}=(\phi_0, \phi_1, ...)$, where $\phi_s=\alpha(1-\alpha)^s$. 
	Hence, setting the teleportation constant $\alpha=1-\lambda_2$ asymptotically equates approximate PPR and LD vectors, up to a scalar factor.
\end{proof}

%	The proofs can be found in the Supplement to the paper. 
%	Here, two remarks are in order.
\textsc{Remark.}
First, a positive factor that differentiates PPR and LD vectors does not change the relative ranking of the nodes, because the ranking via $\p$ or $c\p$ is equivalent.
Hence, Theorem \ref{asymPPRLD} shows that the PPR vector is equivalent to an optimal LD score vector under described population DC-SBM. 
Second, Theorem \ref{asymPPRLD} is an extension of \citet{kloumann2017block}. 
Combining Theorem \ref{asymPPRLD} and Lemma \ref{lem:lp} gives the asymptotic equivalence between PPR and LD vectors under the particular DC-SBM stated. 
%	\subsection{Proof of Lemma \ref{lem:explicit3}}
%	\subsection{Proof of Theorem \ref{asymPPRLD}}
%	\subsection{Random landing probability}

\section{Lists of top 200 handles} \label{sxn:top200}
In this section, we supply three lists of handles resulting from sampling using PPR, aPPR, and rPPR vectors with @NBCPolitics as the seed, as of December 2018. We conceal handles with followers count fewer than 200 for privacy considerations. The biographical descriptions are trimmed for unifying displays. In addition, we annotate handles with whether or not they are followed (``Followed'') by the seed node.

\newpage
\subsection{A PPR's sample of 200}
\scriptsize
\centering
\begin{longtable}{rllrl}
	\captionsetup{width=.9\textwidth}
	\caption{Top (selected) handles returned by PPR.} \label{tab:ppr200} \\
	\hline\noalign{\vskip .7mm}  
	& Name & Followed & Followers & Description \\ 
	\hline\noalign{\vskip .7mm}  
	\endfirsthead
	\multicolumn{4}{@{}l}{\ldots continued}\\
	\hline\noalign{\vskip .7mm}  
	& Name & Followed & Followers & Description \\ 
	\hline\noalign{\vskip .7mm}  
	\endhead
	\hline
	%	\multicolumn{4}{r@{}}{continued \ldots}\\
	\endfoot
	\hline
	\endlastfoot
	1 & Melania Trump & Yes & 11242283 & This account is run by the Office of First Lady Melania Trum... \\ 
	2 & The White House & Yes & 17625630 & Welcome to @WhiteHouse! Follow for the latest from President... \\ 
	3 & Chuck Todd & Yes & 2032038 & Moderator of @meetthepress and @nbcnews political director; ... \\ 
	4 & NBC News & Yes & 6280551 & The leading source of global news and info for more than 75 ... \\ 
	5 & NBC Nightly News & Yes & 962290 & Breaking news, in-depth reporting, context on news from arou... \\ 
	6 & Andrea Mitchell & Yes & 1737764 & NBC News Chief Foreign Affairs Correspondent/anchor, Andrea ... \\ 
	7 & Savannah Guthrie & Yes & 881669 & Mom to Vale \& Charley, TODAY Co-Anchor, Georgetown Law... \\ 
	8 & Joe Scarborough & Yes & 2521215 & With Malice Toward None \\ 
	9 & MSNBC & Yes & 2261911 & The place for in-depth analysis, political commentary and in... \\ 
	10 & Rachel Maddow MSNBC & Yes & 9498076 & I see political people...
	(Retweets do not imply endorsement... \\ 
	11 & Breaking News & Yes & 9223158 &  \\ 
	12 & NBC News First Read & Yes & 53847 & The first place for news and analysis from the @NBCNews Poli... \\ 
	13 & TODAY & Yes & 4276453 & America's favorite morning show $|$ Snapchat: todayshow \\ 
	14 & Meet the Press & Yes & 566713 & Meet the Press is the longest-running television show in his... \\ 
	15 & The Wall Street Journal & Yes & 16188842 & Breaking news and features from the WSJ. \\ 
	16 & Pete Williams & Yes & 70062 & NBC News Justice Correspondent. Covers US Supreme Court, ... \\ 
	17 & Mark Murray & Yes & 97571 & Mark Murray is the senior political editor for NBC News, as ... \\ 
	18 & POLITICO & Yes & 3695835 & Nobody knows politics like POLITICO. Got a news tip for us? ... \\ 
	19 & Katy Tur & Yes & 587474 & MSNBC anchor @2pm, NBC News correspondent, author of NYT ... \\ 
	20 & Bill Clinton & Yes & 10697521 & Founder, Clinton Foundation and 42nd President of the United ... \\ 
	21 & Kasie Hunt & Yes & 381704 & @NBCNews Capitol Hill Correspondent. Host, @KasieDC, Sundays... \\ 
	22 & TIME & Yes & 15584815 & Breaking news and current events from around the globe. Host... \\ 
	23 & Kelly O'Donnell & Yes & 195765 & White House Correspondent @NBCNews Veteran of Cap Hill ... \\ 
	24 & John McCain & Yes & 3181773 & Memorial account for U.S. Senator John McCain, 1936-2018. To... \\ 
	25 & Peter Alexander & Yes & 283522 & @NBCNews White House Correspondent / Weekend @TODAYshow ... \\ 
	26 & Hallie Jackson & Yes & 359099 & Chief White House Correspondent / @NBCNews / @MSNBC ... \\ 
	27 & Kristen Welker & Yes & 182244 & @NBCNews White House Correspondent. Links and retweets ... \\ 
	28 & Carrie Dann & Yes & 37119 & .@NBCNews / @NBCPolitics. RTs not endorsements. \\ 
	29 & Willie Geist & Yes & 807536 & Host @NBC \#SundayTODAY, Co-Host @Morning\_Joe, “Sunday ... \\ 
	30 & Morning Joe & Yes & 563650 & Live tweet during the show! Links to must-read op-eds and ... \\ 
	31 & Frank Thorp V & Yes & 58152 & Producer \& Off-Air Reporter covering Congress at @NBCNews ... \\ 
	32 & Mark Knoller & Yes & 318923 & CBS News White House Correspondent \\ 
	33 & Tom Brokaw & Yes & 308276 & Special correspondent, @NBCNews \\ 
	34 & Mika Brzezinski & Yes & 868124 & "Bipartisanship helps to avoid extremes and imbalances. It ... \\ 
	35 & Chris Jansing & Yes & 72375 & @msnbc Senior National Correspondent, intrepid traveler and ... \\ 
	36 & John Harwood & Yes & 251246 & a Dad who covers Washington, the economy and national politi... \\ 
	37 & Nicolle Wallace & Yes & 413153 & Author of 18 Acres series, mom, dog walker, wife, gardener. ... \\ 
	38 & NBC News Signal & Yes & 83715 & A new streaming news channel from @NBCNews. Catch us Thursda... \\ 
	39 & Sam Stein & Yes & 392003 & Daily Beast/MSNBC  
	newsletter: https://t.co/DVURxntWdL Emai... \\ 
	40 & Chris Matthews & Yes & 882434 & Host of @hardball M-F at 7PM ET on @MSNBC and author of "Bob... \\ 
	41 & Carol Lee & Yes & 51240 & Reporter for NBC News, former WSJ \& POLITICO, Hudson's mom, ... \\ 
	42 & Ali Vitali & Yes & 78839 & @NBCnews Political Reporter. Covered Trump campaign, WH + ... \\ 
	43 & Ken Dilanian & Yes & 124635 & Intelligence and national security reporter for the NBC News... \\ 
	44 & Jim Miklaszewski & Yes & 14196 & Jim Miklaszewski is Chief Pentagon Correspondent for NBC New... \\ 
	45 & John Heilemann & Yes & 247616 & @SHO\_TheCircus host/ep; NBCNews/@MSNBC natl affairs analyst;... \\ 
	46 & Stephanie Ruhle & Yes & 352895 & Mom, MSNBC LIVE Anchor 9AM M-F, VELSHI \& RUHLE 1 PM  ... \\ 
	47 & Nick Confessore & Yes & 172359 & Reporter for @NYTimes, writer-at-large for @NYTmag, MSNBC ... \\ 
	48 & Talking Points Memo & Yes & 275692 & Breaking news and analysis from the TPM team. 
	
	“I’ll leave ... \\ 
	49 & Tom Costello & Yes & 17268 & NBC News Correspondent covering Aviation, Transportation, Ec... \\ 
	50 & Post Politics & Yes & 384611 & The latest political news and analysis from The Washington P... \\ 
	51 & Alex Moe & Yes & 28245 & @NBCNews Capitol Hill Producer + Off-Air Reporter; '12 \& '16... \\ 
	52 & Benjy Sarlin & Yes & 100896 & Political reporter for @NBCNews. I cover elections and their... \\ 
	53 & Preet Bharara & Yes & 945030 & Patriotic American \& proud immigrant. Movie buff. @Springste... \\ 
	54 & Matthew Miller & Yes & 229867 & Partner at Vianovo. MSNBC Justice \& Security Analyst. Recove... \\ 
	55 & Leigh Ann Caldwell & Yes & 20714 & NBC Capitol Hill reporter. Formerly at CNN and public radio.... \\ 
	56 & Ken Strickland & Yes & 2693 & NBC News Washington Bureau Chief \\ 
	57 & Ron Fournier & Yes & 64356 & President: Truscott Rossman. Best-seller https://t.co/09CdTN... \\ 
	58 & Mike Memoli & Yes & 39693 & National Political Reporter @nbcnews; @latimes alum mike dot... \\ 
	59 & Miguel Almaguer & Yes & 14082 & Prolific coffee drinker. Chronic under sleeper. Raging road ... \\ 
	60 & Courtney Kube & Yes & 9494 & NBC News National Security \& Military Reporter. Links and ... \\ 
	61 & NBC News World & Yes & 279165 & A dynamic look at world events from @NBCNews. \\ 
	62 & Jonathan Martin & Yes & 241690 & Nat'l Political Correspondent, NY Times. Husband of the ... \\ 
	63 & Steve Schmidt & Yes & 498812 & "Patriotism means to stand by the country. It does not mean ... \\ 
	64 & Jenna Bush Hager & Yes & 207106 & Mama to M and P, NBC News correspondent, Editor-at-Large ... \\ 
	65 & Sean Spicer & Yes & 406957 & President of RigWil, Sr Advisor @AmericaFirstPAC  check out ... \\ 
	66 & Roll Call & Yes & 356374 & Breaking news, reporter tweets and analysis from the Source ... \\ 
	67 & POLITICO 45 & Yes & 88470 & A daily diary of the 45th president of the United States. \\ 
	68 & Scott Foster & Yes & 3464 & Senior Producer, Washington @NBCNEWS @TODAYshow \\ 
	69 & Domenico Montanaro & Yes & 83999 & "Congress shall make no law respecting an est. of religion, ... \\ 
	70 & Tom Winter & Yes & 40777 & NBC News Investigations reporter based in New York focusing ... \\ 
	71 & Kailani Koenig & Yes & 11416 & Producer with @MSNBC \& @NBCNews. Team @MeetThePress ... \\ 
	72 & Capital Journal & Yes & 131212 & WSJ’s home for politics, policy and national security news. ... \\ 
	73 & NBC News Videos & Yes & 7838 & The latest video from http://t.co/xPyvMOTEF6 \\ 
	74 & Diane Sawyer & Yes & 876906 & I like my news 24/7, my food spicy, my drinks caffeinated, ... \\ 
	75 & Jane C. Timm & Yes & 6478 & @nbcnews political reporter and fact checker. More fun than ... \\ 
	76 & Elyse PG & Yes & 2697 & White House producer @nbcnews $|$@USCAnnenberg alum $|$ LA kid ... \\ 
	77 & Libby Leist & Yes & 7946 & Executive Producer @todayshow \\ 
	78 & Mike Barnicle & Yes & 116588 & Mike Barnicle is an award-winning print and broadcast journa... \\ 
	79 & Reuters Politics & Yes & 259106 & U.S. political coverage, breaking news and special investiga... \\ 
	80 & Beth Fouhy & Yes & 13684 & Senior editor, politics, NBC News and MSNBC \\ 
	81 & HuffPost & Yes & 11401771 & Know what's real. \\ 
	82 & Joey Scarborough & Yes & 6277 & NBC News Social Media Editor. New York Daily News Alum. RTs ... \\ 
	83 & Marianna Sotomayor & Yes & 11965 & Running around Capitol Hill for @NBCNews. Covers politics ... \\ 
	84 & Shaquille Brewster & Yes & 5362 & @NBCNews Producer/Politics $|$ @HowardU Alum$|$ Journalist ... \\ 
	85 & Joyce Alene & Yes & 185116 & U of Alabama Law Professor$|$@MSNBC Contributor$|$Obama US ... \\ 
	86 & Garrett Haake & Yes & 40714 & Correspondent @msnbc • Taller than I look on TV • Long-suffe... \\ 
	87 & Andrew Rafferty & Yes & 16567 & Senior political editor for @newsy Before that @NBCNews ... \\ 
	88 & Jacob Soboroff & Yes & 144153 & @MSNBC correspondent. Instagram \& Snapchat: jacobsoboroff \\ 
	89 & Perry Bacon Jr. & Yes & 26853 & I write about government (mostly federal, often state, ... \\ 
	90 & Alex Witt & Yes & 28126 & Weekend host on @MSNBC (9am, noon \& 1pm). Tigger’s mom ... \\ 
	91 & Mark Halperin & Yes & 332564 & New York, New York \\ 
	92 & Heidi Przybyla & Yes & 66489 & NBC News, n'tl political reporter "Prezbella"    Heidi.Przyb... \\ 
	93 & Morgan Radford & Yes & 20967 & @NBCnews Correspondent: @TODAYShow/@NBCNightlyNews .  \\ 
	94 & Savannah Sellers & Yes & 4637 & News junkie. Host of NBC's "Stay Tuned" on Snapchat. Storyte... \\ 
	95 & Marist Poll & Yes & 16030 & Founded in 1978, MIPO is home to the Marist Poll and regular... \\ 
	96 & Jill Wine-Banks & Yes & 158753 & @NBCNews \& @MSNBC Contributor. Speaker. Watergate prosecutor... \\ 
	97 & NBC Field Notes & Yes & 1390 & NBC News correspondents and http://t.co/1eSopOQt8s reporters... \\ 
	98 & Olivia Nuzzi & Yes & 190919 & Washington Correspondent, New York Magazine \\ 
	99 & NBC News THINK & Yes & 12017 & THINK is NBC News' home for fresh opinion, sharp analysis ... \\ 
	100 & Making a Difference & Yes & 670 & @NBCNightlyNews' popular feature profiles ordinary people do... \\ 
	101 & adam nagourney & Yes & 25307 & LA Bureau Chief for The New York Times. Story ideas welcome ... \\ 
	102 & Phil McCausland & Yes & 2519 & @NBCNews Digital reporter focused on the rural-urban divide.... \\ 
	103 & Katie Couric & Yes & 1746116 & Journalist, podcaster, @SU2C founder, doc filmmaker of @FedU... \\ 
	104 & Monica Alba & Yes & 30034 & @NBCNews White House team. Covered Hillary Clinton on the ... \\ 
	105 & Vicente Fox Quesada & Yes & 1244017 & Presidente de México de 2000 a 2006 y ahora trabajando po... \\ 
	106 & Alex Johnson & Yes & 4371 & News, data and analysis for @NBCNews; data geek; non-celebri... \\ 
	108 & Alex Seitz-Wald & Yes & 50168 & Political reporter for @NBCNews covering Democrats $|$ Tips, ... \\ 
	109 & Anthony Terrell & Yes & 6827 & Emmy Award winning journalist. Political observer. Covered ... \\ 
	110 & Sam Petulla & Yes & 2588 & Editor @cnnpolitics • Usually looking for datasets. You can ... \\ 
	111 & Debra Messing & Yes & 532941 & Actor. Mama. Global Ambassador for HIV/AIDS for PSI. Activis... \\ 
	112 & Corky Siemaszko & Yes & 2538 & Senior Writer at NBC News Digital (former NY Daily News rewr... \\ 
	114 & Zach Haberman & Yes & 3693 & Lead Breaking News Editor, @NBCNews. Previously had other jobs ... \\ 
	115 & NBC Latino & Yes & 67920 & Elevating the conversation around Latino news in the United ... \\ 
	116 & Vivian Salama & Yes & 16020 & White House reporter for @WSJ. Formerly AP Baghdad bureau ... \\ 
	117 & Zeke Miller & Yes & 215054 & White House Reporter @AP. Email: zekejmiller@gmail.com Links... \\ 
	118 & Vaughn Hillyard & Yes & 31464 & On the Road, Meeting Good Folk $|$ NBC News $|$ Arizonan $|$ IG: @... \\ 
	119 & Jonathan Allen & Yes & 44477 & political reporter, @NBCNews Digital $|$ co-author, NYT bestse... \\ 
	121 & HuffPost Politics & Yes & 1428870 & The latest political news from HuffPost's politics team. \\ 
	122 & Nick Akerman & Yes & 14949 & Partner in the AmLaw 100 law firm of Dorsey \& Whitney, Water... \\ 
	123 & CSPAN & Yes & 1915821 & Capitol Hill. The White House. National Politics. \\ 
	124 & John McCormack & Yes & 30688 & Senior writer at The Weekly Standard. \\ 
	125 & Jo Ling Kent & Yes & 32957 & NBC News Correspondent @NBCNightlyNews, @TODAYshow ... \\ 
	126 & PolitiFact & Yes & 628659 & Home of the Truth-O-Meter and independent fact-checking ... \\ 
	127 & Bob Corker & Yes & 10042 & Serving Tennesseans in the U.S. Senate \\ 
	128 & Elise Jordan & Yes & 58884 & Co-host of @WMM\_podcast podcast. @MSNBC/@NBCNews political... \\ 
	129 & Greg Martin & Yes & 1161 & Political Booking Producer at @nbcnews @todayshow \\ 
	130 & Education Nation & Yes & 276468 & Hosted by @NBCNews. Creator of Parent Toolkit \& moderator of... \\ 
	131 & Micah Grimes & Yes & 25948 & Head of Social, @NBCNews \& @MSNBC -- Foreign and domestic ... \\ 
	132 & Jill Lawrence & Yes & 17282 & Commentary editor and columnist @USATODAY. Author of The Art... \\ 
	133 & McKay Coppins & Yes & 131623 & Staff writer at @TheAtlantic. Author of THE WILDERNESS. 'Sor... \\ 
	134 & Emmanuelle Saliba & Yes & 4004 & Head of Social Media Strategy @Euronews $|$ Launched \#THECUBE ... \\ 
	135 & Hasani Gittens & Yes & 3002 & Level 29 Mage. Senior News Ed. @NBCNews. Sheriff of Nattahna... \\ 
	136 & Rebecca Sinderbrand & Yes & 18691 & Now: @NBCNews Senior Washington Editor, visiting lecturer @Y... \\ 
	137 & BuzzFeed Politics & Yes & 121646 & News and updates from the politics team @BuzzFeedNews. \\ 
	138 & Adam Edelman & Yes & 2341 & Political reporter @nbcnews. Wisconsin native, Bestchester r... \\ 
	139 & Ethan Klapper & Yes & 18292 & Journalist (@YahooNews) and \#avgeek. \\ 
	140 & President Trump & No & 24593638 & 45th President of the United States of America, @realDonaldT... \\ 
	141 & Vice President Mike ... & No & 6795022 & Vice President Mike Pence. Husband, father, \& honored to ... \\ 
	142 & Donald J. Trump & No & 56050499 & 45th President of the United States of America \\ 
	143 & Karen Pence & No & 403315 & Educator, mom, wife of @VP Pence. Passionate about art thera... \\ 
	144 & Sarah Sanders & No & 3522219 & @WhiteHouse Press Secretary. Proudly representing @POTUS ... \\ 
	145 & Kellyanne Conway & No & 2506546 & Mom. Patriot. Catholic. Counselor. \\ 
	146 & DRUDGE REPORT & No & 1408129 & The DRUDGE REPORT is a U.S. based news aggregation website ... \\ 
	147 & White House History & No & 104010 & The White House Historical Association is a non-profit organ... \\ 
	148 & The New York Times & No & 42412491 & Where the conversation begins. Follow for breaking news, ... \\ 
	149 & White House Archived & No & 13379715 & This is an archive of an Obama Administration account mainta... \\ 
	150 & Dan Scavino Jr. & No & 324561 & Assistant to President @realDonaldTrump, Director of Social ... \\ 
	151 & Drudge Buzz & No & 104111 & Tracking the buzz made by Americas \#1 newsmaker Matt Drudge.... \\ 
	152 & David Gregory & No & 1749373 & CNN, Georgetown U \\ 
	153 & Hillary Clinton & No & 23643522 & 2016 Democratic Nominee, SecState, Senator, hair icon. Mom, ... \\ 
	154 & CNN Breaking News & No & 54476034 & Breaking news from CNN Digital. Now 54M strong. Check @cnn ... \\ 
	155 & The Cabinet & No & 123597 & The @WhiteHouse Office of Cabinet Affairs. Tweets may be arc... \\ 
	156 & Lester Holt & No & 501427 & Anchor @NBCNightlyNews and @datelinenbc, reporting on the to... \\ 
	157 & John Dickerson & No & 48122 & Co-host CBS This Morning. This account @johndickerson is mos... \\ 
	158 & CNN & No & 40854429 & It’s our job to \#GoThere \& tell the most difficult stories. ... \\ 
	159 & J Earnest (Archived) & No & 1182091 & WH Press Secretary. This is an archive of an Obama Administr... \\ 
	160 & The Washington Post & No & 13117609 & Breaking news, analysis, and opinion. Founded in 1877. Our ... \\ 
	161 & Adam Liptak & No & 61589 & Supreme Court reporter for The New York Times \\ 
	162 & NSC & No & 35905 & National Security Council $|$ Tweets may be archived ... \\ 
	163 & MSNBC video & No & 40669 & Favorite video highlights from @msnbc. \\ 
	164 & Gorsuch Facts & No & 39143 & Judge Gorsuch will be fair to all regardless of their backgr... \\ 
	165 & Greg Stohr & No & 11651 & Supreme Court reporter for Bloomberg News. Baseball dad ... \\ 
	166 & OMB Press & No & 11182 & Office of Management and Budget $|$ Tweets may be archived: ... \\ 
	167 & Richard Engel & No & 288066 & @NBCNews Chief Foreign Correspondent \\ 
	168 & Norah O'Donnell & No & 195549 & Wife, mother of 3, Co-Host @cbsthismorning, \#1 fan of @chefg... \\ 
	169 & Robert Barnes & No & 37361 & Robert Barnes covers the Supreme Court for The Washington Po... \\ 
	170 & Luke Russert & No & 253495 & Sometimes nothing can be a real cool hand. STA'04/BC'08 \\ 
	171 & Stephen Colbert & No & 18269222 & the guy on CBS \\ 
	172 & Mark Sherman & No & 6336 &  \\ 
	173 & U.S. Attorney EDVA & No & 5709 & Led by U.S. Attorney G. Zachary Terwilliger. 130+ attorneys ... \\ 
	174 & The Associated Press & No & 13051963 & News from The Associated Press, and a taste of the great jou... \\ 
	175 & Joe Palazzolo & No & 10938 & WSJ reporter covering legal issues.  joe.palazzolo@wsj.com. ... \\ 
	176 & Natalie Morales & No & 443991 & @TODAYshow Anchor and @AccessOnline Anchor, Author, mom ... \\ 
	177 & Brent Kendall & No & 5451 & WSJ legal affairs reporter in Washington. Native Tar Heel, ... \\ 
	178 & Joan Biskupic & No & 11021 & CNN legal analyst \& Supreme Court biographer; Chicago native... \\ 
	179 & Keith Olbermann & No & 1097676 & Dogs. And sports. And whales (Tom Jumbo-Grumbo on BoJack ... \\ 
	180 & Brian Williams & No & 230947 &  \\ 
	181 & Pope Francis & No & 17791867 & Welcome to the official Twitter page of His Holiness Pope Fr... \\ 
	182 & Ezra Klein & No & 2500383 & Founder and editor-at-large, https://t.co/5gESirESRH. Why ... \\ 
	183 & Anderson Cooper & No & 9967099 & tweets by Anderson Cooper. Anchor @AC360 and correspondent... \\ 
	184 & BBC News (World) & No & 24153838 & News, features and analysis from the World's newsroom. Break... \\ 
	185 & Reince Priebus & No & 935431 & President @MichaelBestLaw; Exclusive Speaker @WashSpeakers; ... \\ 
	186 & Joe Biden & No & 3111675 & Represented Delaware in the Senate for 36 years, 47th Vice P... \\ 
	187 & Department of State & No & 5149607 & Welcome to the official U.S. Department of State Twitter acc... \\ 
	188 & Jim Miklaszewski & No & 1956 & Chief Pentagon Correspondent for NBC News \\ 
	189 & Tony Mauro & No & 20310 & Supreme Court correspondent, https://t.co/571ZdQnzo2 and The... \\ 
	190 & David Axelrod & No & 1113850 & Director, UChicago Institute of Politics. Senior Political ... \\ 
	191 & Nate Silver & No & 3176243 & Editor-in-Chief, @FiveThirtyEight. Author, The Signal and ... \\ 
	192 & George Bush & No & 356042 & A tribute site to the 41st President of the United States of... \\ 
	193 & CBS News & No & 6537991 & Your source for original reporting and trusted news. \\ 
	194 & Jonathan Karl & No & 206986 & ABC News Chief White House Correspondent. insta @jonkarl ... \\ 
	195 & BBC Breaking News & No & 38539186 & Breaking news alerts and updates from the BBC. For news, ... \\ 
	196 & Mitt Romney & No & 1977201 & Senator-elect from Utah. \\ 
	197 & ABC News & No & 13985606 & All the news and information you need to see, curated by the... \\ 
	198 & Deborah Turness & No & 10389 & President of NBC News International \\ 
	199 & The Hill & No & 3162118 & The Hill is the premier source for policy and political news... \\ 
	200 & Ann Curry & No & 1536122 & Journalism is an act of faith in the future. \\ 
	\hline
\end{longtable}
\normalsize

\newpage
\subsection{An aPPR's sample of 200}
\scriptsize
\centering
\begin{longtable}{rllrl}
	\captionsetup{width=.9\textwidth}
	\caption{Top (selected) handles returned by aPPR. The handles with fewer than 200 followers are hidden for privacy considerations.} \label{tab:appr200} \\
	\hline\noalign{\vskip .7mm}  
	& Name & Followed & Followers & Description \\ 
	\hline\noalign{\vskip .7mm}  
	\endfirsthead
	\multicolumn{4}{@{}l}{\ldots continued}\\
	\hline\noalign{\vskip .7mm}  
	& Name & Followed & Followers & Description \\ 
	\hline\noalign{\vskip .7mm}  
	\endhead
	\hline
	%	\multicolumn{4}{r@{}}{continued \ldots}\\
	\endfoot
	\hline
	\endlastfoot
	1 &  & Yes & 198 & Enroll America National Regional Director
	http://t.co/X6jJIE... \\ 
	2 & Jennifer Sizemore & Yes & 386 &  \\ 
	3 & Alissa Swango & Yes & 441 & Director of Digital Programming at @natgeo. All things food.... \\ 
	4 & Making a Difference & Yes & 670 & @NBCNightlyNews' popular feature profiles ordinary people do... \\ 
	5 &  & No &   1 &  \\ 
	6 &  & No &   3 &  \\ 
	7 & Greg Martin & Yes & 1161 & Political Booking Producer at @nbcnews @todayshow \\ 
	8 &  & No &   1 & I am Area Man. I pwn your news feed. \\ 
	9 &  & No &   2 &  \\ 
	10 & NBC Field Notes & Yes & 1390 & NBC News correspondents and http://t.co/1eSopOQt8s reporters... \\ 
	11 &  & No &   2 &  \\ 
	12 &  & No &   2 &  \\ 
	13 &  & No &   1 &  \\ 
	14 &  & No &   1 &  \\ 
	15 &  & No &   1 &  \\ 
	16 &  & No &   1 &  \\ 
	17 &  & No &   3 & yet another activist twitter, fighting all those fun -isms ... \\ 
	18 &  & No &   4 &  \\ 
	19 &  & No &   7 & Dianne Kube is an Author with a passion, for family, holiday... \\ 
	20 &  & No &   7 &  \\ 
	21 & Adam Edelman & Yes & 2341 & Political reporter @nbcnews. Wisconsin native, Bestchester ... \\ 
	22 & Phil McCausland & Yes & 2519 & @NBCNews Digital reporter focused on the rural-urban divide.... \\ 
	23 & Corky Siemaszko & Yes & 2538 & Senior Writer at NBC News Digital (former NY Daily News rewr... \\ 
	24 & Sam Petulla & Yes & 2588 & Editor @cnnpolitics • Usually looking for datasets. You can ... \\ 
	25 & Ken Strickland & Yes & 2693 & NBC News Washington Bureau Chief \\ 
	26 &  & No &   7 &  \\ 
	27 & Elyse PG & Yes & 2697 & White House producer @nbcnews $|$@USCAnnenberg alum $|$ LA kid ... \\ 
	28 &  & No &   2 & Change your thoughts \& you change your world. -Norman Vincen... \\ 
	29 &  & No &   4 &  \\ 
	30 &  & No &  13 &  \\ 
	31 &  & No &   6 &  \\ 
	32 &  & No & 154 & We distribute new, never-worn clothing and merchandise to ... \\ 
	33 &  & No &  10 &  \\ 
	34 & Hasani Gittens & Yes & 3002 & Level 29 Mage. Senior News Ed. @NBCNews. Sheriff of Nattahna... \\ 
	35 &  & No &   1 &  \\ 
	36 & Scott Foster & Yes & 3464 & Senior Producer, Washington @NBCNEWS @TODAYshow \\ 
	37 &  & No &   2 &  \\ 
	38 &  & No &  13 &  \\ 
	39 &  & No &   5 &  \\ 
	40 & Zach Haberman & Yes & 3693 & Lead Breaking News Editor, @NBCNews. Previously had other jobs... \\ 
	41 &  & No &   3 & just like to stay in the know :) 
	
	just like to stay in the ... \\ 
	42 &  & No &   2 &  \\ 
	43 &  & No &   5 &  \\ 
	44 &  & No &   7 &  \\ 
	45 &  & No &   1 &  \\ 
	46 & Emmanuelle Saliba & Yes & 4004 & Head of Social Media Strategy @Euronews $|$ Launched \#THECUBE ... \\ 
	47 &  & No &   2 &  \\ 
	48 & Alex Johnson & Yes & 4371 & News, data and analysis for @NBCNews; data geek; non-celebri... \\ 
	49 &  & No &   8 &  \\ 
	50 & Savannah Sellers & Yes & 4637 & News junkie. Host of NBC's "Stay Tuned" on Snapchat. Storyte... \\ 
	51 &  & No &  21 &  \\ 
	52 &  & No &   6 & Anti-money laundering professional with federal law enforcem... \\ 
	53 &  & No &  15 &  \\ 
	54 & Shaquille Brewster & Yes & 5362 & @NBCNews Producer/Politics $|$ @HowardU Alum$|$ Journalist $|$ Pol... \\ 
	55 &  & No &   2 & Just another DIY, punk kid from the black land dirt of NEPA'... \\ 
	56 &  & No &  18 & Cdr Bob Mehal, Public Affairs Office, Office of the Secretar... \\ 
	57 &  & No &   5 &  \\ 
	58 &  & No &   4 &  \\ 
	59 &  & No &   8 &  \\ 
	60 &  & No &  10 &  \\ 
	61 &  & No &   2 &  \\ 
	62 & Joey Scarborough & Yes & 6277 & NBC News Social Media Editor. New York Daily News Alum. RTs ... \\ 
	63 &  & No &   5 &  \\ 
	64 &  & No &   1 &  \\ 
	65 & Voices United & No & 310 & Voices United is a non profit educational organization for ... \\ 
	66 & Jane C. Timm & Yes & 6478 & @nbcnews political reporter and fact checker. More fun than ... \\ 
	67 & Social Headlines & No & 344 & Daily roundup of top social media and networking stories. \\ 
	68 & James Miklaszewski & No & 337 & Writer, Photographer, Editor, Director, Producer, Newshound ... \\ 
	69 &  & No &  12 &  \\ 
	70 & Anthony Terrell & Yes & 6827 & Emmy Award winning journalist. Political observer. Covered ... \\ 
	71 &  & No &  10 &  \\ 
	72 &  & No &   8 &  \\ 
	73 &  & No &   8 & I'm the real Charlie Sheen. If you are a Winner, stick aroun... \\ 
	74 &  & No &   9 & Quotes from a nice jewish mom who's just tryna get some nice... \\ 
	75 &  & No &   2 &  \\ 
	76 &  & No &   4 &  \\ 
	77 &  & No &   6 & "Rawr!" \\ 
	78 & NBC News Videos & Yes & 7838 & The latest video from http://t.co/xPyvMOTEF6 \\ 
	79 &  & No &   9 &  \\ 
	80 &  & No &   4 &  \\ 
	81 & Libby Leist & Yes & 7946 & Executive Producer @todayshow \\ 
	82 &  & No &   8 &  \\ 
	83 &  & No &   2 & I'm running for President of the United States of America. \\ 
	84 &  & No &  35 &  \\ 
	85 &  & No &   8 &  \\ 
	86 &  & No &   2 &  \\ 
	87 &  & No &   2 &  \\ 
	88 &  & No &  16 &  \\ 
	89 &  & No &   4 &  \\ 
	90 &  & No &   5 & Happy princess \\ 
	91 &  & No &   1 &  \\ 
	92 &  & No &   4 &  \\ 
	93 & Courtney Kube & Yes & 9494 & NBC News National Security \& Military Reporter. Links and ... \\ 
	94 &  & No &   5 &  \\ 
	95 &  & No &   5 &  \\ 
	96 &  & No & 169 &  \\ 
	97 &  & No &   5 &  \\ 
	98 &  & No &   2 &  \\ 
	99 & Vets Helping Heroes & No & 449 & Raising funds to sponsor the training of assistance dogs for... \\ 
	100 &  & No &  12 &  \\ 
	101 &  & No &   4 &  \\ 
	102 &  & No &   8 &  \\ 
	103 & Bob Corker & Yes & 10042 & Serving Tennesseans in the U.S. Senate \\ 
	104 &  & No &   4 &  \\ 
	105 &  & No &   2 &  \\ 
	106 &  & No &  11 & Spécialiste développement produit et marketing des produits ... \\ 
	107 &  & No &   4 &  \\ 
	108 &  & No &   8 & Not your average Grandma \\ 
	109 &  & No &  29 &  \\ 
	110 &  & No &   2 &  \\ 
	111 &  & No &   6 &  \\ 
	112 & Kailani Koenig & Yes & 11416 & Producer with @MSNBC \& @NBCNews. Team @MeetThePress alum. 20... \\ 
	113 &  & No &  13 &  \\ 
	114 &  & No &  14 &  \\ 
	115 & Gloria Turkin & No & 204 & I am honest and straight to the point.  Retired Civilian Fed... \\ 
	116 &  & No &   7 &  \\ 
	117 &  & No &  28 & An unconventional appreciation account for @DeadlineWH host,... \\ 
	118 &  & No &   6 &  \\ 
	119 &  & No &  10 & Live like Bones \\ 
	120 &  & No &   2 &  \\ 
	121 & Marianna Sotomayor & Yes & 11965 & Running around Capitol Hill for @NBCNews. Covers politics ... \\ 
	122 & NBC News THINK & Yes & 12017 & THINK is NBC News' home for fresh opinion, sharp analysis ... \\ 
	123 &  & No &   1 &  \\ 
	124 &  & No &  15 &  \\ 
	125 &  & No &   2 &  \\ 
	126 &  & No &   3 & Photographer, artist, newsletter editor, designer, writer... \\ 
	127 &  & No &  18 &  \\ 
	128 &  & No &   5 &  \\ 
	129 &  & No &   5 & The Quest for the Denim Jacket \\ 
	130 &  & No &   9 &  \\ 
	131 &  & No &  15 &  \\ 
	132 &  & No &  16 & Author of A Traumatic History: A Unique Look at PTSD and ... \\ 
	133 &  & No &   5 &  \\ 
	134 &  & No &   7 &  \\ 
	135 &  & No &   5 &  \\ 
	136 &  & No &   7 &  \\ 
	137 &  & No &   7 &  \\ 
	138 & Beth Fouhy & Yes & 13684 & Senior editor, politics, NBC News and MSNBC \\ 
	139 & Jim Miklaszewski & Yes & 14196 & Jim Miklaszewski is Chief Pentagon Correspondent for NBC New... \\ 
	140 & Miguel Almaguer & Yes & 14082 & Prolific coffee drinker. Chronic under sleeper. Raging road ... \\ 
	141 &  & No &  16 &  \\ 
	142 &  & No &   4 &  \\ 
	143 &  & No &   3 &  \\ 
	144 &  & No &  19 & The Northeast Tennessee Victory program will create a grassr... \\ 
	145 &  & No &  17 &  \\ 
	146 &  & No &  14 & Just a dude with a crappy job. \\ 
	147 &  & No &   5 &  \\ 
	148 & Nick Akerman & Yes & 14949 & Partner in the AmLaw 100 law firm of Dorsey \& Whitney, Water... \\ 
	149 &  & No &   5 &  \\ 
	150 &  & No &  59 &  \\ 
	151 &  & No &   8 &  \\ 
	152 &  & No &   8 &  \\ 
	153 &  & No &   4 & Grad student at JHU \\ 
	154 &  & No &   6 &  \\ 
	155 & Marist Poll & Yes & 16030 & Founded in 1978, MIPO is home to the Marist Poll and regular... \\ 
	156 &  & No &  10 & Sharing the best news from the e-Discovery world. Tweets by ... \\ 
	157 &  & No &   7 &  \\ 
	158 &  & No &   4 &  \\ 
	159 &  & No &  11 & Workforce and Economic Development Consultant; Employment ... \\ 
	160 &  & No &   7 & We're the workers of the @villagevoice, trying to get a fair... \\ 
	161 & Vivian Salama & Yes & 16020 & White House reporter for @WSJ. Formerly AP Baghdad bureau ... \\ 
	162 &  & No &   8 &  \\ 
	163 &  & No &  24 &  \\ 
	164 &  & No &  19 & I should be the real trix rabbit \\ 
	165 &  & No &   4 &  \\ 
	166 &  & No &  24 & Curious food lover always looking for the best food everywhe... \\ 
	167 & Andrew Rafferty & Yes & 16567 & Senior political editor for @newsy Before that @NBCNews. And... \\ 
	168 &  & No &   5 &  \\ 
	169 &  & No &  36 &  \\ 
	170 & Tom Costello & Yes & 17268 & NBC News Correspondent covering Aviation, Transportation, ... \\ 
	171 &  & No &  68 & Wanderlust journalist  ... A man is but the product of ... \\ 
	172 &  & No &   6 & Bibliophile, Animal lover, Realtor, Volunteer, \\ 
	173 &  & No &  25 &  \\ 
	174 &  & No &  70 & Director or Product Marketing @ Microsoft. 
	My tweets. My li... \\ 
	175 &  & No &   3 & Experienced (and successful) grantwriter, author, wife, moth... \\ 
	176 &  & No &   5 &  \\ 
	177 & Jill Lawrence & Yes & 17282 & Commentary editor and columnist @USATODAY. Author of The Art... \\ 
	178 &  & No &   8 & Howard McKinnon is Town Manager of Havana, Florida. \\ 
	179 &  & No & 136 &  \\ 
	180 &  & No &  59 &  \\ 
	181 &  & No &   8 &  \\ 
	182 &  & No &  12 &  \\ 
	183 &  & No &   7 &  \\ 
	184 &  & No &   8 &  \\ 
	185 &  & No &   8 &  \\ 
	186 &  & No &  15 & Old and getting older. \\ 
	187 &  & No &  15 &  \\ 
	188 &  & No &  15 & Married \\ 
	189 &  & No &   4 &  \\ 
	190 &  & No &   2 & Director of the Essex, Connecticut Public Library aka "Your ... \\ 
	191 & Ethan Klapper & Yes & 18292 & Journalist (@YahooNews) and \#avgeek. \\ 
	192 &  & No &  38 &  \\ 
	193 &  & No &   5 &  \\ 
	194 & Rebecca Sinderbrand & Yes & 18691 & Now: @NBCNews Senior Washington Editor, visiting lecturer ... \\ 
	195 &  & No &   3 &  \\ 
	196 &  & No &  11 & Tireless trend researcher. \\ 
	197 &  & No &   5 &  \\ 
	198 &  & No &   5 &  \\ 
	199 &  & No &  11 &  \\ 
	200 &  & No &   3 &  \\ 
	\hline
\end{longtable}
\normalsize

\newpage
\subsection{An rPPR's sample of 200}
\scriptsize
\centering
\begin{longtable}{rllrl}
	\captionsetup{width=.9\textwidth}
	\caption{Top (selected) handles returned by rPPR. The handles with fewer than 200 followers are hidden for privacy considerations.} \label{tab:rppr200} \\
	\hline\noalign{\vskip .7mm}  
	& Name & Followed & Followers & Description \\ 
	\hline\noalign{\vskip .7mm}  
	\endfirsthead
	\multicolumn{4}{@{}l}{\ldots continued}\\
	\hline\noalign{\vskip .7mm}  
	& Name & Followed & Followers & Description \\ 
	\hline\noalign{\vskip .7mm}  
	\endhead
	\hline
	%	\multicolumn{4}{r@{}}{continued \ldots}\\
	\endfoot
	\hline
	\endlastfoot
	1 &  & Yes & 198 & Enroll America National Regional Director
	http://t.co/X6jJIE... \\ 
	2 & Jennifer Sizemore & Yes & 386 &  \\ 
	3 & Alissa Swango & Yes & 441 & Director of Digital Programming at @natgeo. All things food.... \\ 
	4 & Making a Difference & Yes & 670 & @NBCNightlyNews' popular feature profiles ordinary people do... \\ 
	5 & Greg Martin & Yes & 1161 & Political Booking Producer at @nbcnews @todayshow \\ 
	6 & NBC Field Notes & Yes & 1390 & NBC News correspondents and http://t.co/1eSopOQt8s reporters... \\ 
	7 & Adam Edelman & Yes & 2341 & Political reporter @nbcnews. Wisconsin native, Bestchester ... \\ 
	8 & Phil McCausland & Yes & 2519 & @NBCNews Digital reporter focused on the rural-urban divide.... \\ 
	9 & Corky Siemaszko & Yes & 2538 & Senior Writer at NBC News Digital (former NY Daily News ... \\ 
	10 & Sam Petulla & Yes & 2588 & Editor @cnnpolitics • Usually looking for datasets. You can ... \\ 
	11 & Ken Strickland & Yes & 2693 & NBC News Washington Bureau Chief \\ 
	12 & Elyse PG & Yes & 2697 & White House producer @nbcnews $|$@USCAnnenberg alum $|$ LA kid ... \\ 
	13 & Hasani Gittens & Yes & 3002 & Level 29 Mage. Senior News Ed. @NBCNews. Sheriff of Nattahna... \\ 
	14 & Scott Foster & Yes & 3464 & Senior Producer, Washington @NBCNEWS @TODAYshow \\ 
	15 & Zach Haberman & Yes & 3693 & Lead Breaking News Editor, @NBCNews. Previously had other jobs ... \\ 
	16 & Emmanuelle Saliba & Yes & 4004 & Head of Social Media Strategy @Euronews $|$ Launched \#THECUBE ... \\ 
	17 & Alex Johnson & Yes & 4371 & News, data and analysis for @NBCNews; data geek; non-celebri... \\ 
	18 & Savannah Sellers & Yes & 4637 & News junkie. Host of NBC's "Stay Tuned" on Snapchat. Storyte... \\ 
	19 &  & No & 154 & We distribute new, never-worn clothing and merchandise to ... \\ 
	20 & Shaquille Brewster & Yes & 5362 & @NBCNews Producer/Politics $|$ @HowardU Alum$|$ Journalist $|$ Pol... \\ 
	21 & Joey Scarborough & Yes & 6277 & NBC News Social Media Editor. New York Daily News Alum. RTs ... \\ 
	22 & Jane C. Timm & Yes & 6478 & @nbcnews political reporter and fact checker. More fun than ... \\ 
	23 & Anthony Terrell & Yes & 6827 & Emmy Award winning journalist. Political observer. Covered ... \\ 
	24 & NBC News Videos & Yes & 7838 & The latest video from http://t.co/xPyvMOTEF6 \\ 
	25 & Libby Leist & Yes & 7946 & Executive Producer @todayshow \\ 
	26 & Voices United & No & 310 & Voices United is a non profit educational organization for ... \\ 
	27 & Social Headlines & No & 344 & Daily roundup of top social media and networking stories. \\ 
	28 & James Miklaszewski & No & 337 & Writer, Photographer, Editor, Director, Producer, Newshound ... \\ 
	29 & Courtney Kube & Yes & 9494 & NBC News National Security \& Military Reporter. Links and ... \\ 
	30 & Bob Corker & Yes & 10042 & Serving Tennesseans in the U.S. Senate \\ 
	31 & Kailani Koenig & Yes & 11416 & Producer with @MSNBC \& @NBCNews. Team @MeetThePress alum... \\ 
	32 & Vets Helping Heroes & No & 449 & Raising funds to sponsor the training of assistance dogs for... \\ 
	33 & Marianna Sotomayor & Yes & 11965 & Running around Capitol Hill for @NBCNews. Covers politics ... \\ 
	34 & NBC News THINK & Yes & 12017 & THINK is NBC News' home for fresh opinion, sharp analysis ... \\ 
	35 & Beth Fouhy & Yes & 13684 & Senior editor, politics, NBC News and MSNBC \\ 
	36 & Jim Miklaszewski & Yes & 14196 & Jim Miklaszewski is Chief Pentagon Correspondent for NBC New... \\ 
	37 & Miguel Almaguer & Yes & 14082 & Prolific coffee drinker. Chronic under sleeper. Raging road ... \\ 
	38 &  & No & 169 &  \\ 
	39 & Nick Akerman & Yes & 14949 & Partner in the AmLaw 100 law firm of Dorsey \& Whitney, Water... \\ 
	40 & Marist Poll & Yes & 16030 & Founded in 1978, MIPO is home to the Marist Poll and regular... \\ 
	41 & Vivian Salama & Yes & 16020 & White House reporter for @WSJ. Formerly AP Baghdad bureau ... \\ 
	42 & Andrew Rafferty & Yes & 16567 & Senior political editor for @newsy Before that @NBCNews. And... \\ 
	43 & Tom Costello & Yes & 17268 & NBC News Correspondent covering Aviation, Transportation, ... \\ 
	44 & Gloria Turkin & No & 204 & I am honest and straight to the point.  Retired Civilian Fed... \\ 
	45 & Jill Lawrence & Yes & 17282 & Commentary editor and columnist @USATODAY. Author of The Art... \\ 
	46 & Ethan Klapper & Yes & 18292 & Journalist (@YahooNews) and \#avgeek. \\ 
	47 & Rebecca Sinderbrand & Yes & 18691 & Now: @NBCNews Senior Washington Editor, visiting lecturer ... \\ 
	48 & Leigh Ann Caldwell & Yes & 20714 & NBC Capitol Hill reporter. Formerly at CNN and public radio.... \\ 
	49 & Morgan Radford & Yes & 20967 & @NBCnews Correspondent: @TODAYShow/@NBCNightlyNews/... \\ 
	50 & GuardAnglSolPet & No & 927 & Supporting the Military, our Veterans and their Beloved Pets... \\ 
	51 & adam nagourney & Yes & 25307 & LA Bureau Chief for The New York Times. Story ideas welcome ... \\ 
	52 &  & No &  13 &  \\ 
	53 & Micah Grimes & Yes & 25948 & Head of Social, @NBCNews \& @MSNBC -- Foreign and domestic ... \\ 
	54 & Perry Bacon Jr. & Yes & 26853 & I write about government (mostly federal, often state, occas... \\ 
	55 &  & No &  21 &  \\ 
	56 & Alex Moe & Yes & 28245 & @NBCNews Capitol Hill Producer + Off-Air Reporter; '12 \& '16... \\ 
	57 & Ray Farmer & No & 603 & NBC News staff photographer. Colorado based \\ 
	58 & Alex Witt & Yes & 28126 & Weekend host on @MSNBC (9am, noon \& 1pm). Tigger’s mom + ... \\ 
	59 & Monica Alba & Yes & 30034 & @NBCNews White House team. Covered Hillary Clinton on the ... \\ 
	60 & Jim Miklaszewski & No & 1956 & Chief Pentagon Correspondent for NBC News \\ 
	61 &  & No &  13 &  \\ 
	62 & John McCormack & Yes & 30688 & Senior writer at The Weekly Standard. \\ 
	63 &  & No & 136 &  \\ 
	64 & Vaughn Hillyard & Yes & 31464 & On the Road, Meeting Good Folk $|$ NBC News $|$ Arizonan $|$ IG... \\ 
	65 &  & No &  35 &  \\ 
	66 & Madelyn Monteath & No & 257 & NFL Network, wife, mother.. not necessarily in that order. \\ 
	67 & Thomas DeFrank & No & 593 & Veteran White House correspondent (every prez since LBJ) and... \\ 
	68 & Jo Ling Kent & Yes & 32957 & NBC News Correspondent @NBCNightlyNews, @TODAYshow... \\ 
	69 &  & No &  10 &  \\ 
	70 & Carrie Dann & Yes & 37119 & .@NBCNews / @NBCPolitics. RTs not endorsements. \\ 
	71 &  & No &   3 &  \\ 
	72 &  & No &   7 & Dianne Kube is an Author with a passion, for family, holiday... \\ 
	73 &  & No &  18 & Cdr Bob Mehal, Public Affairs Office, Office of the Secretar... \\ 
	74 &  & No &   7 &  \\ 
	75 & Mike Memoli & Yes & 39693 & National Political Reporter @nbcnews; @latimes alum mike dot... \\ 
	76 & John Boxley & No & 1201 & NBC News Producer...Living life one day at a time. \\ 
	77 &  & No &  15 &  \\ 
	78 & Tom Winter & Yes & 40777 & NBC News Investigations reporter based in New York focusing ... \\ 
	79 &  & No &   7 &  \\ 
	80 & Garrett Haake & Yes & 40714 & Correspondent @msnbc • Taller than I look on TV • Long-suffe... \\ 
	81 &  & No &  59 &  \\ 
	82 &  & No &  70 & Director or Product Marketing @ Microsoft. 
	My tweets. My li... \\ 
	83 &  & No &  68 & Wanderlust journalist  ... A man is but the product of ... \\ 
	84 &  & No & 158 & Marketing nerd at Cornerstone OnDemand. \\ 
	85 & Jonathan Allen & Yes & 44477 & political reporter, @NBCNews Digital $|$ co-author, NYT bestse... \\ 
	86 & NBC News First Read & Yes & 53847 & The first place for news and analysis from the @NBCNews Poli... \\ 
	87 &  & No &  92 & Smokin Meat \& Raising Kids That Raise Hell. 
	Live Every Day ... \\ 
	88 & Sam Singal & No & 1016 & Executive Producer, @nbcnightlynews \\ 
	89 &  & No &  29 &  \\ 
	90 &  & No &  59 &  \\ 
	91 & Carol Lee & Yes & 51240 & Reporter for NBC News, former WSJ \& POLITICO, Hudson's mom, ... \\ 
	92 & Alex Seitz-Wald & Yes & 50168 & Political reporter for @NBCNews covering Democrats $|$ Tips, ... \\ 
	93 &  & No &  28 & An unconventional appreciation account for @DeadlineWH host,... \\ 
	94 &  & No & 188 & I am a Senior Video Producer at NBCNews.com, as well a few ... \\ 
	95 & HailYeah63 & No & 483 & \#RedskinsTweetTeam \#HTTR \\ 
	96 & Eva's Heroes & No & 2067 & To enrich the lives of individuals with intellectual special... \\ 
	97 &  & No &   6 &  \\ 
	98 & Chi Omega & No & 278 & Chi Omega Chapter at CU Boulder \\ 
	99 & Aarne Heikkila & No & 1210 & Coordinating Producer for @JacobSoboroff @MSNBC \& @NBCNews, ... \\ 
	100 & Dani & No & 447 & only here to talk shit \& complain \\ 
	101 & Frank Thorp V & Yes & 58152 & Producer \& Off-Air Reporter covering Congress at @NBCNews. ... \\ 
	102 & Youcef & No & 228 &  \\ 
	103 &  & No &  76 & Pentagon correspondent http://t.co/Qo0w3AnYOb \\ 
	104 & project c.u.r.e. & No & 2260 & delivering donated medical supplies and equipment to develop... \\ 
	105 &  & No & 117 &  \\ 
	106 &  & No &   4 &  \\ 
	107 & Elise Jordan & Yes & 58884 & Co-host of @WMM\_podcast podcast. @MSNBC/@NBCNews political ... \\ 
	108 & Patrick Burkey & No & 2313 & Executive Producer, @NBCNews, @MSNBC. Former EP, @NBCNightly... \\ 
	109 & bill hartnett & No & 2500 & Stripmining the internets for remarkable ephemera Social Mus... \\ 
	110 &  & No &   7 &  \\ 
	111 &  & No &   8 &  \\ 
	112 &  & No &  16 &  \\ 
	113 &  & No &  36 &  \\ 
	114 & Ron Fournier & Yes & 64356 & President: Truscott Rossman. Best-seller https://t.co/09CdTN... \\ 
	115 &  & No &  12 &  \\ 
	116 & Pete Williams & Yes & 70062 & NBC News Justice Correspondent. Covers US Supreme Court, ... \\ 
	117 &  & No &  65 & Wife, Mother. Litigation Specialist. Designer. Activist for ... \\ 
	118 &  & No &  10 &  \\ 
	119 & Heidi Przybyla & Yes & 66489 & NBC News, n'tl political reporter "Prezbella"    Heidi.Przyb... \\ 
	120 & NBC Latino & Yes & 67920 & Elevating the conversation around Latino news in the United ... \\ 
	121 &  & No & 189 &  \\ 
	122 &  & No &  38 &  \\ 
	123 & Chris Jansing & Yes & 72375 & @msnbc Senior National Correspondent, intrepid traveler and ... \\ 
	124 &  & No &   1 &  \\ 
	125 & Brent Kendall & No & 5451 & WSJ legal affairs reporter in Washington. Native Tar Heel, ... \\ 
	126 &  & No &   2 &  \\ 
	127 & U.S. Attorney EDVA & No & 5709 & Led by U.S. Attorney G. Zachary Terwilliger. 130+ attorneys ... \\ 
	128 &  & No &  74 & Life long learner Paralegal Arts \& Culture Black Community ... \\ 
	129 &  & No &   2 &  \\ 
	130 &  & No &   2 &  \\ 
	131 & Tammy Fine & No & 1584 & Corporate Communications by day. Teen Negotiator by night... \\ 
	132 & Bonnie Optekman & No & 2242 & Digital media strategist. Voice over artist. 
	News junkie, ... \\ 
	133 &  & No &   3 & yet another activist twitter, fighting all those fun -isms ... \\ 
	134 &  & No & 109 & Communicator through an eclectic lens of \#healthcare \#hospit... \\ 
	135 &  & No &  88 & Earth and Physical Science Teacher, Mom of 2, Self-declared ... \\ 
	136 & Amy Lynn-Cramer & No & 1590 & Mommy to 2 amazing kiddos, Wife to @tecramer AND Corporate ... \\ 
	137 &  & No &   5 &  \\ 
	138 & prodjay & No & 304 & NBC News producer \\ 
	139 &  & No & 109 &  \\ 
	140 &  & No &  10 &  \\ 
	141 & Meghann Ludemann & No & 216 & Stay Tuned Associate Producer @NBCNews on @Snapchat \\ 
	142 &  & No &   4 &  \\ 
	143 & Ali Vitali & Yes & 78839 & @NBCnews Political Reporter. Covered Trump campaign, WH + ... \\ 
	144 & Doug Adams & No & 1902 & NBC Sr. Political desk editor; Father; Baseball fan; Lover ... \\ 
	145 &  & No &  99 &  \\ 
	146 & Mark Sherman & No & 6336 &  \\ 
	147 & Robin Gradison & No & 272 & NBC News DC Deputy Bureau Chief, politics junkie, road run... \\ 
	148 & NBC News Signal & Yes & 83715 & A new streaming news channel from @NBCNews. Catch us Thursda... \\ 
	149 &  & No &  45 & Professor at Columbia Journaism School. \\ 
	150 &  & No &   8 &  \\ 
	151 &  & No &  18 &  \\ 
	152 & Stacey Klein & No & 914 & @NBCNews White House Producer, Born and raised in BalDimore ... \\ 
	153 &  & No &  97 &  \\ 
	154 & Rich Latour & No & 1883 & From Broadcast News to Digital Storytelling. Dad of 3 Boys ... \\ 
	155 & Domenico Montanaro & Yes & 83999 & "Congress shall make no law respecting an est. of religion, ... \\ 
	156 &  & No &   5 &  \\ 
	157 &  & No &   6 & Anti-money laundering professional with federal law enforcem... \\ 
	158 &  & No &  24 &  \\ 
	159 &  & No &  24 & Curious food lover always looking for the best food everywhe... \\ 
	160 &  & No &  25 &  \\ 
	161 &  & No & 161 & 1 of 12 U.S.-led PRTs. Improving Panjshir's stability, incre... \\ 
	162 & Anna Matthews & No & 230 &  \\ 
	163 &  & No &  46 &  \\ 
	164 & POLITICO 45 & Yes & 88470 & A daily diary of the 45th president of the United States. \\ 
	165 &  & No &   9 & Quotes from a nice jewish mom who's just tryna get some nice... \\ 
	166 &  & No &  19 & The Northeast Tennessee Victory program will create a grassr... \\ 
	167 &  & No & 130 & @NBCNews Producer in London, Links \& retweets aren't endorse... \\ 
	168 & samgo & No & 1161 & Executive Producer, @MSNBC Digital \\ 
	169 & Megan Stark & No & 263 & over served Coloradan \\ 
	170 &  & No &  70 &  \\ 
	171 & Katie Yu & No & 484 & NBC News Senior Producer / formerly @Nightline, @NBCNightlyN... \\ 
	172 & Mark Murray & Yes & 97571 & Mark Murray is the senior political editor for NBC News, as ... \\ 
	173 & Kevin Thurm & No & 1946 & Chief Executive Officer @ClintonFdn. Dad, sports fan \& trivi... \\ 
	174 &  & No & 122 & Mom, wife, grandma, Airedale Terrier lover \\ 
	175 &  & No & 173 & Providing conservatives with breaking news, opinion, blogs ... \\ 
	176 &  & No &  14 &  \\ 
	177 &  & No &  12 &  \\ 
	178 &  & No & 137 & Celebrate the simple loveliness of every day things, scarves... \\ 
	179 &  & No &  15 &  \\ 
	180 &  & No &   8 &  \\ 
	181 &  & No &  16 & Author of A Traumatic History: A Unique Look at PTSD and ... \\ 
	182 &  & No &   9 &  \\ 
	183 &  & No &   8 & I'm the real Charlie Sheen. If you are a Winner, stick aroun... \\ 
	184 & David Espo & No & 1308 & Dad, AP Special Correspondent, Dad, Red Sox fan, Dad. \\ 
	185 &  & No &  40 &  \\ 
	186 & matt toder & No & 253 & supervising producer, documentaries/verticals at NBC News ... \\ 
	187 &  & No &  13 &  \\ 
	188 & Benjy Sarlin & Yes & 100896 & Political reporter for @NBCNews. I cover elections and their... \\ 
	189 &  & No &  15 &  \\ 
	190 &  & No &  29 &  \\ 
	191 &  & No &  17 &  \\ 
	192 &  & No &  28 & Director of the Marist Poll, poll obsessed,  epistemophilic,... \\ 
	193 &  & No &  16 &  \\ 
	194 &  & No & 144 & Vice President, Standards @NBCNews \\ 
	195 &  & No & 108 & trey.daly@gmail.com \\ 
	196 & Daniella Mayer & No & 314 & DON'T forget the A. I think everything about North Korea is ... \\ 
	197 & Bill Hatfield & No & 635 & Washington news producer for NBC News TODAY; politics/histor... \\ 
	198 &  & No &  19 & I should be the real trix rabbit \\ 
	199 &  & No &  50 &  \\ 
	200 & Phil Griffin & No & 231 &  \\ 
\end{longtable}
\normalsize

\newpage
\bibliographystyle{plainnat}
\bibliography{pprsbm}

\begin{thebibliography}{39}
\providecommand{\natexlab}[1]{#1}
\providecommand{\url}[1]{\texttt{#1}}
\expandafter\ifx\csname urlstyle\endcsname\relax
  \providecommand{\doi}[1]{doi: #1}\else
  \providecommand{\doi}{doi: \begingroup \urlstyle{rm}\Url}\fi

\bibitem[Airoldi et~al.(2008)Airoldi, Blei, Fienberg, and
  Xing]{airoldi2008mixed}
Edoardo~M Airoldi, David~M Blei, Stephen~E Fienberg, and Eric~P Xing.
\newblock Mixed membership stochastic blockmodels.
\newblock \emph{Journal of Machine Learning Research}, 9\penalty0
  (Sep):\penalty0 1981--2014, 2008.

\bibitem[Alamgir and von Luxburg(2010)]{alamgir2010multi}
M.~Alamgir and U.~von Luxburg.
\newblock Multi-agent random walks for local clustering.
\newblock pages 18--27, Piscataway, NJ, USA, December 2010.
  Max-Planck-Gesellschaft, IEEE.

\bibitem[Andersen and Lang(2006)]{andersen2006communities}
Reid Andersen and Kevin~J. Lang.
\newblock Communities from seed sets.
\newblock In \emph{Proceedings of the 15th International Conference on World
  Wide Web}, WWW '06, pages 223--232, New York, NY, USA, 2006. ACM.
\newblock ISBN 1-59593-323-9.
\newblock \doi{10.1145/1135777.1135814}.
\newblock URL \url{http://doi.acm.org/10.1145/1135777.1135814}.

\bibitem[Andersen and Peres(2009)]{andersen2009finding}
Reid Andersen and Yuval Peres.
\newblock Finding sparse cuts locally using evolving sets.
\newblock In \emph{Proceedings of the Forty-first Annual ACM Symposium on
  Theory of Computing}, STOC '09, pages 235--244, New York, NY, USA, 2009. ACM.
\newblock ISBN 978-1-60558-506-2.
\newblock \doi{10.1145/1536414.1536449}.
\newblock URL \url{http://doi.acm.org/10.1145/1536414.1536449}.

\bibitem[Andersen et~al.(2006)Andersen, Chung, and Lang]{andersen2006local}
Reid Andersen, Fan Chung, and Kevin Lang.
\newblock Local graph partitioning using pagerank vectors.
\newblock In \emph{Proceedings of the 47th Annual IEEE Symposium on Foundations
  of Computer Science}, FOCS '06, pages 475--486, Washington, DC, USA, 2006.
  IEEE Computer Society.
\newblock ISBN 0-7695-2720-5.
\newblock \doi{10.1109/FOCS.2006.44}.
\newblock URL \url{http://dx.doi.org/10.1109/FOCS.2006.44}.

\bibitem[Barab{\'a}si and Albert(1999)]{barabasi1999emergence}
Albert-L{\'a}szl{\'o} Barab{\'a}si and R{\'e}ka Albert.
\newblock Emergence of scaling in random networks.
\newblock \emph{science}, 286\penalty0 (5439):\penalty0 509--512, 1999.

\bibitem[Berkhin(2006)]{berkhin2006bookmark}
Pavel Berkhin.
\newblock Bookmark-coloring algorithm for personalized pagerank computing.
\newblock \emph{Internet Mathematics}, 3\penalty0 (1):\penalty0 41--62, 2006.

\bibitem[Boucheron et~al.(2013)Boucheron, Lugosi, and
  Massart]{boucheron2013concentration}
St{\'e}phane Boucheron, G{\'a}bor Lugosi, and Pascal Massart.
\newblock \emph{Concentration inequalities: A nonasymptotic theory of
  independence}.
\newblock Oxford university press, 2013.

\bibitem[Br{\'e}maud(2013)]{bremaud2013markov}
Pierre Br{\'e}maud.
\newblock \emph{Markov chains: Gibbs fields, Monte Carlo simulation, and
  queues}, volume~31.
\newblock Springer Science \& Business Media, 2013.

\bibitem[Brin and Page(1998)]{Brin1998The}
Sergey Brin and Lawrence Page.
\newblock The anatomy of a large-scale hypertextual web search engine.
\newblock \emph{Comput. Netw. ISDN Syst.}, 30\penalty0 (1-7):\penalty0
  107--117, April 1998.
\newblock ISSN 0169-7552.
\newblock \doi{10.1016/S0169-7552(98)00110-X}.
\newblock URL \url{http://dx.doi.org/10.1016/S0169-7552(98)00110-X}.

\bibitem[Chen et~al.(2019)Chen, Fan, Ma, and Wang]{chen2017spectral}
Yuxin Chen, Jianqing Fan, Cong Ma, and Kaizheng Wang.
\newblock Spectral method and regularized mle are both optimal for top-$k$
  ranking.
\newblock \emph{Ann. Statist.}, 47\penalty0 (4):\penalty0 2204--2235, 08 2019.
\newblock \doi{10.1214/18-AOS1745}.
\newblock URL \url{https://doi.org/10.1214/18-AOS1745}.

\bibitem[Davis and Kahan(1970)]{davis1970rotation}
Chandler Davis and William~Morton Kahan.
\newblock The rotation of eigenvectors by a perturbation. iii.
\newblock \emph{SIAM Journal on Numerical Analysis}, 7\penalty0 (1):\penalty0
  1--46, 1970.

\bibitem[Frobenius et~al.(1912)Frobenius, Frobenius, Frobenius, Frobenius, and
  Mathematician]{frobenius1912matrizen}
Georg~Ferdinand Frobenius, Ferdinand~Georg Frobenius, Ferdinand~Georg
  Frobenius, Ferdinand~Georg Frobenius, and Germany Mathematician.
\newblock \emph{{\"U}ber Matrizen aus nicht negativen Elementen}.
\newblock K{\"o}nigliche Akademie der Wissenschaften, 1912.

\bibitem[Gharan and Trevisan(2012)]{gharan2012approximating}
Shayan~Oveis Gharan and Luca Trevisan.
\newblock Approximating the expansion profile and almost optimal local graph
  clustering.
\newblock In \emph{Proceedings of the 2012 IEEE 53rd Annual Symposium on
  Foundations of Computer Science}, FOCS '12, pages 187--196, Washington, DC,
  USA, 2012. IEEE Computer Society.
\newblock ISBN 978-0-7695-4874-6.
\newblock \doi{10.1109/FOCS.2012.85}.
\newblock URL \url{https://doi.org/10.1109/FOCS.2012.85}.

\bibitem[Ghoshal and Barab{\'a}si(2011)]{ghoshal2011ranking}
Gourab Ghoshal and Albert-L{\'a}szl{\'o} Barab{\'a}si.
\newblock Ranking stability and super-stable nodes in complex networks.
\newblock \emph{Nature communications}, 2:\penalty0 394, 2011.

\bibitem[Gleich(2015)]{gleich2015pagerank}
David~F Gleich.
\newblock Pagerank beyond the web.
\newblock \emph{SIAM Review}, 57\penalty0 (3):\penalty0 321--363, 2015.

\bibitem[Gupta et~al.(2013)Gupta, Goel, Lin, Sharma, Wang, and
  Zadeh]{gupta2013wtf}
Pankaj Gupta, Ashish Goel, Jimmy Lin, Aneesh Sharma, Dong Wang, and Reza Zadeh.
\newblock Wtf: The who to follow service at twitter.
\newblock In \emph{Proceedings of the 22nd international conference on World
  Wide Web}, pages 505--514. ACM, 2013.

\bibitem[Haveliwala(2003)]{haveliwala2003topic}
Taher~H Haveliwala.
\newblock Topic-sensitive pagerank: A context-sensitive ranking algorithm for
  web search.
\newblock \emph{IEEE transactions on knowledge and data engineering},
  15\penalty0 (4):\penalty0 784--796, 2003.

\bibitem[Holland et~al.(1983)Holland, Laskey, and
  Leinhardt]{holland1983stochastic}
Paul~W Holland, Kathryn~Blackmond Laskey, and Samuel Leinhardt.
\newblock Stochastic blockmodels: First steps.
\newblock \emph{Social networks}, 5\penalty0 (2):\penalty0 109--137, 1983.

\bibitem[Jeh and Widom(2003)]{jeh2003scaling}
Glen Jeh and Jennifer Widom.
\newblock Scaling personalized web search.
\newblock In \emph{Proceedings of the 12th International Conference on World
  Wide Web}, WWW '03, pages 271--279, New York, NY, USA, 2003. ACM.
\newblock ISBN 1-58113-680-3.
\newblock \doi{10.1145/775152.775191}.
\newblock URL \url{http://doi.acm.org/10.1145/775152.775191}.

\bibitem[Karrer and Newman(2011)]{karrer2011stochastic}
Brian Karrer and Mark~EJ Newman.
\newblock Stochastic blockmodels and community structure in networks.
\newblock \emph{Physical Review E}, 83\penalty0 (1):\penalty0 016107, 2011.

\bibitem[Karypis and Kumar(1998)]{karypis1998multilevelk}
George Karypis and Vipin Kumar.
\newblock Multilevelk-way partitioning scheme for irregular graphs.
\newblock \emph{Journal of Parallel and Distributed computing}, 48\penalty0
  (1):\penalty0 96--129, 1998.

\bibitem[Khanna et~al.(2017)Khanna, Elenberg, Dimakis, and
  Negahban]{khanna2017approximation}
Rajiv Khanna, Ethan Elenberg, Alexandros~G Dimakis, and Sahand Negahban.
\newblock On approximation guarantees for greedy low rank optimization.
\newblock \emph{arXiv preprint arXiv:1703.02721}, 2017.

\bibitem[Kloumann et~al.(2017)Kloumann, Ugander, and
  Kleinberg]{kloumann2017block}
Isabel~M. Kloumann, Johan Ugander, and Jon Kleinberg.
\newblock Block models and personalized pagerank.
\newblock \emph{Proceedings of the National Academy of Sciences}, 114\penalty0
  (1):\penalty0 33--38, 2017.
\newblock ISSN 0027-8424.
\newblock \doi{10.1073/pnas.1611275114}.
\newblock URL \url{https://www.pnas.org/content/114/1/33}.

\bibitem[Le et~al.(2016)Le, Levina, Vershynin, et~al.]{le2016optimization}
Can~M Le, Elizaveta Levina, Roman Vershynin, et~al.
\newblock Optimization via low-rank approximation for community detection in
  networks.
\newblock \emph{The Annals of Statistics}, 44\penalty0 (1):\penalty0 373--400,
  2016.

\bibitem[Liao et~al.(2009)Liao, Lu, Baym, Singh, and Berger]{liao2009isorankn}
Chung-Shou Liao, Kanghao Lu, Michael Baym, Rohit Singh, and Bonnie Berger.
\newblock {IsoRankN: spectral methods for global alignment of multiple protein
  networks}.
\newblock \emph{Bioinformatics}, 25\penalty0 (12):\penalty0 i253--i258, 05
  2009.
\newblock ISSN 1367-4803.
\newblock \doi{10.1093/bioinformatics/btp203}.
\newblock URL \url{https://doi.org/10.1093/bioinformatics/btp203}.

\bibitem[Lov\'asz(1996)]{lovasz1993random}
L\'aszl\'o Lov\'asz.
\newblock Random walks on graphs: A survey.
\newblock In D.~{Mikl\'os}, V.~T. {S\'os}, and T.~{Sz\H{o}nyi}, editors,
  \emph{Combinatorics, Paul Erd\H{o}s is Eighty}, volume~2, pages 353--398.
  J\'anos Bolyai Mathematical Society, Budapest, 1996.

\bibitem[Lu et~al.(2013)Lu, Malmros, Liljeros, Britton,
  et~al.]{lu2013respondent}
Xin Lu, Jens Malmros, Fredrik Liljeros, Tom Britton, et~al.
\newblock Respondent-driven sampling on directed networks.
\newblock \emph{Electronic Journal of Statistics}, 7:\penalty0 292--322, 2013.

\bibitem[Macropol et~al.(2009)Macropol, Can, and Singh]{macropol2009rrw}
Kathy Macropol, Tolga Can, and Ambuj~K. Singh.
\newblock Rrw: repeated random walks on genome-scale protein networks for local
  cluster discovery.
\newblock \emph{BMC Bioinformatics}, 10\penalty0 (1):\penalty0 283, Sep 2009.
\newblock ISSN 1471-2105.
\newblock \doi{10.1186/1471-2105-10-283}.
\newblock URL \url{https://doi.org/10.1186/1471-2105-10-283}.

\bibitem[Newman et~al.(2006)Newman, Barabasi, and Watts]{mark2006structure}
Mark Newman, Albert-Laszlo Barabasi, and Duncan~J. Watts, editors.
\newblock \emph{The Structure and Dynamics of Networks}.
\newblock Princeton University Press, Princeton, NJ, USA, 2006.

\bibitem[Page et~al.(1998)Page, Brin, Motwani, and Winograd]{page1999pagerank}
L.~Page, S.~Brin, R.~Motwani, and T.~Winograd.
\newblock The pagerank citation ranking: Bringing order to the web.
\newblock In \emph{Proceedings of the 7th International World Wide Web
  Conference}, pages 161--172, Brisbane, Australia, 1998.
\newblock URL \url{citeseer.nj.nec.com/page98pagerank.html}.

\bibitem[Perron(1907)]{perron1907theorie}
Oskar Perron.
\newblock Zur theorie der matrices.
\newblock \emph{Mathematische Annalen}, 64\penalty0 (2):\penalty0 248--263,
  1907.

\bibitem[Qin and Rohe(2013)]{qin2013regularized}
Tai Qin and Karl Rohe.
\newblock Regularized spectral clustering under the degree-corrected stochastic
  blockmodel.
\newblock In \emph{Proceedings of the 26th International Conference on Neural
  Information Processing Systems - Volume 2}, NIPS'13, pages 3120--3128, USA,
  2013. Curran Associates Inc.
\newblock URL \url{http://dl.acm.org/citation.cfm?id=2999792.2999960}.

\bibitem[Rohe et~al.(2011)Rohe, Chatterjee, and Yu]{rohe2011spectral}
Karl Rohe, Sourav Chatterjee, and Bin Yu.
\newblock Spectral clustering and the high-dimensional stochastic blockmodel.
\newblock \emph{The Annals of Statistics}, 39\penalty0 (4):\penalty0
  1878--1915, 2011.

\bibitem[Sengupta and Chen(2018)]{sengupta2018block}
Srijan Sengupta and Yuguo Chen.
\newblock A block model for node popularity in networks with community
  structure.
\newblock \emph{Journal of the Royal Statistical Society: Series B (Statistical
  Methodology)}, 80\penalty0 (2):\penalty0 365--386, 2018.

\bibitem[Spielman and Teng(1996)]{spielman1996spectral}
Daniel~A Spielman and Shang-Hua Teng.
\newblock Spectral partitioning works: Planar graphs and finite element meshes.
\newblock In \emph{Foundations of Computer Science, 1996. Proceedings., 37th
  Annual Symposium on}, pages 96--105. IEEE, 1996.

\bibitem[Spielman and Teng(2004)]{spielman2004nearly}
Daniel~A Spielman and Shang~Hua Teng.
\newblock Nearly-linear time algorithms for graph partitioning, graph
  sparsification, and solving linear systems.
\newblock In \emph{Proceedings of the thirty-sixth annual ACM symposium on
  Theory of computing}, pages 81--90. ACM, 2004.

\bibitem[Watts and Strogatz(1998)]{watts1998collective}
Duncan~J Watts and Steven~H Strogatz.
\newblock Collective dynamics of ‘small-world’networks.
\newblock \emph{nature}, 393\penalty0 (6684):\penalty0 440, 1998.

\bibitem[Zhu et~al.(2013)Zhu, Yan, and Moore]{zhu2013oriented}
Yaojia Zhu, Xiaoran Yan, and Cristopher Moore.
\newblock Oriented and degree-generated block models: generating and inferring
  communities with inhomogeneous degree distributions.
\newblock \emph{Journal of Complex Networks}, 2\penalty0 (1):\penalty0 1--18,
  2013.

\end{thebibliography}

\end{document}